\documentclass[conference,compsoc]{IEEEtran}

\def\noeditingmarks{} 

\newcommand{\G}{\mathbb{G}}

\newcommand{\Adv}{\mathcal{A}}

\newcommand{\PRF}{\textrm{PRF}}
\newcommand{\PRG}{\textrm{PRG}}

\newcommand{\GenGraph}{\textsc{GenGraph}}
\newcommand{\FindNeighbors}{\textsc{FindNeighbors}}
\newcommand{\ChooseSet}{\textsc{ChooseSet}}

\newcommand{\AsymEnc}{\textrm{AsymEnc}}
\newcommand{\AsymDec}{\textrm{AsymDec}}
\newcommand{\AsymGen}{\textrm{AsymGen}}

\newcommand{\SymGen}{\textrm{SymAuthGen}}
\newcommand{\SymEnc}{\textrm{SymAuthEnc}}
\newcommand{\SymDec}{\textrm{SymAuthDec}}

\newcommand{\Share}{\emph{\textrm{Share}}\xspace}
\newcommand{\Recon}{\emph{\textrm{Recon}}\xspace}

\newcommand{\DKG}{\textrm{DKG}}

\newcommand{\QUAL}{\textrm{QUAL}}

\newcommand{\Fmal}{\mathcal{F}_{\text{mal}}}
\newcommand{\Fmalrb}{\mathcal{F}_{\text{mal-robust}}}
\newcommand{\Fsetup}{\mathcal{F}_{\text{setup}}}
\newcommand{\Frand}{\mathcal{F}_{\text{rand}}}
\newcommand{\Fsumt}{\mathcal{F}_{\text{sum}}}

\newcommand{\Sim}{\mathcal{S}}

\newcommand{\Coms}{Decryptors\xspace}
\newcommand{\coms}{decryptors\xspace}

\newcommand{\com}{decryptor\xspace} 

\newcommand{\comset}{\mathcal{D}}

\usepackage{caption,subcaption}
\usepackage{graphicx}
\usepackage{booktabs}  
\usepackage{dcolumn}   
\usepackage{multirow}  
\usepackage{makecell}  
\usepackage{rotating}
\usepackage{xspace}
\usepackage{hyphenat}  
\usepackage{color}
\usepackage[dvipsnames]{xcolor}
\usepackage{enumerate}

\usepackage[square,comma,numbers,sort&compress]{natbib}

\usepackage[normalem]{ulem}      

\usepackage{wrapfig}
\usepackage{textcomp}
\usepackage{lastpage}
\usepackage{tabularx}
\usepackage{pifont}
\usepackage{bm}
\usepackage{adjustbox}
\usepackage{blindtext}
\usepackage{ifthen}

\makeatletter
\def\NAT@def@citea{\def\@citea{\NAT@separator}}
\makeatother

\usepackage{float}
\usepackage{newfloat}
\DeclareFloatingEnvironment[
    fileext=los,
    listname=List of Schemes,
    name=Scheme,
    placement=H,
    within=section,
]{scheme}

\usepackage{times}
\PassOptionsToPackage{hyphens}{url}
\usepackage[breaklinks=true,
	    pdfdisplaydoctitle=true,
	    pdfborder={0 0 0},
	    bookmarksnumbered=true,
	    linktocpage,
	    citebordercolor={.60 .60 .60}, 
	    linkbordercolor={.60 .60 .60},
	    urlbordercolor={.60 .60 .60},
	    pdfpagelabels,
      colorlinks=true,
      citecolor=BrickRed,
      linkcolor=BrickRed,
	    pdfpagelayout=SinglePage,
	    hyperfootnotes=false,
	    ]{hyperref}

\usepackage{wrapfig}

\usepackage{balance}

\usepackage{colortbl} 
\usepackage{array}    

\usepackage{dblfloatfix}

\usepackage{listings}
\definecolor{maroon}{cmyk}{0, 0.87, 0.68, 0.32}
\definecolor{halfgray}{gray}{0.55}
\definecolor{ipython_frame}{RGB}{207, 207, 207}
\definecolor{ipython_bg}{RGB}{247, 247, 247}
\definecolor{ipython_red}{RGB}{186, 33, 33}
\definecolor{ipython_green}{RGB}{0, 128, 0}
\definecolor{ipython_cyan}{RGB}{64, 128, 128}
\definecolor{ipython_purple}{RGB}{170, 34, 255}
\definecolor{syellow}{HTML}{B58900}
\definecolor{sorange}{HTML}{CB4B16}
\definecolor{smagenta}{HTML}{D33682}
\definecolor{sviolet}{HTML}{6C71C4}
\definecolor{sblue}{HTML}{268BD2}
\definecolor{scyan}{HTML}{2AA198}
\definecolor{sgreen}{HTML}{859900}

\usepackage{algpseudocode}
\algrenewcomment[1]{\hfill// #1}%
\algnotext{EndFunction}
\algnotext{EndFor}
\algnotext{EndWhile}
\algnotext{EndIf}

\usepackage{amsmath,amscd}
\usepackage{amssymb}
\usepackage{amsfonts}
\usepackage{amsthm}
\usepackage[most]{tcolorbox}

\algrenewcommand\algorithmicindent{1em}
\algrenewcommand{\algorithmiccomment}[1]{{\color{CadetBlue}\hfill// #1}}

\usepackage[noeka]{mathrmletter}

\newif\ifextended
\newif\iflongbatching
\newif\ifsubmission
\newif\ifelementary

\ifx\buildextended\undefined
\else
    \extendedtrue
\fi

\ifx\noeditingmarks\undefined
\newcommand{\pgwrapper}[3]{\begingroup \color{#1} #2: #3 \endgroup}
\newcommand{\pgwrapperb}[1]{\textbf{#1}}
\else
   \newcommand{\pgwrapperb}[1]{}
   \newcommand{\pgwrapper}[3]{}
\fi

\def\hn{\usefont{OT1}{phv}{mc}{n}\selectfont}

\newcommand{\mpfont}{\hn\scriptsize}

\ifx\noeditingmarks\undefined
\newcommand{\MPworker}[2]{{\color{#1}\vrule\vrule}{\marginpar{\color{#1}\mpfont #2}}}
\else
    \newcommand{\MPworker}[2]{}
\fi

\ifx\noeditingmarks\undefined
    
\else
    
\fi

\newcommand{\sys}{Flamingo\xspace}  

\theoremstyle{definition}
\newtheorem{theorem}{Theorem}
\newtheorem{definition}{Definition}

\newtheorem{lemma}{Lemma}

\makeatletter
\renewcommand*{\@fnsymbol}[1]{\ensuremath{\ifcase#1\or \star\or \dagger\or \ddagger\or
   \mathsection\or \mathparagraph\or \|\or **\or \dagger\dagger
   \or \ddagger\ddagger \else\@ctrerr\fi}}
\makeatother

\newcommand{\A}{\mathcal{A}}
\newcommand{\C}{\mathcal{C}}

\newcommand{\Z}{\mathbb{Z}}
\newcommand{\Q}{\mathbb{Q}}

\newcommand{\R}{\mathcal{R}}

\newcommand{\ceil}[1]{\lceil#1\rceil}

\makeatletter
\def\imod#1{\allowbreak\mkern10mu({\operator@font mod}\,\,#1)}
\makeatother


\def\compactify{\itemsep=0in \topsep=2pt \parsep=0.00in \partopsep=0pt
\leftmargin=2em}
\let\latexusecounter=\usecounter

\newenvironment{myitemize}%
  {\begin{list}{\labelitemi}{\itemsep3pt \topsep3pt \parsep0.00in
  \partopsep=3pt \leftmargin1.2em}}%
  {\end{list}}
\newenvironment{myitemize2}%
  {\begin{list}{\labelitemi}{\itemsep1pt \topsep2pt \parsep0.00in
  \partopsep=1pt \leftmargin1.2em}}%
  {\end{list}}
  {\begin{list}{\labelitemi}{\itemsep2pt \topsep2pt \parsep0.00in
  \partopsep=0pt \leftmargin1.2em}}%
  {\end{list}}
  {\begin{list}{\threequartdash}{\itemsep3pt \topsep3pt \parsep0.00in
  \partopsep=3pt \leftmargin1.5em}}%
  {\end{list}}

\def\compactsortof{\itemsep=0in \topsep=2pt \parsep=0.00in \partopsep=0pt
\leftmargin=1.7em}
\newenvironment{myenumerate2}
  {\def\usecounter{\compactsortof\latexusecounter}
   \begin{enumerate}}
  {\end{enumerate}\let\usecounter=\latexusecounter}

\def\compactsqueeze{\itemsep=0pt \topsep0pt \parsep=0ex \partopsep=0pt
\leftmargin=1.63em}

\ifx\normalpar\undefined
  
\else
  
\fi

\def\discretionaryslash{\discretionary{/}{}{/}}
{\catcode`\/\active
\gdef\URLprepare{\catcode`\/\active\let/\discretionaryslash
        \def~{\char`\~}}}%
\def\URL{\bgroup\URLprepare\realURL}%
\def\realURL#1{\tt #1\egroup}%

\newcommand{\heading}[1]{\vspace{1ex}\noindent\textbf{#1}}

\begin{document}

\newboolean{longver}
\setboolean{longver}{true}

\title{\sys: Multi-Round Single-Server Secure Aggregation \\ with Applications to Private Federated Learning}

\author{Yiping Ma$^\star$ \quad  
        Jess Woods$^\star$ \quad  
        Sebastian Angel$^{\star\dagger}$ \quad  
        Antigoni Polychroniadou$^\ddagger$ \quad  
        Tal Rabin$^{\star}$ \\
        {\normalsize \em $^\star$University of Pennsylvania \quad $^\dagger$Microsoft Research 
        \quad $^\ddagger$J.P. Morgan AI Research $\&$ AlgoCRYPT CoE}}

\date{}

\pagestyle{plain}

\maketitle

\ifthenelse{\boolean{longver}}{%
\thispagestyle{plain}
}{
\pagenumbering{gobble}
}

\begin{abstract}
This paper introduces \emph{\sys}, a system for \emph{secure aggregation} of data across a large set of clients.
In secure aggregation, a server sums up the private inputs of clients and obtains the 
  result without learning anything about the individual inputs beyond what is implied 
  by the final sum.
\sys{} focuses on the multi-round setting found in federated learning in which many 
  consecutive summations (averages) of model weights are performed to derive a good model. Previous protocols, such as Bell et al. (CCS '20), have been designed for a single round and are adapted to the federated learning setting by repeating the protocol multiple times. 
\sys{} eliminates the need for the per-round setup of previous protocols, and has a new lightweight
  dropout resilience protocol to ensure that if clients leave in the middle of a sum the server can still obtain a meaningful result. 
Furthermore, \sys{} introduces a new way to locally choose the 
  so-called client neighborhood introduced by Bell et al. 
These techniques help \sys{} reduce the number of 
  interactions between clients and the server, resulting in a significant reduction in the end-to-end runtime for a full training session over prior work.

We implement and evaluate \sys{} and show that it can securely train a neural network on the (Extended) MNIST and CIFAR-100 datasets, and the model converges without a loss in accuracy, compared to a non-private federated learning system.

\end{abstract}

\section{Introduction}\label{s:intro}
In \emph{federated learning}, a server wants to train a model using 
  data owned by many clients (e.g., millions of mobile devices).
In each \emph{round} of the training, the server randomly selects a subset of clients, 
  and sends them the current model's weights. 
Each selected client updates the model's weights by running a prescribed training algorithm
  on its data locally, and then sends the updated weights to the server.
The server updates the model by averaging the collected weights.
The training takes multiple such rounds until the model converges.

This distributed training pattern is introduced with the goal of providing a 
  critical privacy guarantee in training---the raw data never leaves 
  the clients' devices. 
However, prior works~\cite{melis18exploiting, zhu19deep} show that the 
  individual weights still leak information about the raw data,
  which highlights the need for a mechanism that can securely aggregate the weights
  computed by client devices~\cite{mcmahan17communication, yuan20federated}.
This is precisely an instance of \emph{secure aggregation}. 

Many protocols and systems for secure aggregation have been proposed, e.g., 
  in the scenarios of private error reporting and statistics
  collection~\cite{shi11privacy, chan11privacy, corrigan-gibbs17prio,
  anderson21aggregate, boneh21lightweight, zhong22ibex}.
However, secure aggregation in federated learning, due to its specific model,
  faces unprecedented challenges: a large number of clients, high-dimensional 
  input vectors (e.g., model weights), multiple rounds of aggregation prior
  to model convergence, and unstable devices (i.e., some devices might go 
  offline prior to completing the protocol).
It is therefore difficult to directly apply these protocols in a black-box way
  and still get good guarantees and performance.

Recent works~\cite{bonawitz17practical, bell20secure, so22lightsecagg} propose
  secure aggregation tailored to federated learning scenarios.
In particular, a state-of-the-art protocol~\cite{bell20secure} (which we call
  BBGLR) can handle one aggregation with thousands of clients and high-dimensional
  input vectors, while tolerating devices dropping out at any point during 
  their execution.
The drawback of these protocols is that they only examine one round
  of aggregation in the 
  full training process, i.e., a selection of a 
  subset of the clients and a sum over their inputs.

Utilizing the BBGLR protocol (or its variants) multiple times to do summations in the full 
  training of a model incurs high costs.
Specifically, these protocols follow the pattern of having each client establish
  input-independent secrets with several other clients 
  in a \emph{setup} phase, and then 
  computing a single sum in a \emph{collection} phase using the secrets.
These secrets cannot be reused for privacy reasons.
Consequently, for each round of aggregation in the training process, 
  one needs to perform an expensive fresh setup.
Furthermore, in each \emph{step} (client-server round trip) of the setup and 
  the collection phases, the server has to interact with all of the clients. 
In the setting of federated learning, such interactions are especially costly 
  as clients may be geographically
  distributed and may have limited compute and battery power or varying network conditions. 

In this paper we propose \emph{\sys{}}, the first single-server secure 
  aggregation system that works well for multiple rounds of aggregation 
  and that can support full sessions of training in the stringent setting
  of federated learning.
At a high level, \sys{} introduces a one time setup and a collection 
  procedure for computing multiple sums, such that the secrets established in 
  the setup can be reused throughout the collection procedure.
For each summation in the collection procedure,
  clients mask their input values (e.g., updated weight vectors in
  federated learning) with a random mask derived from those secrets, and then 
  send the masked inputs to the server.
The server sums up all the masked inputs and obtains the final aggregate value 
  (what the server wants)  masked with the sum of all masks.
\sys{} then uses a lightweight protocol whereby a small number of 
  randomly chosen clients (which we call \emph{\coms}) interact with the 
  server to remove the masks.

The design of \sys{} significantly reduces the overall training time of a 
  federated learning model compared to running BBGLR multiple times.
First, \sys{} eliminates the need for a setup phase for each summation, 
  which reduces the number of steps in the full training session.
Second, for each summation, \sys{} only has one step that requires the server 
  to contact all of the clients (asking for their inputs); the rest of the 
  interactions are performed between the server and a few clients who serve as \coms.

Besides training efficiency, the fact that a client needs to 
  only speak once in a round reduces the model's bias towards results
  that contain only data from powerful, stably connected devices.
In \sys{}, the server contacts the clients once to collect inputs;
  in contrast, prior works have multiple client-server interaction steps 
  in the setup phase (before input collection) and thus
  filter out weak devices for summation, as staying available longer is challenging. 
Seen from a different angle, if we fix the number of clients, the failure 
  probability of a client in a given step, and the network conditions, \sys's results 
  are of higher quality, as they are computed over more inputs than prior works.

In summary, \sys's technical innovations are:
\begin{myitemize2}
  \item \emph{Lightweight dropout resilience.} A new mechanism to 
    achieve dropout resilience in which the server only contacts a 
    small number of clients to remove the masks. All other clients are free to leave after the one step in which they submit their inputs
    without harming the results.
  \item \emph{Reusable secrets.} A new way to generate masks that allows the
    reuse of secrets across rounds of aggregation.
  \item \emph{Per-round graphs.} A new graph generation procedure
    that allows clients in \sys{} to unilaterally determine (without the help
    of the server as in prior work) which pairwise masks they should use in any 
    given round.
\end{myitemize2}

These advancements translate into significant performance 
  improvements~(\S\ref{s:eval:mnist}).
For a 10-round pure summation task, \sys{} is 3$\times$ faster than BBGLR (this
  includes \sys's one-time setup cost), and includes the contribution of more 
  clients in its result.
When training a neural network on the Extended MNIST
dataset, \sys{} takes about 40 minutes to converge while BBGLR needs roughly 3.5
hours to reach the same training accuracy.

\section{Problem Statement}\label{s:pb}

Secure aggregation is useful in a variety of domains:
  collecting, in a privacy-preserving way, error
  reports~\cite{corrigan-gibbs17prio,boneh21lightweight}, usage statistics~\cite{shi11privacy, chan11privacy}, 
  and ad conversions~\cite{zhong22ibex, anderson21aggregate}; it has even
  been shown to be a key building block for computing private 
  auctions~\cite{zhong23addax}.
But one key application is \emph{secure federated learning}, whereby a server wishes
  to train a model on data that belongs to many clients, but the clients do not wish 
  to share their data (or other intermediate information such as weights that 
  might leak their data~\cite{zhu19deep}).
To accomplish this, each client receives the original model from the server and computes 
  new \emph{private} weights based on their own data.
The server and the clients then engage in a secure aggregation protocol that helps the server
  obtain the sum of the clients' private weights, without learning anything
  about individual weights beyond what is implied by their sum.
The server then normalizes the sum to obtain the average weights which 
  represent the new state of the model.
The server repeats this process until the training converges.

This process is formalized as follows. 
Let $[z]$ denote the set of integers $\{1, 2, \ldots, z\}$, and let 
  $\vec{x}$ denote a vector; all operations on vectors are component-wise.  
A total number of $N$ clients are fixed before the training starts.
Each client is indexed by a number in $[N]$.
The training process consists of $T$ rounds.
In each round $t \in [T]$, a set of clients is randomly sampled 
  from the $N$ clients, denoted as $S_t$.
Each client $i \in S_t$ has an input vector, $\vec{x}_{i,t}$, for round $t$. (A client may be selected
in multiple rounds and may have different inputs in each round.) In each round, the server wants to securely compute the sum of the $|S_t|$ input vectors, 
  $\sum_{i\in S_t} \vec{x}_{i,t}$. 
  
In practical deployments of federated learning, a complete sum is hard to guarantee, 
as some clients may drop out in the middle of the aggregation process and
the server must continue the protocol without waiting for them to come back (otherwise
the server might be blocked for an unacceptable amount of time).
So the real goal is to compute the sum of the input vectors from the largest possible 
subset of $S_t$; we elaborate on this in the next few sections.

\subsection{Target deployment scenario}
Based on a recent survey of
federated learning deployments~\cite{kairouz21advances}, common 
  parameters are as follows.
$N$ is in the range of $100$K--$10$M clients, where $|S_t| = 50$--$5{,}000$ 
  clients are chosen to participate in a given round $t$.
The total number of rounds $T$ for a full training session is $500$--$10{,}000$. 
Input weights ($\vec{x}_{i,t}$) have typically on the order of 1K--500K entries for the datasets we surveyed~\cite{krizhevskycifar, krizhevsky2009learning, caldas2018leaf, caldasleafgithub}.

Clients in these systems are heterogeneous devices with varying degrees of 
  reliability (e.g., cellphones, servers) and 
  can stop responding due to device or network failure.

\subsection{Communication model}\label{s:comm-model}
Each client communicates with the server through a private and authenticated channel.
Messages sent from clients to other clients are forwarded via the server,
  and are end-to-end encrypted and authenticated.

\subsection{Failure and threat model}\label{s:threat}
We model failures of two kinds: (1) honest clients that disconnect
  or are too slow to respond as a result of unstable network conditions, power loss, etc; 
  and (2) arbitrary actions by an adversary that controls the server and a 
  bounded fraction of the clients.
We describe each of these below. 

\heading{Dropouts.} 
In each round of summation, the server interacts with the clients (or a subset 
  of them) several times (several \emph{steps}) to compute a sum. 
If a server contacts a set of clients during a step and some of the clients are 
  unable to respond in a timely manner, the server has no choice but to keep 
  going; the stragglers are dropped and do \emph{not} participate in the 
  rest of the steps for this summation round. 
In practice, the fraction of dropouts in the set depends on the client response time distribution
  and a server timeout (e.g., one second); longer timeouts mean lower fraction of dropouts. 

There are two types of clients in the system: regular clients that
  provide their input, and \emph{\coms}, who are special clients whose 
  job is to help the server recover the final result. 
We upper bound the fraction of regular clients that drop out
  in any given aggregation round by $\delta$,
  and upper bound the fraction of \coms 
  that drop out in any given aggregation round by $\delta_D$.

\heading{Adversary.}
We assume a static, malicious adversary that corrupts the server 
  and up to an $\eta$ fraction of the total $N$ clients.
That is, the adversary compromises $N \eta$ clients independent of
  the protocol execution and the corrupted set stays the same
  throughout the entire execution (i.e., all rounds).
Note that malicious clients can obviously choose to drop out during protocol 
  execution, but to make our security definition and analysis clear, 
  we consider the dropout of malicious clients separate from, and in addition to, 
  the dropout of honest clients.

Similarly to our dropout model, we distinguish between the fraction
  of corrupted regular clients ($\eta_{S_t}$) and corrupted \coms ($\eta_D$).
Both depend on $\eta$ but also on how \sys{} samples 
  regular clients and \coms from the set of all $N$ clients.
We defer the details to Appendix~\ref{app:failure-model}, but briefly, 
  $\eta_{S_t} \approx \eta$; 
  and given a (statistical) security parameter $\kappa$, $\eta_D$ is upper bounded by $\kappa \eta$ with probability $2^{-\Theta(\kappa)}$.

\heading{Threshold requirement.}
The minimum requirement for \sys{} to work is $\delta_D +  \eta_D < 1/3$.  
For a target security parameter $\kappa$, we show in
Appendix~\ref{app:param} how to select other parameters 
  for \sys{} to satisfy the above requirement and result in minimal asymptotic costs.

\heading{Comparison to prior works.}
BBGLR and other works~\cite{bonawitz17practical, bell20secure}
  also have a static, malicious adversary but only for a single round of aggregation. In fact, in Section~\ref{s:gaps} we show that their protocol cannot be
  naturally extended to multiple aggregations that withstands a malicious adversary
  throughout.

\subsection{Properties}\label{s:properties}

\sys{} is a secure aggregation system that achieves the following 
  properties under the above threat and failure model.
We give informal definitions here, and defer the formal definitions to Section~\ref{s:formal-thms}.
\begin{myitemize2}
    \item \emph{Dropout resilience}: 
    when all parties follow the protocol, 
    the server, in each round $t$, will get a sum of inputs from all the online
    clients in $S_t$. 
    Note that this implicitly assumes that the protocol both completes all the rounds
    and outputs meaningful results.
   
    \item \emph{Security}: for each round $t$ summing over the inputs
      of clients in $S_t$, a malicious adversary learns the sum of inputs 
      from at least $(1-\delta-\eta) |S_t|$ clients. 
\end{myitemize2}

\medskip

Besides the above, we introduce a new notion that quantifies
  the quality of a single sum. 
\begin{myitemize2}
    \item \emph{Sum accuracy}: A round of summation has sum 
    accuracy $\tau$ if the final sum result contains
  the contribution of a $\tau$ fraction of the clients who are selected to 
  contribute to that round (i.e., $\tau |S_t|$).
\end{myitemize2}

\heading{Input correctness.}
In the context of federated learning, if malicious clients input bogus weights, 
  then the server could derive a bad model (it may even contain ``backdoors'' that cause
  the model to misclassify certain inputs~\cite{bagdasaryan20how}).
Ensuring correctness against this type of attack is out of the scope of this work; to our knowledge, 
  providing strong guarantees against malicious inputs remains an open problem.
Some works~\cite{roth19honeycrisp, roth20orchard, roth21mycelium, chowdhury21eiffel, angel22efficient, bell22acorn} 
  use zero-knowledge proofs to bound how much a client can bias the final 
  result, but they are unable to formally prove the absence of all possible attacks.

\heading{Comparison to prior work.}
\sys{} provides a stronger security guarantee than BBGLR.\@
In \sys, an adversary who controls the server and some clients learns
  a sum that contains inputs from at least a $1-\delta-\eta$ fraction of clients. 
In contrast, the malicious protocol in BBGLR leaks several sums:
  consider a partition of the $|S_t|$ clients,  
  where each partition set has size at least $\alpha \cdot |S_t|$;
  a malicious adversary in BBGLR learns the sum of each of the partition sets.
Concretely, for 5K clients, when both $\delta$ and $\eta$ are 0.2,
  $\alpha < 0.5$.
This means that the adversary learns the sum of two subsets. 
This follows from Definition 4.1 and Theorem 4.9 in BBGLR~\cite{bell20secure}.

\section{Background}\label{bg}

In this section, we discuss BBGLR~\cite{bell20secure}, which is the state of the
  art protocol for a single round secure aggregation in the federated learning setting. 
We borrow some ideas from this protocol, but design \sys quite differently
  in order to support a full training session.

\subsection{Cryptographic building blocks}\label{s:blocks}

We start by reviewing some standard cryptographic primitives used by BBGLR and \sys.

\heading{Pseudorandom generators.}
A PRG is a deterministic function 
  that takes a random seed in $\{0,1\}^{\lambda}$ and outputs a longer string that is computationally 
  indistinguishable from a random string ($\lambda$ is the computational security parameter).
For simplicity, whenever we use PRG with a seed that is not in $\{0,1\}^\lambda$, 
  we assume that there is a deterministic function that maps the seed to
  $\{0,1\}^\lambda$ and preserves security.
Such mapping is discussed in detail in Section~\ref{s:impl}.

\heading{Pseudorandom functions.}
A $\PRF: \mathcal{K} \times \mathcal{X} \rightarrow \mathcal{Y}$ is a family of deterministic functions indexed by a key 
  in $ \mathcal{K}$ that map an input in $\mathcal{X}$ to an output in $\mathcal{Y}$ in such a way that the indexed function 
  is computationally indistinguishable from a truly random function from $\mathcal{X}$ to $\mathcal{Y}$.
We assume a similar deterministic map for inputs as described in the PRG above.

\heading{Shamir's secret sharing.}
An $\ell$-out-of-$L$ secret sharing scheme consists of the following two
  algorithms, $\Share$ and $\Recon$.
$\Share(s, \ell, L) \rightarrow (s_1, \ldots, s_L)$ takes in a secret $s$, a threshold
$\ell$, and the number of desired shares $L$, and outputs $L$
shares $s_1, \ldots, s_L$.
$\Recon$ takes in at least $\ell + 1$ of the shares,
and output the secret $s$; i.e.,
for a set $U \subseteq [L]$ and $|U| \ge \ell+1$,
$\Recon(\{s_u\}_{u \in U}) \rightarrow s$. 
Security requires that fewer than $\ell+1$ shares reveal no information about $s$.

\heading{Diffie-Hellman key exchange.}
Let $\G$ be a group of order $q$ in which the Decisional Diffie-Hellman (DDH) problem is hard,
  and $g$ be a generator of $\G$.
Alice and Bob can safely establish a shared secret (assuming a passive adversary) as follows. 
Alice samples a secret $a \xleftarrow{\$} \Z_q$, and sets her public value to $g^a \in \G$.
Bob samples his secret $b \xleftarrow{\$} \Z_q$, and sets his public value to $g^b \in \G$.
Alice and Bob exchange the public values and raise the other party’s value to their secret, 
  i.e., $g^{ab} = {(g^a)}^b = {(g^b)}^a$.
If DDH is hard, the shared secret $g^{ab}$ is only known to Alice and Bob but no one else.

\subsection{The BBGLR protocol}\label{s:bbglr}
BBGLR is designed for computing a single sum on the inputs of a set of clients.
To apply it to the federated learning setting, we can simply assume that
  in a given round of the training process, there are $n$ clients selected 
  from a large population of size $N$.
We can then run BBGLR on these $n$ clients to compute a sum of their inputs.

The high level idea of BBGLR is for clients to derive pairwise
  random masks and to add those masks to their input vectors 
  in such a way that when all the masked input vectors across 
  all clients are added, the masks cancel out.
It consists of a \emph{setup} phase and a \emph{collection} phase.
We first describe a semi-honest version below.

\heading{Setup phase.} 
The setup phase consists of three steps: 
  (1) create a database containing public keys of all of the $n$ clients; 
  (2) create an undirected graph where vertices are clients, and each vertex
  has enough edges to satisfy certain properties;
  (3) have each client send shares of two secrets to its neighbors in the graph.
We discuss these below.

In the first step, each client $i \in [n]$ generates a secret $a_i$ and
  sends $g^{a_i}$ to the server, where $g^{a_i}$ represents client $i$'s
  public key.
The server then stores these public keys in a database.
Note that the malicious-secure version of BBGLR requires the server to be semi-honest for this particular step, or
  the existence of a trusted public key infrastructure (PKI).

In the second step, the graph is established as follows.
Each client $i \in [n]$ randomly chooses $\gamma$ other clients in $[n]$ as 
  its neighbors, and tells the server about their choices.
After the server collects all the clients' choices, it notifies each client of 
  their neighbors indexes in $[n]$ and public keys. 
The neighbors of client $i$, denoted as $A(i)$, are those corresponding to 
  vertices that have an edge with $i$ (i.e., $i$ chose them or they chose $i$).

Finally, each client $i$ uses Shamir's secret sharing to share $a_i$ and 
  an additional random value $m_i$ to its neighbors $A(i)$ (let the threshold 
  be $\ell < |A(i)|$), where the shares are end-to-end encrypted with a secure authenticated encryption scheme and sent via the server
  (\S\ref{s:comm-model}).

\heading{Collection phase.}
Client $i$  sends
  the following masked vector to the server: 
\[\small  Vec_i = \vec{x}_{i} + \underbrace{\sum_{j \in A(i), i < j} \PRG(r_{i,j} )
  - \sum_{j \in A(i), i > j} \PRG(r_{i,j} )}_{\text{pairwise mask}} 
  + \underbrace{\PRG(m_i)}_{\text{individual mask}}, \]
  where $r_{i,j} =  g^{a_i a_j}$, which can be computed by client $i$ since it has 
  the secret $a_i$ and $j$'s public key, $g^{a_j}$. (These are essentially Diffie-Hellman key exchanges between a client and its neighbors.)
Here we view the output of the $\PRG$ as a vector of integers instead
  of a binary string.
Also, we will write the pairwise mask term as $\small \sum_{j \in A(i)} \pm \PRG(r_{i,j} )$ for ease of notation.

As we mentioned earlier (\S\ref{s:pb}), clients may drop out due to unstable network 
  conditions, power loss, etc.
This means that the server may not receive some of the masked vectors 
  within an acceptable time period. 
Once the server times out, the server labels the clients whose vectors have been 
  received  as ``online''; the rest are labeled ``offline''. 
The server shares this information with all the $n$ clients.
The server then sums up all of the received vectors, which yields
  a vector that contains the masked sum.
To recover the correct sum, the server needs a way to remove the masks.
It does so by requesting for each \emph{offline} client $i$, the shares of 
  $a_i$ from $i$'s neighbors; and for each \emph{online} client $j$, the 
  shares of $m_{j}$ from $j$'s neighbors.
These shares allow the server to reconstruct either the pairwise mask or the 
  individual mask for each client.
As long as there are more than $\ell$ neighbors that send the requested shares,
  the server can successfully remove the masks and obtain the sum. 
This gives the dropout resilience property of BBGLR.

One might wonder the reason for having the individual mask $m_i$, since
  the pairwise mask already masks the input.
To see the necessity of having $m_i$, assume that it is not added,
  i.e., $\small Vec_i = \vec{x}_i + \sum \pm \PRG(r_{ij})$.
Suppose client $i$'s message is sent but not received on time. 
Thus, the server reconstructs $i$'s pairwise mask $\small \sum \pm \PRG(r_{ij})$. 
Then, $i$'s vector $Vec_i$ arrives at the server.
The server can then subtract the pairwise mask from $Vec_i$ to learn $\vec{x}_i$. 
The individual mask $m_i$ prevents this.

\heading{Preventing attacks in fault recovery.}
The above protocol only works in the semi-honest setting.
There are two major attacks that a malicious adversary can perform. 
First, a malicious server can give inconsistent dropout information
  to honest clients and recover both the pairwise and individual masks.
For example, suppose client $i$ has neighbors $j_1, \ldots, j_\gamma$,
  and a malicious server lies to the neighbors of $j_1, \ldots, j_\gamma$
  that $j_1, \ldots, j_\gamma$ have dropped out
  (when they actually have not).
In response, their neighbors, including $i$, will provide the server with the information it needs to
  reconstruct $a_{j_1}, \ldots, a_{j_\gamma}$, thereby deriving all the 
  pairwise secrets $r_{i, j_1}, \ldots, r_{i, j_\gamma}$.
At the same time, the server can tell $j_1, \ldots,
  j_\gamma$ that $i$ was online and request the shares of $m_i$.
This gives the server both the pairwise mask and the individual mask of client 
  $i$, violating $i$'s privacy.
To prevent this, BBGLR has a consistency check step performed among all neighbors of 
  each client to reach an agreement on which nodes actually dropped out.
In this case, $i$ would have learned that none of its neighbors dropped out
  and would have refused to give the shares of their pairwise mask.

Second, malicious clients can submit a share that is different than the share that they received from their neighbors.
This could lead to reconstruction failure at the server, or to the server 
  deriving a completely different secret.
BBGLR fixes the latter issue by having the clients hash their secrets and 
  send these hashes to the server when they send their input vectors; 
  however, reconstruction could still fail because of an insufficient threshold in error correction\footnote{To apply a standard error correction algorithm such as Berlekamp-Welch in this setting, the polynomial degree should be at most $\gamma/3$.
  Definition 4.2 in BBGLR implies that the polynomial degree may be larger than required for error correction.}.

In sum, the full protocol of BBGLR that withstands a malicious adversary (assuming 
  a PKI or a trusted server during setup) has six steps in total: three steps 
  for the setup and three steps for computing the sum. 

\subsection{Using BBGLR for federated learning}\label{s:gaps}
BBGLR works well for \emph{one} round of training, 
   but when many rounds are required, several issues arise.
First, in federated learning the set of clients chosen to participate in 
  a round changes, so a new graph needs to be derived and new secrets must be 
  shared.
Even if the graph stays the same, the server cannot reuse the secrets from the 
  setup in previous rounds as the masks are in fact one-time pads that cannot 
  be applied again.
This means that we must run the setup phase for each round, which incurs a 
  high latency since the setup contains three steps involving all the clients. 

Moreover, BBGLR's threat model does not naturally extend to
  multi-round aggregation.
It either needs a semi-honest server or a PKI during the first step of 
  the protocol.
If we assume the former, then this means the adversary
  has to be semi-honest during the exact time of setup in each round, 
  which is practically impossible to guarantee.
If we use a PKI, none of the keys can be reused (for the above reasons);
  as a result, all of the keys in the PKI need to be updated for each round, which
  is costly.

\section{Efficient Multi-Round Secure Aggregation}\label{s:overview}
Flamingo supports multi-round aggregation without redoing the setup for each 
  round and withstands a malicious adversary throughout.
The assumptions required are: (1) in the setup, all parties are provided with the 
  same random seed from a trusted source (e.g., a distributed randomness 
  beacon~\cite{das22spurt}); and (2) a PKI 
  (e.g., a key transparency log~\cite{chase19seemless, melara15coniks, leung22aardvark,
  tomescu19transparency, tyagi22versa, tzialla22transparency, hu21merkle2}).
Below we describe the high-level ideas underpinning \emph{\sys} (\S\ref{s:intuition}) and
  then we give the full protocol (\S\ref{s:detail:setup} and \S\ref{s:detail:collection}).

\subsection{High-level ideas}\label{s:intuition}
\sys{} has three key ideas: 
  
\heading{(1) Lightweight dropout-resilience.}
Instead of asking clients to secret share $a_i$ and $m_i$ for their masks with all of their
  neighbors, we let each client encrypt---in a special way---the PRG seeds of their pairwise and 
  individual masks, append the resulting ciphertexts to their masked 
  input vectors, and submit them to the server in a single step.
Then, with the help of a special set of $L$ clients that we call \emph{\coms},
  the server can decrypt one of the two seeds associated with each
  masked input vector, but not both.
In effect, this achieves a similar fault-tolerance property as 
  BBGLR~(\S\ref{s:bbglr}), but with a different mechanism.

The crucial building block enabling this new protocol is 
  \emph{threshold decryption}~\cite{gennaro08threshold, desmedt89threshold,
  shoup02securing}, in which clients can encrypt data with a system-wide known
  public key $PK$, in such a way that the resulting ciphertexts can only be 
  decrypted with a secret key $SK$ that is secret shared among \coms in \sys{}. 
Not only does this mechanism hide the full secret key from every party in the system,
  but the \coms can decrypt a ciphertext without ever having to 
  interact with each other.
Specifically, the server in \sys{} sends the ciphertext (pealed from the submitted vector) to each of the
  \coms, obtains back some threshold $\ell$ out of $L$ responses, and
  locally combines the $\ell$ responses which produce the corresponding
  plaintext.
Our instantiation of threshold decryption is based on the ElGamal cryptosystem
  and Shamir's secret sharing; we describe it in Section~\ref{s:distributed-decryption}.  
\ifthenelse{\boolean{longver}}{%
Technically one can also instantiate the threshold decryption with other protocols,
  but we choose ElGamal cryptosystem because it enables efficient proof of decryption
  (for malicious security, \S\ref{s:detail:setup}) 
  and simple distributed key generation. 
}{
}

One key technical challenge that we had to overcome when designing this 
  protocol is figuring out how to secret share the key $SK$ among the \coms{}.
To our knowledge, existing efficient distributed key generation (DKG) 
  protocols~\cite{pedersen91threshold, canny04practical, gennaro06secure} assume a broadcast
  channel or reliable point-to-point channels, whereas
  our communication model is that of a star topology where all messages 
  are proxied by a potential adversary (controlling the server) that can drop them. 
There are also asynchronous DKG protocols~\cite{kokoriskogias20asynchronous, das22practical, abraham23bingo}, but   standard asynchronous communication model assumes eventual delivery of messages which is not the case in our      setting. 
In fact, we can relax the guarantees of DKG and Section~\ref{s:detail:setup} gives an 
  extension of a discrete-log-based DKG protocol~\cite{gennaro06secure} (since we 
  target ElGamal threshold decryption) that works in the star-topology communication model.

In sum, the above approach gives a dropout-resilient protocol for a single summation with two steps: 
  first, each client sends their masked vector and the ciphertexts of the PRG seeds; second, the server uses
  distributed decryption to recover the seeds (and the masks) for dropout clients (we discuss
  how to ensure that \coms agree on which of the two seeds to decrypt in
  \S\ref{s:detail:collection}).
This design improves the run time over BBGLR by
  eliminating the need to involve all the clients to remove the masks---the
  server only needs to wait until it has collected enough shares from the decryptors,
  instead of waiting for almost all the shares to arrive.
Furthermore, the communication overhead of appending several 
  small ciphertexts (64 bytes each) to a large input vector (hundreds of KBs) is minimal.

\heading{(2) Reusable secrets.}
\sys's objective is to get rid of the setup phase for each round of aggregation.
Before we discuss \sys's approach, consider what would happen if we were to 
  naively run the setup phase in BBGLR once, followed by running the 
  collection procedure multiple times.
First, we are immediately limited to performing all of the aggregation tasks
  on the same set of clients, since BBGLR establishes the graph of neighbors
  during the setup phase.
This is problematic since federated learning often chooses different sets of
  clients for each round of aggregation~(\S\ref{s:pb}).
Second, clients' inputs are exposed.
To see why, suppose that client $i$ drops out in round~1 but not in 2.
In round~1, the server reconstructs $r_{i,j}$ for $j\in A(i)$ to unmask the sum. 
In round~2, client $i$ sends $\small \vec{x}_i +  \sum_{j\in A(i)} \pm \PRG(r_{i,j}) + \PRG(m_i)$ 
  and the server reconstructs $m_i$ by asking $i$'s neighbors for their shares.
Since all the $r_{i,j}$ are reconstructed in round 1 and are reused in round 2, 
  the server can derive both masks.

The above example shows that the seeds should be new and independent 
in  each round.
We accomplish this with a simple solution that adds a level of indirection.
\sys{} treats $r_{i,j}$ as a long-term secret and lets the clients apply a $\PRF$ to 
  generate a new seed for each pairwise mask. 
Specifically, clients $i$ computes the PRG seed for pairwise mask in round $t$ 
  as $h_{i,j, t} := \PRF(r_{i,j}, t)$ for all $j\in A(i)$. 
Note that client $j$ will compute the same $h_{i,j,t}$ as it agrees with $i$ on $r_{i,j}$. 
In addition, each client also generates a fresh seed $m_{i,t}$ for the individual mask in round $t$.
Consequently, combined with idea (1), each client uses $PK$ to encrypt the per-round seeds, 
  $\{h_{i,j,t}\}_{j\in A(i)}$ and $m_{i,t}$.
Then, the server recovers one of the two for each client. We later describe an optimization where clients do not 
  encrypt $m_{i,t}$ with $PK$ (\S\ref{s:detail:collection}). 

A nice property of $r_{i,j}$ being a long-term secret is that \sys{} 
  can avoid performing all the Diffie-Hellman 
  key exchanges between graph neighbors (proxied through the server). \sys{} relies instead on an external PKI or a verifiable public key directory 
  such as CONIKS~\cite{melara15coniks} and its successors (which are a common building block for bootstrapping end-to-end encrypted systems).

We note that this simple technique cannot be applied to BBGLR 
  to obtain a secure multi-round protocol.\@
It is possible in \sys{} precisely because clients encrypt
  their per-round seeds for pairwise masks directly so the server never needs to compute
  these seeds from the long-term pairwise secrets. 
In contrast, in BBGLR, clients derive pairwise secrets ($g^{a_i, a_j}$)
  during the setup phase.
When client $i$ drops out, the server collects enough shares to
  reconstruct $a_i$ and compute the pairwise secrets,
  $g^{a_i, a_j}$, for all online neighbors $j$ of client $i$.
Even if we use a \PRF{} here, the server already has the
  pairwise secret; so it can run the \PRF{} for any round and conduct
  the attacks described earlier.

 \begin{figure}
 \footnotesize{
  \begin{algorithmic}[1]
    \State Parameters: $\epsilon$. \Comment{the probability that an edge is added}
    
    \Function{ChooseSet}{$v, t, n_t, N$}{}
        \State $S_t \leftarrow \emptyset$.
        \State  $v^*_t :=\PRF(v, t)$.
        \While{$|S_t| < n_t$} 
            \State Parse $\log N$ bits from $\PRG(v^*_t)$ as $i$, add $i$ to $ S_t$. 
        \EndWhile
        \State Output $S_t$.
    \EndFunction
    
    \Function{GenGraph}{$v, t, S_t$}{}
        \State $G_t \leftarrow n_t\times n_t$ empty matrix; $\rho
        \leftarrow \log (1/\epsilon)$.
        \For {$i \in S_t, j \in S_t$}
                \State Let $v'$ be the first $\rho$ bits of $\PRF(v, (i, j)) $. 
                \If {$v' = 0^\rho$}
                    set $G_t(i, j) := 1$
                \EndIf
            \EndFor
        \State Output $G_t$.
    \EndFunction
  
    \Function{FindNeighbors}{$ v, S_t, i$}{}
        \State $A_t(i) \leftarrow \emptyset$;  $\rho \leftarrow \log (1/\epsilon)$.
        \For {$j \in S_t$}
                \State Let $v'$ be the first $\rho$ bits of $\PRF(v, (i, j)) $. 
                \If {$v' = 0^\rho$}
                    add $j$ to $A_t(i)$.
                \EndIf
        \EndFor
        \For {$j \in S_t$}
                \State Let $v'$ be the first $\rho$ bits of $\PRF(v, (j, i)) $. 
                \If {$v' = 0^\rho$}
                    add $j$ to $A_t(i)$.
                \EndIf
        \EndFor
        \State Output $A_t(i)$.
    \EndFunction
    
  \end{algorithmic}
 }
\caption{Pseudocode for generating graph $G_t$ in round $t$.}
\label{fig:gengraph}
 \end{figure}

\heading{(3) Per-round graphs.}
BBGLR uses a sparse graph instead of a fully-connected graph for efficiency 
  reasons (otherwise each client would need to secret share its seeds with
  every other client).
In federated learning, however, establishing sparse graphs requires 
  per-round generation since the set $S_t$ changes in each round (some clients
  are chosen again, some are new~\cite{lai21oort}). 
A naive way to address this is to let all clients in $[N]$ establish a 
  big graph $G$ with $N$ nodes in the setup phase: each client in $[N]$ sends 
  its choice of $\gamma$ neighbors to the server, and the server sends to 
  each client the corresponding neighbors.
Then, in each round $t$, the corresponding subgraph $G_t$ consists of clients 
  in $S_t$ and the edges among clients in $S_t$.
 
However, this solution is unsatisfactory.
If one uses a small $\gamma$ (e.g., $\log N$), $G_t$ might not be connected 
  and might even have isolated nodes (leaking a client's input vector 
  since it has no pairwise masks); if one uses a large $\gamma$ (e.g., the 
  extreme case being $N$), $G_t$ will not be sparse and the communication
  cost for the server will be high (e.g., $O(N^2)$).

\sys{} introduces a new approach for establishing the graph with a communication
  cost independent of $\gamma$.
The graph for each round $t$ is generated by a random string 
  $v\in \{0,1\}^\lambda$ known to all participants (obtained from a randomness beacon or a trusted setup).
Figure~\ref{fig:gengraph} lists the procedure.
$\ChooseSet(v, t, n_t, N)$ determines the set of clients involved in round $t$, 
   namely $S_t \subseteq [N]$ with size $n_t$.
The server computes $G_t \leftarrow \GenGraph(v, t, S_t)$ as the graph in 
  round $t$ among clients in $S_t$.
A client $i \in S_t$ can find its neighbors in $G_t$ without materializing the 
  whole graph using $\FindNeighbors(v, S_t, i)$.
In this way, the neighbors of $i$ can be locally generated.
We choose a proper $\epsilon$ such that in each round, the graph is connected with high probability (details in \S\ref{s:security}).
\ifthenelse{\boolean{longver}}
{We note that this technique might be of independent interest.}
{This technique is of independent interest.}
   
\begin{figure*}
\centering
 \includegraphics[width=\linewidth]{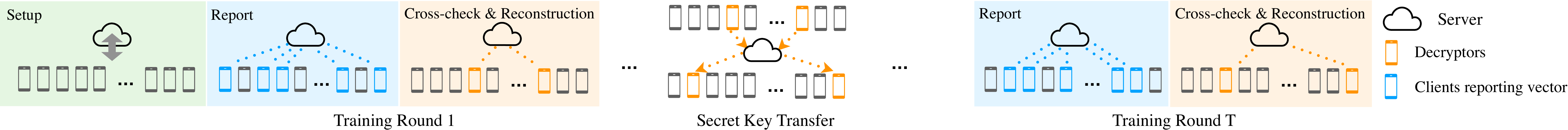}
  \caption{\small Workflow of Flamingo. 
  The server first does a setup for all clients in the system.
  In each round $t$ of training, the server securely aggregates the masked input vectors in
  the report step; in the cross-check and reconstruction steps, the server communicates with 
  a small set of randomly chosen clients who serve 
  as \coms. The \coms are chosen independently from the set $S_t$ that provides inputs in a given round. 
  Every $R$ rounds, the \coms switch and the old \coms transfers shares of $SK$ to new \coms.}
  \label{fig:flow}
\end{figure*}

\medskip

The above ideas taken together eliminate the need for per-round setup, 
  which improves the overall run time of multi-round aggregation over BBGLR.\@ 
Figure~\ref{fig:flow} depicts the overall protocol, and the next sections
  describe each part.

\subsection{Types of keys at PKI}\label{s:detail:pki}
Before giving our protocol, we need to specify what types of keys the PKI needs to store.
The keys depend on the cryptographic primitives that we use 
  (signature schemes, symmetric encryption and ElGamal encryption); for ease of reading,
  we formally give these primitives in Appendix~\ref{app:primitives}.

The PKI stores three types of keys for all clients in $[N]$:
  \begin{myitemize2}
      \item $g^{a_i}$ of client $i$ for its secret $a_i$; this is for client $j$ to derive
        the pairwise secret $r_{i,j}$ with client $i$ by computing $(g^{a_j})^{a_i}$. 
      
      \item $g^{b_i}$ of client $i$ for deriving a symmetric encryption key $k_{i,j}$ for 
      an authenticated encryption scheme $\SymEnc$ (Definition~\ref{def:authenc}); this scheme is used when a client sends messages to another client via the server.
      Later when we say client $i$ sends a message to client $j$ via the server in the protocol,
      we implicitly assume the messages are encrypted using $k_{i,j}$.
      
      \item $pk_i$ of client $i$ for verifying $i$'s signature on messages signed by $sk_i$. 
  \end{myitemize2}

\subsection{Setup phase}\label{s:detail:setup}
The setup phase consists of two parts: (1) distributing a 
  random seed $v\in \{0,1\}^\lambda$ to all participants, and (2) selecting 
  a random subset of clients as \coms{} and distribute the shares of the secret key
  of an asymmetric encryption scheme $\AsymEnc$. In
  our context, $\AsymEnc$ is the ElGamal cryptosystem's encryption function~(Definition~\ref{def:elgamal}).

As we mentioned earlier, the first part can be done through a trusted source of 
  randomness, or by leveraging a randomness beacon that is already 
  deployed, such as Cloudflare's~\cite{cloudflare-beacon}.  
The second part can be done by selecting a set of $L$ clients as \coms, $\mathcal{D}$, 
  using the random seed $v$ ($\ChooseSet$), and then running 
  a DKG protocol among them.
We use a discrete-log based DKG protocol~\cite{gennaro06secure} (which we call
  GJKR-DKG) since it is compatible with the ElGamal cryptosystem. 
However, this DKG does not work under our communication model and requires 
  some changes and relaxations, as we discuss next.

\heading{DKG with an untrusted proxy.}
The correctness and security of the GJKR-DKG protocol relies on
  a secure broadcast channel. Our communication model does not have such a channel,
  since the server can tamper, replay or drop messages.
Below we give the high-level ideas of how we modify GJKR-DKG and Appendix~\ref{app:dkg}
  gives the full protocol.

We begin with briefly describing the GJKR-DKG protocol.
It has a threshold of $1/2$, which means that at most half of the 
  participants can be dishonest; the remaining must perform the 
  following steps correctly:
  (1) Each party $i$ generates a random value $s_i$ and acts as a dealer
  to distribute the shares of $s_i$ (party $j$ gets $s_{i,j}$).
  (2) Each party $j$ verifies the received shares (we defer how the 
  verification is done to Appendix~\ref{app:dkg}).
  If the share from the party $i$ fails the verification, 
  $j$ broadcasts a complaint against party $i$.
  (3) Party $i$ broadcasts, for each complaint from party $j$, the $s_{i,j}$ 
  for verification. (4) Each party disqualifies those parties that fail the 
  verification; the rest of the parties form a set QUAL.\@
Then each party sums up the shares from QUAL to derive a share of the 
  secret key.

Given our communication model, it appears hard to guarantee the standard
  DKG correctness property, which states that if there are enough honest
  parties, at the end of the protocol the honest parties hold valid shares of a unique 
  secret key.
Instead, we relax this correctness property by allowing honest parties to 
  instead abort if the server who is proxying the messages acts maliciously.

We modify GJKR-DKG in the following ways. 
First, we assume the threshold of dishonest participants is $1/3$.
Second, all of the messages are signed; honest parties abort if they do not 
  receive the prescribed messages. 
Third, we add another step before each client decides on the eventual set QUAL: 
  all parties sign their QUAL set and send it to the server;
  the server sends the signed QUALs to all the parties.
Each party then checks whether it receives $2\ell+1$ or more valid signed QUAL 
  sets that are the same. 
If so, then the QUAL set defines a secret key; otherwise the party aborts.
We give the detailed algorithms and the corresponding proofs 
  in Appendix~\ref{app:dkg}. 
Note that the relaxation from GJKR-DKG is that we allow parties to abort (so no secret key is shared at the end), 
  and this is reasonable particularly in the federated learning
  setting because the server will not get the result if it misbehaves.

\medskip

\heading{Decryptors run DKG.}
At the end of our DKG, a \emph{subset} of the selected \coms{} 
  will hold the shares of the secret key $SK$.
The generated $PK$ is signed by the \coms and sent to all of the $N$ clients 
  by the server; the signing prevents the server from distributing different 
  public keys or distributing a public key generated from a set of malicious clients. 
Each client checks if it received $2\ell+1$ valid signed $PK$s from the set of \coms determined
  by the random seed (from beacon); if not, the client aborts.
In Appendix~\ref{app:full-mal}, we provide the pseudocode for the entire setup 
  protocol $\Pi_{\text{setup}}$ (Fig.~\ref{fig:protocol-setup}).

\subsection{Collection phase}\label{s:detail:collection}
The collection phase consists of $T$ rounds; each round $t$ has three 
  steps: report, cross-check, and reconstruction.
Below we describe each step and we defer the full protocol $\Pi_{\text{sum}}$ 
  to Figure~\ref{fig:collection} in Appendix~\ref{app:full-mal}.  
The cryptographic primitives we use here ($\SymEnc$ and $\AsymEnc$) are formally given in Appendix~\ref{app:primitives}.

\heading{Report.}
In round $t$, the server uses the random value $v$ (obtained from the setup) to select 
  the set of clients $S_t \subseteq [N]$ of size $n_t$ by running
  $\ChooseSet(v, t, n_t, N)$.
It then establishes the graph $G_t \leftarrow \GenGraph(v, t, S_t)$
  as specified in Figure~\ref{fig:gengraph}.
We denote the neighbors of $i$ as $A_t(i) \subseteq S_t$.
The server asks each client $i \in S_t$ to send a message consisting of
  the following three things:
  \begin{myenumerate2}  
    \item $Vec_i = \small \vec{x}_{i,t} + \sum_{j \in A_t(i)} \pm
    \PRG(h_{i,j,t}) + \PRG(m_{i,t})$,
  where $h_{i,j,t}$ is computed as $\PRF(r_{i,j}, t)$ and $r_{i,j}$ is derived 
  from the key directory by computing $(g^{a_j})^{a_i}$; $m_{i,t}$ is freshly 
  generated in round $t$ by client $i$.

    \item $L$ symmetric ciphertexts: $\SymEnc(k_{i,u}, m_{i,u,t})$ for all $u\in \comset$, 
    where $m_{i,u,t}$ is the share of $m_{i,t}$ 
    meant for $u$ (i.e., $\Share(m_{i,t}, \ell, L) \rightarrow \{m_{i,u,t}\}_{u\in \comset}$),
    and $k_{i,u}$ is the symmetric encryption key shared between client $i$ and \com $u$
    (they can derive $k_{i,u}$ from the PKI); 

    \item $|A_t(i)|$ ElGamal ciphertexts: $ \AsymEnc (PK, h_{i,j,t})$ for all $j \in A_t(i)$.
  \end{myenumerate2}

The above way of using symmetric encryption for individual masks
  and public-key encryption for pairwise masks is 
  for balancing computation and communication in practice. 
Technically, one can also encrypt the shares of $h_{i,j,t}$ with symmetric
  authenticated encryption as well (eliminating public-key operations),
  but it increases client communication---for each client, 
  the number of ciphertexts appended to the vector is $ |A(i)|\cdot L$. 
This is, for example, 1,600 when $L$ and $|A(i)|$ are both 40.
On the other hand, if one encrypts both the pairwise and individual masks 
  using only public-key encryption, then the number of expensive public key
  operations for reconstructing secrets is proportional to $n_t$;
  whereas it is only proportional to the number of dropouts in our proposed
  approach.
In practice, the number of dropouts is much smaller than $n_t$, hence the savings.

\heading{Cross-check.}
The server needs to recover $m_{i,t}$ for online clients, 
  and $h_{i,j,t}$ for clients that drop out.
To do so, the server labels the clients as ``offline'' or ``online'' and
  asks the \coms to recover the corresponding masks. 
For BBGLR, we described how this step involves most clients during
  the fault recovery process and highlighted an issue where a malicious 
  server can send inconsistent labels to clients and recover both 
  the pairwise mask and individual mask for some target client (\S\ref{s:bbglr}).
\sys{} also needs to handle this type of attack (the server tells 
  some of honest decryptors to decrypt $m_{i,t}$ and 
  other honest decryptors to decrypt $h_{i,j,t}$, and utilizes the malicious decryptors
  to reconstruct both), but it only 
  needs to involve \coms. 
In detail, each \com signs the online/offline labels of the $n_t$ clients
  (each client can only be labeled either offline or online),
  and sends them to the other \coms (via the server). 
Each \com checks it received $2 L / 3 $ or more valid signed
  labels (recall from \S\ref{s:threat} that $\delta_D + \eta_D < 1/3$).
If so, each \com further checks that:
\begin{myenumerate2}
    \item The number of online clients is at least $(1-\delta) n_t$; 
    \item All the online clients in the graph are connected;
    \item Each online client $i$ has at least $k$ online neighbors\footnote{This can be efficiently done in an inverse way of checking how many offline neighbors that each online client has, assuming dropout rate is small.}, such that $\eta^k < 2^{-\kappa}$
          ($\eta$ and $\kappa$ are defined as in \S\ref{s:threat}).
\end{myenumerate2}
If any of the above checks fail, the \com aborts. 
This step ensures either all the honest \coms agree on a valid offline/online label assignment 
  and consequently the server gets the result, or the honest decryptors abort and the server gets nothing.

\heading{Reconstruction.}
The server collects all the ciphertexts to be decrypted:
  the ciphertexts of $m_{i,u,t}$ (symmetric encryption) for the online clients, 
  and the ciphertexts of $h_{i,j,t}$ (public-key encryption) for the offline 
  clients. 
Then the server sends the ciphertexts to all the \coms who perform
  either a symmetric decryption or the threshold ElGamal decryption according to their agreed-upon labels.

The choice of using \coms{} to check the graph and reconstruct all the secrets is based on 
  an important observation in federated learning: the number of clients involved in one round, $n_t$,
  is much smaller than the input vector length~\cite{kairouz21advances}. 
Therefore, the asymptotic costs at a \com{} (which are proportional to $n_t$)
  are actually smaller than the size of an input weight vector.

\subsection{Malicious labeling across rounds}
The server, controlled by a malicious adversary, can ask for the decryption of $h_{i,j,t}$ in round $t$,
  and then in some other round $t'$, the server can ask for the decryption of $m_{i,t}$ (but not $m_{i,t'}$,
  if the server does not care about obtaining a result in round $t'$).
This allows the server to recover $\vec{x}_{i,t}$ in the clear.
To prevent this attack, honest \coms need to know the round for which a 
  ciphertext is sent.
For symmetric ciphertext, the client appends the round number $t$ to
  the plaintext (e.g., $m_{i,u,t} || t$) and uses authenticated encryption; for public-key ciphertexts, the client appends $t$ to the ciphertext $c$
  and signs the tuple $(c, t)$ (the verification key is in the PKI). 
Note that a malicious adversary can still fool enough honest \coms{} 
  into thinking it is round $t$ while it is in fact $t'$.
To prevent this, \coms{} also include the round number 
  in the online/offline labels and sign them. 
The cross-check (\S\ref{s:detail:collection}) guarantees that the \coms{} agree on the round number.

\subsection{Load balancing across \coms}\label{s:load-balancing}
In each summation, a client who is not a \com only sends a single vector.
This is nearly optimal since even if the aggregation is non-private the client 
  has to send the vector (but without the additional small ciphertexts).
The \coms, however, have additional computation and communication 
  in order to help with the result reconstruction.  
This causes a load imbalance in the system and could be unfair since a client 
  selected to be a \com has to do more work than regular clients.

In \sys{}, the \com responsibility shifts across time.
Every $R$ rounds, the current \coms transfer 
  their shares of $SK$ to a new set of randomly selected clients who 
  serve as the new \coms.
To ensure security, the shares of $SK$ have to be modified in a particular 
  way during the transition, as otherwise the adversary may control some
  malicious \coms before the transition and some malicious \coms after the transition, 
  and thus may obtain enough shares to reconstruct $SK$.
We address this by relying on prior proactive secret sharing 
  techniques~\cite{Herzberg95proactive, gennaro08threshold, krishna19churp};
  they additionally enable \sys{} to change the number of \coms and the 
  threshold as needed.
In Appendix~\ref{app:transfer}, we provide details of the transition protocol used in \sys{}.

A final clarification is that \coms{} who dropped out (e.g., due to power loss) 
  at one round can come back later and participate in another round (e.g.,
  when power is resumed).
The decryption always succeeds since we require that less than $1/3$ deryptors are dropped out or malicious at any step (\S\ref{s:security}). 
The secret key transition is purely for system load balancing---neither 
  dropout resilience nor security relies on the parameter $R$.

\subsection{Considerations in federated learning}
A recent work~\cite{pasquini22eluding} describes an 
  attack in the composition of federated learning and
  secure aggregation.
The idea is that the server can elude secure aggregation 
  by sending clients inconsistent models. 
For example, the server sends to client 1 model $M_1$, to client 2 model $M_2$,
  and to client 3 model $M_3$.
Each of the clients then runs the local training on the provided model.
The server chooses the models it sends to clients 1 and 2 in a special way such 
  that after clients 1 and 2 train their local models, their local weights
  will cancel out when added.
Consequently, the server will get the model weights of client 3.
The proposed defense, which works in our context without incurring any overhead, is for clients to append the hash 
of the model they receive in a given round to their 
PRG seed for that round: $\PRG(h_{i, j, t} || Hash(M))$, where $M$ is the model received from the server.
If all the clients receive the same model, the pairwise masks cancel out; otherwise, they do not.

\section{Parameter Selection and Security Analysis}\label{s:security}

The parameters of \sys{} include:
\begin{myitemize2}
    \item System parameters $N, T$ and the number of clients $n_t$ chosen in round $t\in [T]$; 
    \item Threat model parameters $\delta_D, \delta, \eta$ which are given, 
      and $\eta_{S_t}, \eta_D$ which depend on $\eta$ (their relation is
      summarized in Section~\ref{s:threat} and fully derived in
      Appendix~\ref{app:failure-model}).
    \item Security parameter $\kappa$, and the parameters that relates to security: graph generation parameter $\epsilon$, and the number of selected decryptors $L$. 
\end{myitemize2}

We discuss these parameters in detail below and state our formal lemmas
  with respect to them.

\subsection{Security of setup phase}\label{s:param:setup}
\begin{figure} 
 \begin{tcolorbox}[enhanced, boxsep=1mm, left= 0mm, right=0.5mm, title={\textbf{\footnotesize Functionality $\Fsetup$}}]
    \linespread{1.4}
    \footnotesize
    
    Parties: clients $1, \ldots, N$ and a server.
    
    \begin{myitemize2} 

        \item $\Fsetup$ samples $v \xleftarrow{\$} \{0,1\}^\lambda$.
        
        \ifthenelse{\boolean{longver}}
        {
            \item $\Fsetup$ samples a secret key and public key pair $(SK, PK)$.
            
            \Comment{When the public-key cryptosystem is instantiated by ElGamal, then $SK$ is $s \xleftarrow{\$} \Z_{q}$ and $PK=g^s$.}
        }{
            \item $\Fsetup$ samples $a \leftarrow \Z_q$ and computes $PK = g^a$.
        }
               
        \item $\Fsetup$ asks the adversary $\A$ whether it should continue or not. 
        If $\A$ replies with \texttt{abort}, $\Fsetup$ sends \texttt{abort} to all honest parties; 
        if $\A$ replies with \texttt{continue}, $\Fsetup$ sends $v$ and $PK$ to all the parties.
    \end{myitemize2}
    \end{tcolorbox}
\caption{Ideal functionality for the setup phase.} 
\label{fig:Fsetup}
\end{figure}

Let $\delta_D$ upper bound the fraction of \coms{} that 
  drop out during the setup phase; 
  note that in Section~\ref{s:threat} we let $\delta_D$ 
  upper bound the dropouts in one aggregation round and
  for simplicity here we use the same notation.
\sys's DKG requires that $\delta_D + \eta_D < 1/3$. 
Note that $\eta_D$ in fact depends on $\eta$, $L$ and $N$, but we will give the theorems using $\eta_D$
  and discuss how to choose $L$ to guarantee a desired $\eta_D$ in Appendix~\ref{app:param}.

\begin{theorem}[Security of setup phase]\label{thm:security-setup}
Assume that a PKI 
and a trusted source of randomness exist, and that the DDH assumption holds.
Let the dropout rate of decryptors in the setup phase be bounded by
  $\delta_D$.
If $\delta_D + \eta_D < 1/3$, 
then under the communication model defined in Section~\ref{s:comm-model}, 
  protocol $\Pi_{\text{setup}}$ (Fig.~\ref{fig:protocol-setup}) securely realizes functionality
  $\Fsetup$ (Fig.~\ref{fig:Fsetup}) 
  in the presence of a malicious adversary controlling the server and $\eta$ fraction of 
  the $N$ clients. 
\end{theorem}

\subsection{Security of collection phase}\label{s:param:collection}

First, we need to guarantee that each graph $G_t$,
  even \emph{after} removing the vertices corresponding to
  the $\delta+\eta$ fraction of dropout and 
  malicious clients, is still connected. 
This imposes a requirement on $\epsilon$, which we state in Lemma~\ref{lemma:graph-connectivity}. 
For simplicity, we omit the exact formula for the lower bound of $\epsilon$ and defer the details to Appendix~\ref{app:param}.

\begin{lemma}[Graph connectivity]\label{lemma:graph-connectivity}
Given a security parameter $\kappa$, and threat model parameters $\delta, \eta$ (\S\ref{s:threat}).
Let $G$ be a random graph $G(n,\epsilon)$. 
Let $\mathcal{C}, \mathcal{O}$ be two random subsets of nodes in $G$
  where $|\mathcal{O}| \le \delta n$ and $|\mathcal{C}| \le \eta  n$ 
  ($\mathcal{O}$ stands for dropout set and $\mathcal{C}$ stands for malicious set).
Let $\tilde{G}$ be the graph with nodes in $\mathcal{C}$ and $\mathcal{O}$ and the
  associated edges removed from $G$.
There exists $\epsilon^*$
  such that for all $\epsilon \ge \epsilon^*$, 
 $\tilde{G}$ is connected except with probability $2^{-\kappa}$.
\end{lemma}

Secondly, we require $2\delta_D + \eta_D < 1/3$ to ensure that
  all online honest \coms{} reach an agreement in the cross-check step
  and the reconstruction is successful.
Note that the \coms{} in the setup phase who dropped out ($\delta_D$ fraction)
  will not have the share of $SK$;
  while the clients who drop out during a round (another $\delta_D$ fraction) in the collection phase can come back at another round,
  hence we have the above inequality.

\subsection{Main theorems}\label{s:formal-thms}

\ifthenelse{\boolean{longver}}
{The full protocol, denoted as $\Phi_T$, is the sequential execution of $\Pi_{\text{setup}}$ (Fig.~\ref{fig:protocol-setup}) followed by a $T$-round $\Pi_{\text{sum}}$ (Fig.~\ref{fig:collection}). 
We now give formal statements for the properties of \sys,
and defer the proof to Appendix~\ref{app:proofs}.
Note that as we see from the ideal functionality $\mathcal{F}_{\text{mal}}$ (Fig.~\ref{fig:F-ideal-mal-weak}), when the server is corrupted, the sum result in round $t$ is not determined by the actual dropout set $\mathcal{O}_t$, but instead a set $M_t$ chosen by the adversary (see details in Appendix~\ref{app:proof-security-flamingo}). }
{The full protocol, denoted as $\Phi_T$, is the sequential execution of 
$\Pi_{\text{setup}}$ followed by a $T$-round $\Pi_{\text{sum}}$. 
We now give formal statements for the properties of \sys{}; the proof 
is given in Appendix D in the full version~\cite{ma23flamingo}.}

\begin{figure}
 \begin{tcolorbox}[enhanced, boxsep=1mm, left= 0mm, right=0.5mm, title={\textbf{\footnotesize Functionality $\Fmal$}}]
    \linespread{1.4}
    \footnotesize
    
    Parties: clients $1, \ldots, N$ and a server.
    
    Parameters: corrupted rate $\eta$, dropout rate $\delta$, number of per-round participating clients $n$. 
    
    \begin{myitemize2}

        \item $\Fmal$ receives from a malicious adversary $\Adv$ a set of corrupted parties,
                denoted as $\mathcal{C} \subset[N]$, where $|\mathcal{C}|/N \le \eta$.
        \item For each round $t\in [T]$:
            \begin{enumerate}
                \item  $\Fmal$ receives a sized-$n$ random subset $S_t\subset [N]$ and a dropout set $\mathcal{O}_t \subset S_t$, 
                where $|\mathcal{O}_t|/|S_t| \le \delta$, and $|\mathcal{C}|/|S_t| \le \eta$,
                and inputs $\vec{x}_{i,t}$ from client $i \in S_t \backslash (\mathcal{O}_t \cup \mathcal{C})$.

                \item $\Fmal$ sends $S_t$ and $\mathcal{O}_t$ to $\Adv$, and asks $\Adv$ for a set $M_t$: if $\Adv$ replies with $M_t \subseteq S_t\backslash \mathcal{O}_t$ such that $|M_t|/|S_t| \ge 1-\delta$, then $\Fmal$ computes $\vec{y}_t = \sum_{i\in M_t \backslash \mathcal{C}} \vec{x}_{i,t}$ and continues; otherwise $\Fmal$ sends \texttt{abort} to all the honest parties.

                \item Depending on whether the server is corrupted by $\Adv$:
                \begin{myitemize2}
                    \item If the server is corrupted by $\Adv$, then $\Fmal$ outputs $\vec{y}_t$ to all the parties corrupted by $\Adv$.
                    \item If the server is not corrupted by $\Adv$, then $\Fmal$ asks $\Adv$ for a shift $\vec{a}_t$ and outputs $\vec{y}_t + \vec{a}_t$ to the server.
                \end{myitemize2}
                \end{enumerate}
  
    \end{myitemize2}
    \end{tcolorbox}
\caption{Ideal functionality for Flamingo.}
\label{fig:F-ideal-mal-weak}
\end{figure}

\begin{theorem}[Dropout resilience of $\Phi_T$]\label{thm:dropout-resilience}
Let $\delta, \delta_D, \eta, \eta_D$ be threat model parameters as defined (\S\ref{s:param:setup},\S\ref{s:param:collection}). 
If $2\delta_D + \eta_D < 1/3$, 
then protocol $\Phi_T$ satisfies dropout resilience: 
when all parties follow the protocol $\Phi_T$,
for every round $t \in [T]$, and given a set of dropout clients $\mathcal{O}_t$ in the report step
 where $|\mathcal{O}_t| / |S_t| \le \delta$, 
protocol $\Phi_T$ terminates and outputs $\sum_{i\in S_t \backslash \mathcal{O}} \vec{x}_{i, t}$, except probability $2^{-\kappa}$.
\end{theorem}

\begin{theorem}[Security of $\Phi_T$]\label{thm:security-main}
Let the security parameter be $\kappa$.
Let $\delta, \delta_D, \eta, \eta_D$ be threat model parameters as defined (\S\ref{s:param:setup},\S\ref{s:param:collection}). 
Let $\epsilon$ be the graph generation parameter (Fig.~\ref{fig:gengraph}). 
Let $N$ be the total number of clients and $n$ be the number of clients for summation in each round. 
Assuming the existence of a PKI, a trusted source of initial randomness, a PRG, a PRF, an asymmetric encryption $\AsymEnc$, a symmetric authenticated encryption $\SymEnc$, and a signature scheme, 
if $2\delta_D + \eta_D < 1/3$ and $\epsilon \ge \epsilon^*(\kappa)$ (Lemma~\ref{lemma:graph-connectivity}), then under the communication model defined in Section~\ref{s:comm-model}, protocol $\Phi_T$ securely realizes the ideal functionality $\Fmal$ given in Figure~\ref{fig:F-ideal-mal-weak} in the presence of a static malicious adversary controlling $\eta$ fraction of $N$ clients (and the server),\footnote{We assume that in each round, the corrupted fraction of $n$ clients is also $\eta$; see Section~\ref{s:pb}.}  except with probability at most $T n \cdot 2^{-\kappa+1}$.  
\end{theorem}

\ifthenelse{\boolean{longver}}
{Additionally, if the asymmetric encryption is instantiated by ElGamal cryptosystem,
  the security can be based on the DDH assumption and the other primitives stated above.}
{}

The final complication is how to choose $L$ to ensure $2\delta_D + \eta_D < 1/3$ holds;
  note that $\eta_D$ depends on $\eta$ and $L$. 
One can choose $L$ to be $N$ but it does not give an efficient protocol;
on the other hand, choosing a small $L$ may result in
  all the \coms{} being malicious. 
\ifthenelse{\boolean{longver}}
{In Appendix~\ref{app:param}, we give a lower bound of $L$ to ensure a desired $\eta_D$ (w.h.p.), given $N, \eta $, and $\delta_D$. }
{In Appendix~\ref{app:param} of the full version~\cite{ma23flamingo}, we give a lower bound of $L$ to ensure a desired $\eta_D$ (w.h.p.), given $N, \eta $, and $\delta_D$.}

\section{Implementation}\label{s:impl}
\label{s:crypto-instantiation}

We implement \sys{} in 1.7K lines and BBGLR in 1.1K lines of Python.
For $\PRG$, we use AES in counter mode, for authenticated encryption
  we use AES-GCM, 
  and signatures use ECDSA over curve P-256.
Our code is available at \href{https://github.com/eniac/flamingo}{\color{black}{\texttt{https://github.com/eniac/flamingo}}}.
\begin{figure*}
\footnotesize
 \renewcommand{\arraystretch}{1.5} 
\begin{tabular*}{\textwidth}{@{\extracolsep{\fill}}l l c c c | c c c}
 & & \multicolumn{3}{ c }{\textbf{BBGLR}} & \multicolumn{3}{ c }{\textbf{\sys}} \\
 \textbf{Phase} & & \textbf{Steps} & \textbf{Server} & \textbf{Client} &
                    \textbf{Steps} &\textbf{Server} & \textbf{Client} \\
 \cline{1-8}
  Setup & & --- & --- & --- & 4 & $O(L^3)$ & $O(L^2)$ \\
 \cline{1-8}
  \multirow{3}*{$T$ sums} & Round setup & $3T$ & $O(TA n_t)$ & $O(TA)$ & --- & --- & --- \\
 \cline{2-8}
  & \multirow{2}*{Collection} & \multirow{2}*{$3T$} & \multirow{2}*{$O( T n_t
(d+ A ))$} & \multirow{2}*{$O( T(d + A))$} & \multirow{2}*{$3T$} & 
\multirow{2}*{$O(T n_t (d+L +  A))$} & Regular client: $O(T(d + A))$ \\
                                     &&&&&&& \Coms: $O(T(L + \delta A n_t  + (1-\delta)n_t ))$ \\
 \bottomrule
\end{tabular*}
\caption{Communication complexity and number of steps (client-server round-trips) of \sys and BBGLR 
    for $T$ rounds of aggregation. 
$N$ is the total number of clients and $n_t$ is the number of clients chosen 
  to participate in round $t$. 
The number of \coms is $L$, and the dropout rate of clients in $S_t$ is $\delta$.
Let $A$ be the upper bound on the number of neighbors of a client, and let
$d$ be the dimension of client's input vector. }

\label{fig:asym-comm}
\end{figure*}

\heading{Distributed decryption.}\label{s:distributed-decryption}
We build the distributed decryption scheme discussed in Section~\ref{s:intuition} as follows.
We use ElGamal encryption to instantiate the asymmetric encryption. It consists of three algorithms $(\AsymGen, \AsymEnc, \AsymDec)$.
$\AsymGen$ outputs a secret and public key pair $SK \in_R \mathbb{Z}_q$ and $PK := g^{SK} \in \mathbb{G}$.
$\AsymEnc$ takes in $PK$ and plaintext $h \in \mathbb{G}$, and outputs 
  ciphertext $(c_0 , c_1) := (g^w, h \cdot PK^w)$, where 
  $w \in_R \mathbb{Z}_q$ is the encryption randomness.
$\AsymDec$ takes in $SK$ and ciphertext $(c_0, c_1)$ and outputs $h = (c_0^{SK})^{-1}\cdot c_1$.

In threshold decryption~\cite{gennaro08threshold, desmedt89threshold,
  shoup02securing}, the secret key $SK$ is shared among $L$ parties such that
  each party $u\in [L]$ holds a share $s_u$, but no single entity knows
  $SK$, i.e., $(s_1, \ldots, s_L) \leftarrow \Share(SK, \ell, L)$.
Suppose Alice wants to decrypt the ciphertext $(c_0, c_1)$ using the secret-shared $SK$.\@
To do so, Alice sends $c_0$ to each party in $[L]$, and gets back 
  $c_0^{s_u}$ for $u \in U \subseteq [L]$.
If $|U| > \ell$, Alice can compute from $U$ a set of combination 
  coefficients $\{\beta_u\}_{u\in U}$, and
\[c_0^{SK} = \prod_{u \in U} (c_0^{s_u})^{\beta_u }. \]

Given $c_0^{SK}$, Alice can get the plaintext $h = (c_0^{SK})^{-1} \cdot c_1$.
Three crucial aspects of this protocol are that: (1) $SK$ is never reconstructed; 
  (2) the decryption is dropout resilient (Alice can obtain $h$ as long as more than
  $\ell$ parties respond); (3) it is non-interactive: Alice communicates 
  with each party exactly once.

We implement ElGamal over elliptic curve 
  group $\mathbb{G}$ and we use curve P-256.
To map the output of $\PRF(r_{i,j}, t)$, which is a binary string, 
  to $\mathbb{G}$, we first hash it to an element in the field of 
  the curve using the \emph{hash-to-field} algorithm from the 
  IETF draft~\cite[\S5]{hashec}.
We then use the SSWU algorithm~\cite{swu,sswu} to map the field 
  element to a curve point $P \in \mathbb{G}$.
A client will encrypt $P$ with ElGamal,
  and then hash $P$ with SHA256 to obtain $h_{i,j,t}$---the input
  to the pairwise mask's $\PRG$.
When the server decrypts the ElGamal ciphertexts and obtains $P$, it 
  uses SHA256 on $P$ to obtain the same value of $h_{i,j,t}$.

\heading{Optimizations.}
In \sys's reconstruction step, we let the server do reconstruction using partial shares.
That is, if the interpolation threshold is 15, and the server collected shares from 20 \coms, 
  it will only use 15 of them and ignore the remaining 5 shares.
Furthermore, as we only have a single set of \coms, when the server collects shares from $U \subseteq \comset$,
  it computes a single set of interpolation coefficients from $U$ and uses it to do linear combinations
  on all the shares.  
This linear combination of all the shares can be done in parallel.
In contrast, BBGLR requires the server to compute different sets of interpolation coefficients to reconstruct the pairwise
  masks (one set for each client).
  
\heading{Simulation framework.}
We integrade all of \sys's code into 
  ABIDES~\cite{byrd20abides, abides}, which is an open-source 
  high-fidelity simulator designed for AI research in financial markets (e.g., stock exchanges).
ABIDES is a great fit as it supports tens 
  of thousands of clients interacting with a server to 
  facilitate transactions (and in our case to compute sums). 
It also supports configurable pairwise network latencies.

\section{Asymptotic Costs}\label{s:cost}
\label{s:delay}

An advantage of \sys over BBGLR is that \sys requires fewer round trips for one summation.
BBGLR requires six round trips for one summation; in contrast, \sys{} requires three: 
  the report step involves all clients selected in a round, and for the remaining two 
  steps the server contacts the \coms.
Meanwhile, the number of round trips has a significant impact on the overall 
  runtime of the aggregation, as we show experimentally in Section~\ref{s:perf}.
The reasons for this are two-fold: (1) the latency of an RTT over the wide area network 
  is on the order of tens of milliseconds; and (2) the server has to wait for enough clients in each step to send their messages, so tail latency plays a big role.
Depending on the setting, client message arrival can vary widely, 
  with some clients potentially being mobile devices on slow networks.

In addition to fewer round trips, \sys{} is expected to wait for less time during the reconstruction step. 
This is because the server in BBGLR has to wait for the vast majority of clients to 
  respond in order to reconstruct secrets for dropout clients, while the server in 
  \sys{} only needs responses from $1/3$ of the \coms.

\heading{Client and server costs.}
Figure~\ref{fig:asym-comm} compares \sys's communication costs
  with those of BBGLR.  
In short, the total asymptotic cost for $T$ aggregation rounds between BBGLR and \sys{} does not vary much;
  but the number of round trips differ much.  
The computation cost is analyzed as follows.
For the setup phase, if a client is a \com, it has $O(L^2)$ computation for both DKG,
and secret key transfer. 
In the collection phase at round $t$, each client computes $O(A + L)$ 
  encryptions (assuming that for all clients $i$,  $A(i) \leq A$). 
If a client is a \com, it additionally has a 
  $O(\delta A n_t  + (1-\delta)n_t + \epsilon n_t^2)$ computation cost.

\section{Experimental Evaluation}\label{s:eval}

In this section we answer the following questions:
\begin{myitemize2}
    \item What are \sys's concrete server and client costs,
      and how long does it take \sys{} to complete one and multiple rounds
      of aggregation?
    \item Can \sys{} train a realistic neural network?
    \item How does \sys{} compare to the state of the art in terms of
     the quality of the results and running time? 
\end{myitemize2}

\noindent We implement the following baselines:

\heading{Non-private baseline.}
We implement a server that simply sums up the inputs it 
receives from clients. During a particular round, each of the clients sends a vector to the server. These vectors are in the clear, and may be any sort of value (e.g. floating points), unlike \sys, which requires data to be masked positive integers. 
The server tolerates dropouts, as Flamingo does, and aggregates only the vectors from clients who respond before the timeout.

\heading{BBGLR.}
For BGGLR, we implement Algorithm 3 in their paper~\cite{bell20secure} with a 
  slight variation that significantly improves BBGLR's running time, 
  although that might introduce security issues (i.e., we are making this baseline's 
  performance better than it is in reality, even if it means it is no longer secure). 
Our specific change is that we allow clients to drop out during the graph 
  establishment procedure and simply construct a graph with the clients that
  respond in time. BBGLR (originally) requires that no client drops out during graph 
  establishment to ensure that the graph is connected.
Without this modification, BBGLR's server has to wait for all the 
  clients to respond and is severely impacted by the long tail of the 
  client response distribution---which makes our experiments take 
  an unreasonable amount of time.

\subsection{Experimental environment}\label{s:simulate}
 
Prior works focus their evaluation on the server's costs.
While this is an important aspect (and we also evaluate it), a key 
  contributor to the end-to-end completion time of the aggregation (and 
  of the federated learning training) is the number of round-trips between clients and the server.
This is especially true for geodistributed clients.

To faithfully evaluate real network conditions, we run the ABIDES simulator~\cite{abides} on a server with 40 Intel Xeon E5-2660 v3 (2.60GHz) CPUs and
  200 GB DDR4 memory.
Note that in practice, client devices are usually less powerful than the experiment machine.
ABIDES supports the cubic network delay model~\cite{ha08cubic}: the latency 
  consists of a base delay (a range), plus a jitter that controls the
  percentage of messages that arrive within a given time (i.e., the shape of 
  the delay distribution tail).
We set the base delay to the ``global'' setting in ABIDES's default parameters 
  (the range is 21 microseconds to 53 milliseconds), and use the default 
  parameters for the jitter.

Both \sys and BBGLR work in steps (as defined in \S\ref{s:threat}).
We define a \textit{waiting} period for each step of the protocol.
During the waiting period, the server receives messages from clients and puts 
  the received messages in a message pool.
When the waiting period is over, a \emph{timeout} is triggered 
  and the server processes the messages in the pool, and proceeds to the next 
  step.
The reason that we do not let the server send and receive messages 
  at the same time is that, in some steps (in both BBGLR and \sys), the results 
  sent to the clients depend on all the received messages and cannot be 
  processed in a streaming fashion.
For example, the server must decide on the set of online clients before sending the request to reconstruct the shares.

\subsection{Secure aggregation costs and completion time}\label{s:perf}
This section provides microbenchmarks for summation tasks performed by BBGLR and \sys{}.
Client inputs are 16K-dimensional vectors with 32-bit entries.
\ifthenelse{\boolean{longver}}{%
For parameters, unless specified differently later, we set $N$ to 10K and
  the number of decryptors to 60 in \sys{}; and set
  the number of neighbors to $4\log n_t$ for both \sys{} and BBGLR 
  (for BBGLR, this choice satisfies the constraints in Lemma 4.7~\cite{bell20secure}). 
In Figures~\ref{fig:comm} and~\ref{fig:microbenchmark-comp}, ``keyad'' is the step for
  exchanging keys, ``graph'' is the step for clients to send their choices of
  neighbors, ``share'' is the step for clients to shares their secrets
  to their neighbors (marked as 1--8 in their Algorithm 3). 
The steps ``report'', ``check'' and ``recon'' in \sys{} are 
  described in Section~\ref{s:detail:collection}; in BBGLR,
  these steps correspond to the last three round trips in a summation marked as 8--14 in their Algorithm 3.
}{
}

\begin{figure}
\begin{subfigure}[b]{0.48\linewidth}
\centering
 \includegraphics[width=\linewidth]{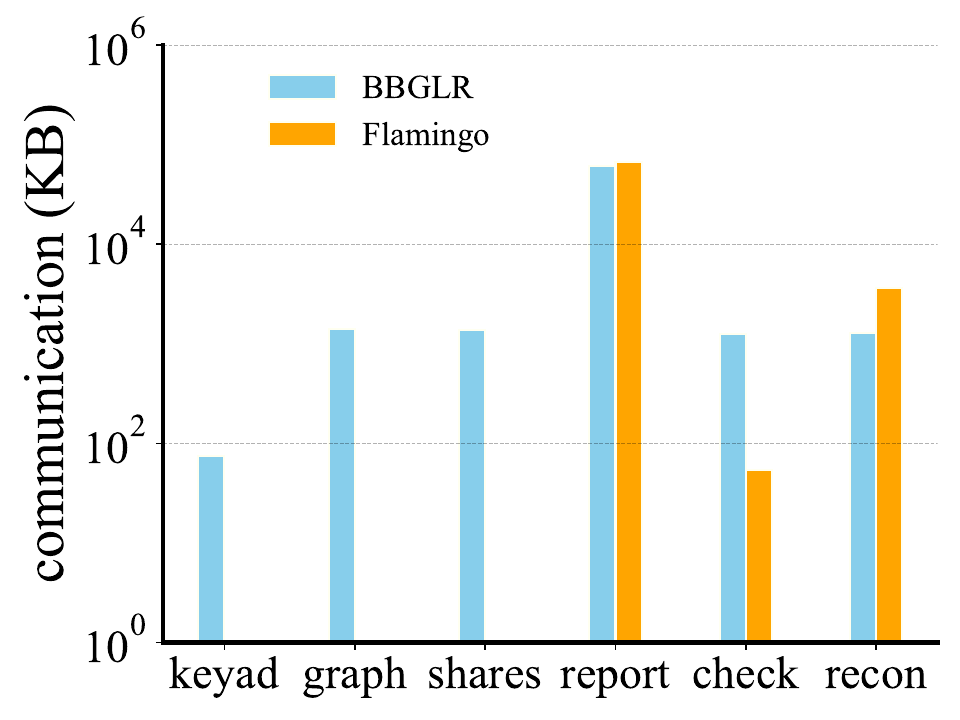}
 \caption{Server communication}
    \label{fig:multisum-acc}
\end{subfigure}
 \begin{subfigure}[b]{0.48\linewidth}
 \centering
 \includegraphics[width=\linewidth]{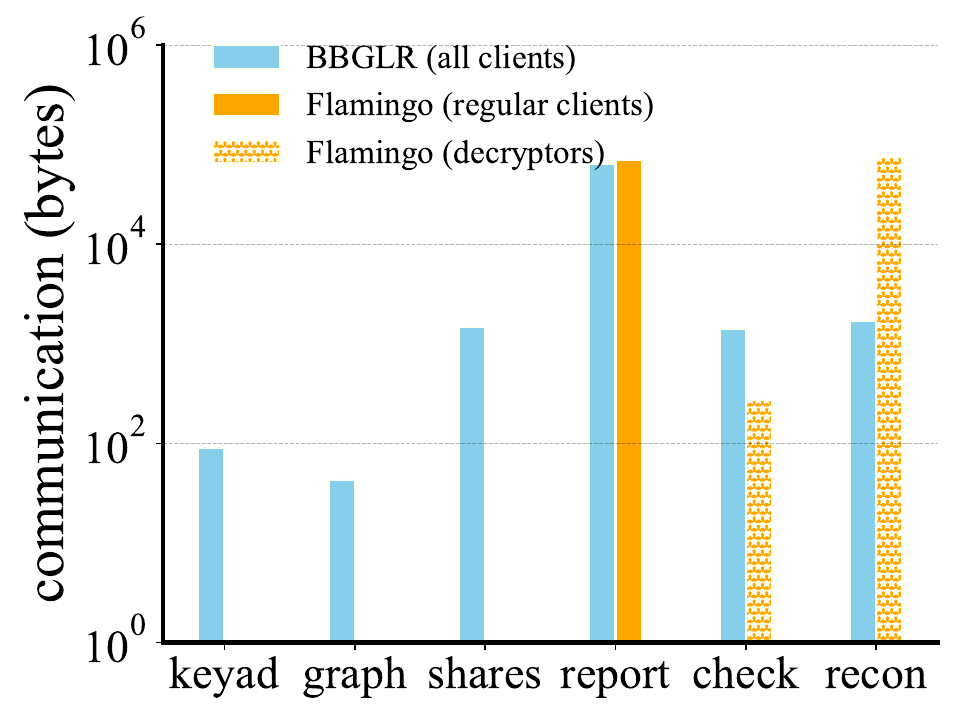}
 \caption{Client communication}
    \label{fig:multisum-acc}
\end{subfigure}
  \caption{Communication costs for different steps in a single summation over 1K clients for \sys and BBGLR.
   }
 \label{fig:comm}
\end{figure}

\heading{Communication costs.}
Figure~\ref{fig:comm} gives the communication cost for 
  a single summation.
The total cost per aggregation round for BBGLR and \sys{} 
  are similar.
This is because \sys{}'s extra cost over BBGLR at the report step
  is roughly the message size that BBGLR has in their three-step setup; this is also reflected in the asymptotic cost analysis of Figure~\ref{fig:asym-comm}. 
In short, compared to BBGLR, \sys{} has fewer round trips with roughly the same total server communication. 
For clients, the story is more nuanced: each client 
has a slightly higher cost in \sys{} than in BBGLR during 
  the report step, as clients in \sys{} need to append ciphertexts to the vector.
However, in the reconstruction step, clients who are not \coms will not need to send or receive any messages. Each 
  \com incurs communication that is slightly larger than sending one input vector.
Note that the vector size only affects the report step.

\begin{figure}
{\footnotesize
\begin{tabular*}{\columnwidth}{@{\extracolsep{\fill}}l r r r r r r}
\textbf{CPU costs} & keyad & graph & share & report & check & recon \\
\toprule
\textbf{Server (sec)}  \\
BBGLR & 0.11 & 0.27  & 0.09  &  0.09  & 0.08  & 0.76  \\
Flamingo & --- & --- & --- & 0.24  & --- &  2.30 \\
\midrule
\textbf{Client (sec)}  \\
BBGLR & $<$0.01 & $<$0.01  & $<$0.01 & 0.02  & 0.01 & $<$0.01  \\
Flamingo \\
\quad Regular clients & ---  & ---  & --- & 0.22 & --- & --- \\
\quad \Coms & ---  & --- & --- & --- & 0.10 & 0.56 \\
\bottomrule
\end{tabular*}
}
\caption{Single-threaded microbenchmarks averaged over 10 runs for server and client computation for a single summation over 1K clients. ``$<$'' means less than. } 
\label{fig:microbenchmark-comp}
\end{figure}

\heading{Computation costs.}
We first microbenchmark a single summation with 
  1K clients, and set $\delta$ to 1\% (i.e., up to
  1\% of clients can drop out in any given step).
This value of $\delta$ results in a server waiting time 
  of 10 seconds.
Figure~\ref{fig:microbenchmark-comp} gives the results.
The report step in \sys has slightly larger server and client costs than BBGLR
because clients need to generate the graph ``on-the-fly''.
In BBGLR, the graph is already established in the 
  first three steps and stored for the report and reconstruction step.
For server reconstruction time, \sys is slightly more costly than BBGLR because of the additional elliptic curve operations. 
The main takeaway is that while \sys's computational costs are slightly higher than BBGLR, these
  additional costs have negligible impact on completion time owing to the much larger
  effect of network delay, as we show next.

\heading{Aggregation completion time.}
To showcase how waiting time $w$ affects dropouts (which consequently affects the
  sum accuracy), we use two waiting times, 5 seconds and 10 seconds.  
The runtime for an aggregation round depends on the timeout for each step,
  the simulated network delay, and the server and client computation time.
Figure~\ref{fig:multisum-time-long} and~\ref{fig:multisum-time-short} show
  the overall completion time across 10 rounds of aggregations.
A shorter waiting time makes the training faster, but it also means that there are more dropouts, 
  which in turn leads to more computation during reconstruction.
As a result, the overall runtime for the two cases are similar.
On 1K clients, \sys{} achieves a $3\times$ improvement over BBGLR; for \sys's cost we included its one-time setup 
  and one secret key transfer. 
If the key transfer is performed less frequently, the improvement will be more significant.

\begin{figure}
\begin{subfigure}[b]{0.45\linewidth}
 \includegraphics[width=\linewidth]{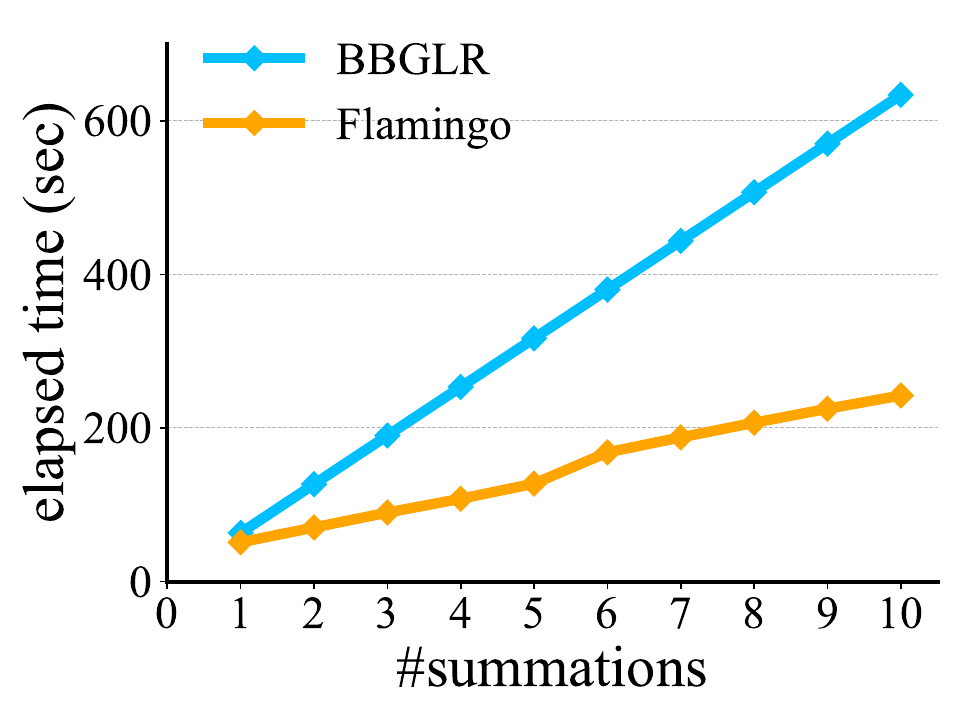}
 \caption{Runtime with $w=10$.}
    \label{fig:multisum-time-long}
\end{subfigure}
 \begin{subfigure}[b]{0.45\linewidth}
 \includegraphics[width=\linewidth]{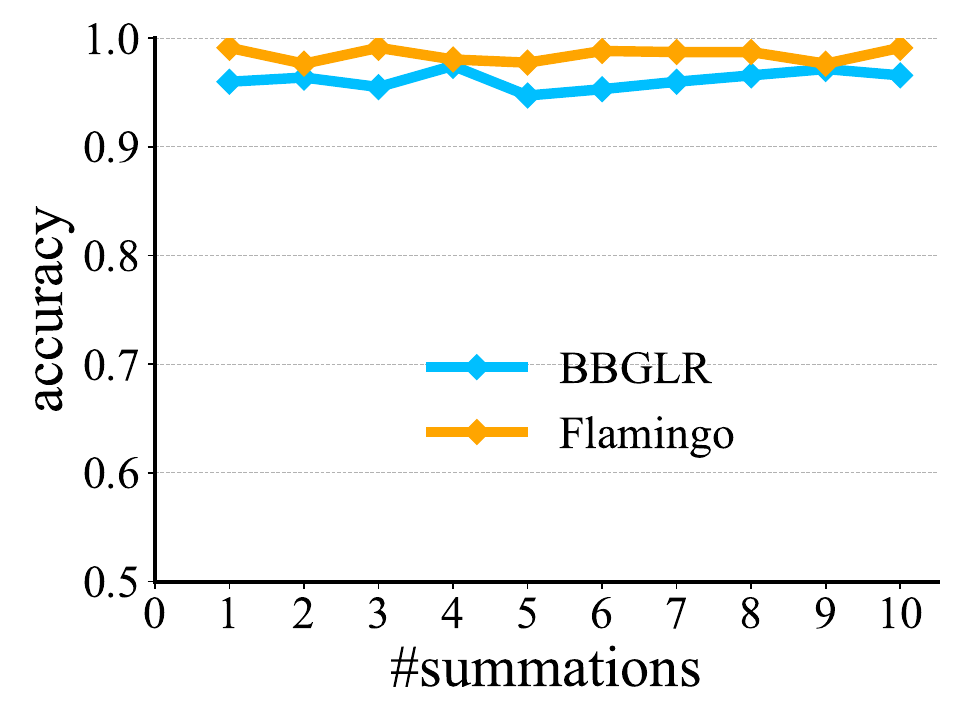}
 \caption{Sum accuracy $\tau$; $w=10$.}
    \label{fig:multisum-acc-long}
\end{subfigure}
    
\begin{subfigure}[b]{0.45\linewidth}
 \includegraphics[width=\linewidth]{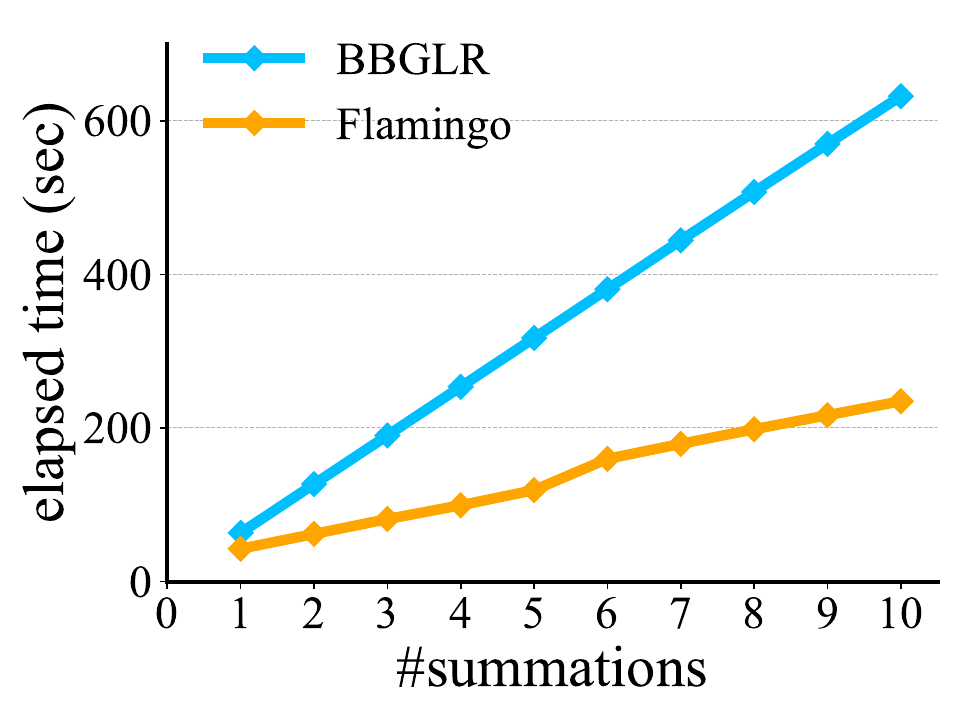}
 \caption{Runtime with $w=5$. }
    \label{fig:multisum-time-short}
\end{subfigure}
 \begin{subfigure}[b]{0.45\linewidth}
 \includegraphics[width=\linewidth]{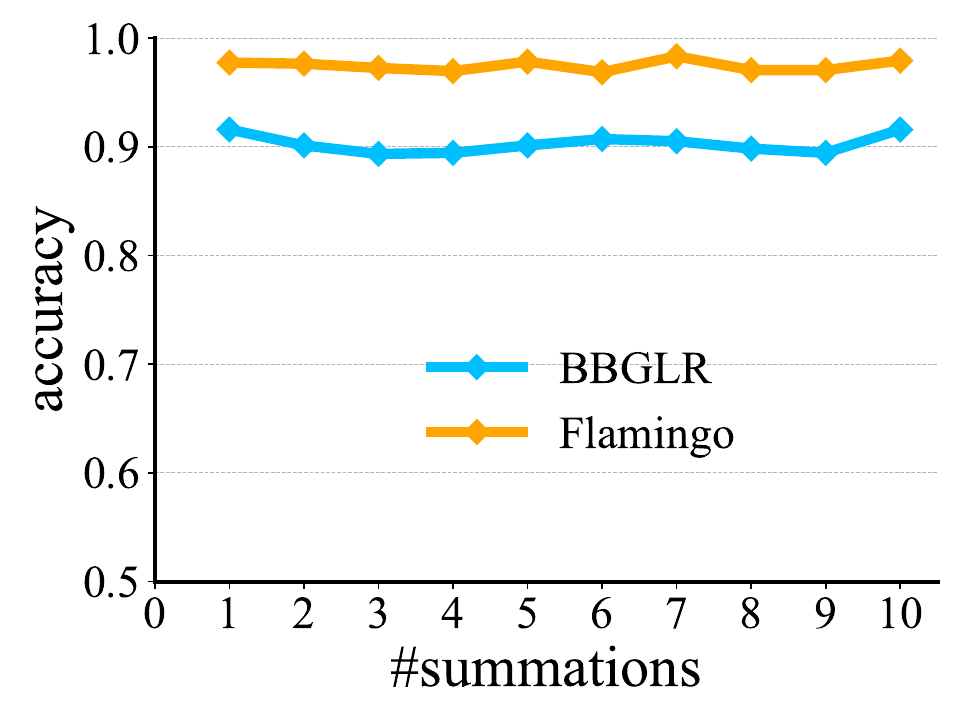}
 \caption{Sum accuracy $\tau$; $w=5$. }
    \label{fig:multisum-acc-short}
\end{subfigure}
  \caption{End-to-end completion time and accuracy of 10 secure aggregation rounds with 1K clients.
  The elapsed time is the finishing time of round $t$.
  For \sys{}, round 1 includes all of the costs of its 
  one-time setup, and between round 5 and 6
  \sys performs a secret key transfer.}
 \label{fig:multisum}
\end{figure}

The cost of the DKG procedure (part of the setup and which we also added to the first round in 
  Figure~\ref{fig:multisum}) is shown in Figure~\ref{fig:dkg}.
A complete DKG takes less than 10 seconds as the number of \coms is not large and we allow \com dropouts.
For 60 \coms, the local computation performed by each \com during the DKG is 2 seconds.

\begin{figure}
\centering
 \includegraphics[width=\linewidth]{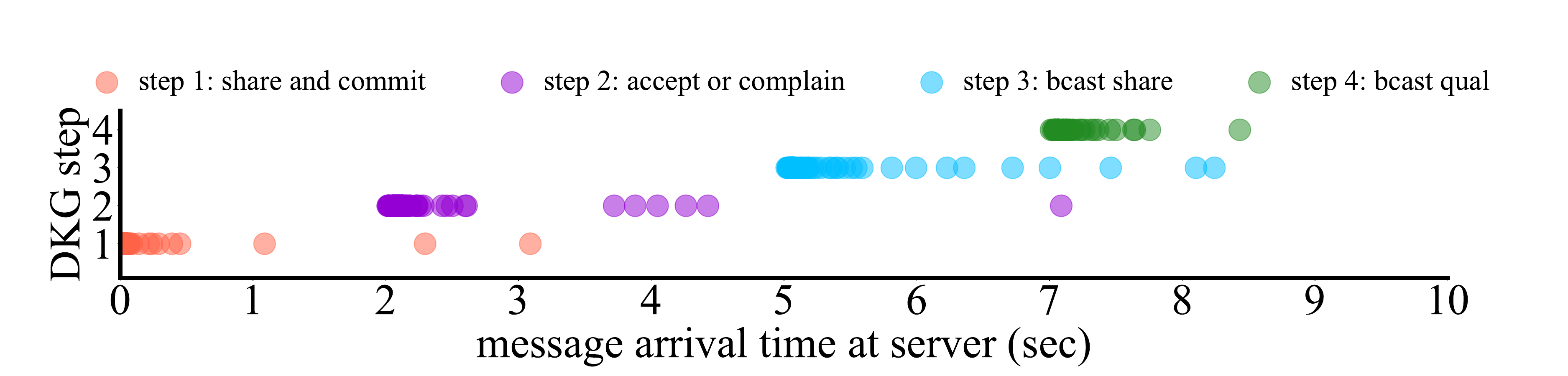}
  \caption{\small Generating shares of the secret key among 60 \coms.
  The four steps are described in Section~\ref{s:overview} and given as part~(1) in $\Pi_{\text{DKG}}$ in Appendix~\ref{app:dkg}. 
  } 
  \label{fig:dkg}
\end{figure}

\heading{Summation accuracy.}
Figure~\ref{fig:multisum-acc-long} and~\ref{fig:multisum-acc-short} 
  show that \sys{} achieves better sum accuracy $\tau$ (defined in \S\ref{s:properties}) 
  than BBGLR when they start a summation
  with the same number of clients.
When the waiting time is shorter, as in Figure~\ref{fig:multisum-acc-short}, 
  in each step there are more clients excluded from the summation and therefore 
  the discrepancy between \sys and BBGLR grows larger.

\subsection{Feasibility of a full private training session}\label{s:eval:mnist}

We implement the federated learning algorithm \texttt{FedAvg}~\cite{mcmahan17communication} on the non-private baseline.
We also use this algorithm for \sys{} and BBGLR but replace 
  its aggregation step with either \sys{} or BBGLR 
  to produce a secure version.
Inside of \texttt{FedAvg}, we use a multilayer perceptron for  
  image classification. 
Computation of the weights is done separately by each client on local data, and then aggregated by the server to update a global model. 
The server then sends the global model back to the clients. 
The number of training iterations that clients perform on their local data is referred to as an \textit{epoch}. 
We evaluated \sys{} on epochs of size 5, 10, and 20. 
We found that often, a larger epoch was correlated with faster convergence of \texttt{FedAvg} to some ``maximum" accuracy score. 
Additionally, because our neural network model calculations run very fast---there was, on average, less than a second difference between 
  clients' model fitting times for different epochs---and because \sys{} and the baselines were willing to wait for clients' inputs for up 
  to 10 seconds, the epoch size did not affect their overall runtime.

\begin{figure}
\begin{subfigure}[b]{0.48\linewidth}
\centering
 \includegraphics[width=\linewidth]{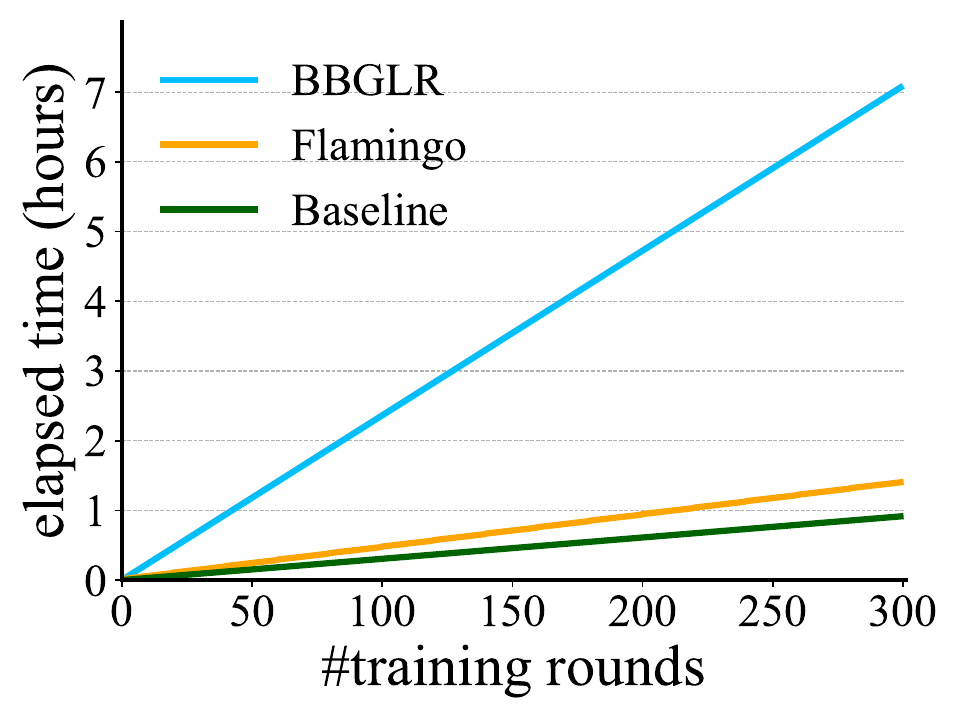}
 \caption{Run time. \texttt{EMNIST}}
    \label{fig:sgd-runtime-mnist}
\end{subfigure}
 \begin{subfigure}[b]{0.48\linewidth}
 \centering
 \includegraphics[width=\linewidth]{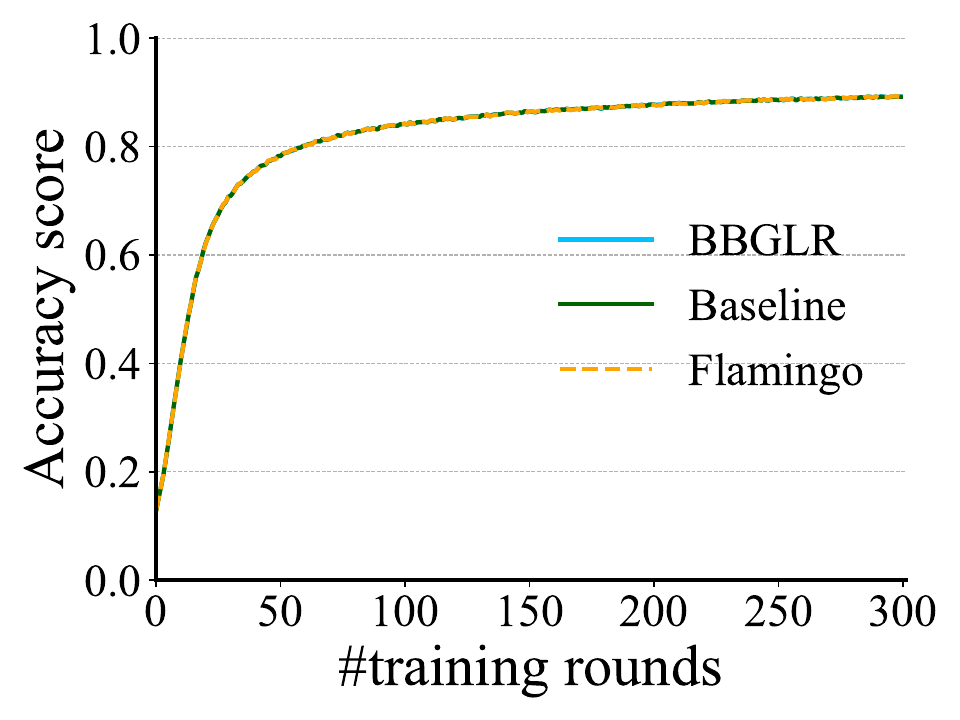}
 \caption{Accuracy. \texttt{EMNIST}}
    \label{fig:sgd-acc-mnist}
\end{subfigure}

\begin{subfigure}[b]{0.48\linewidth}
\centering
 \includegraphics[width=\linewidth]{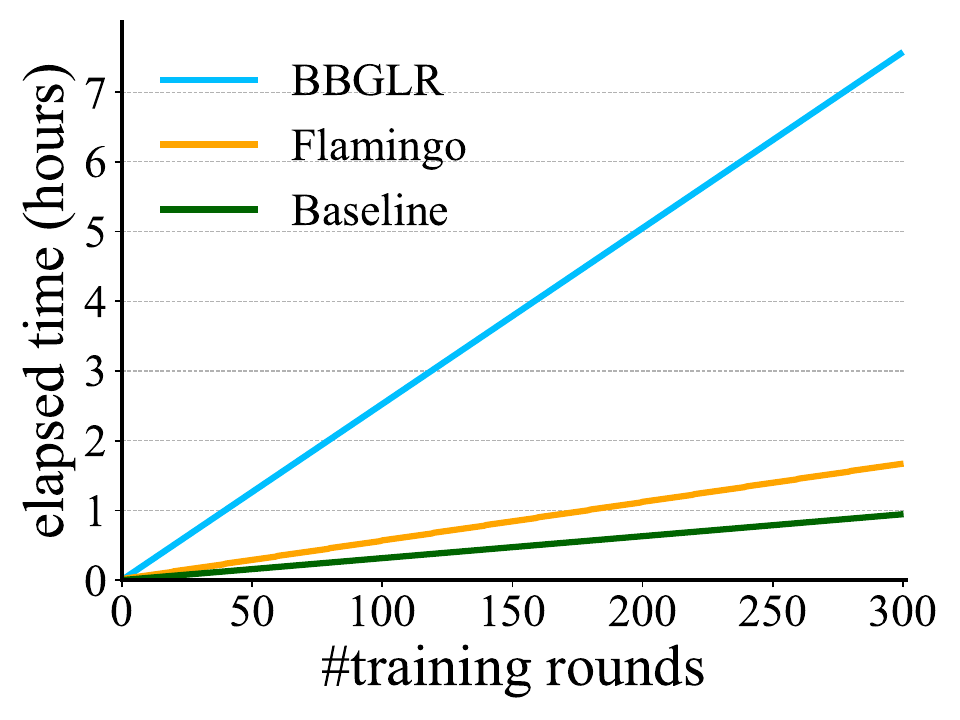}
 \caption{Run time. \texttt{CIFAR100}}
    \label{fig:sgd-runtime-letter}
\end{subfigure}
 \begin{subfigure}[b]{0.48\linewidth}
 \centering
\includegraphics[width=\linewidth]{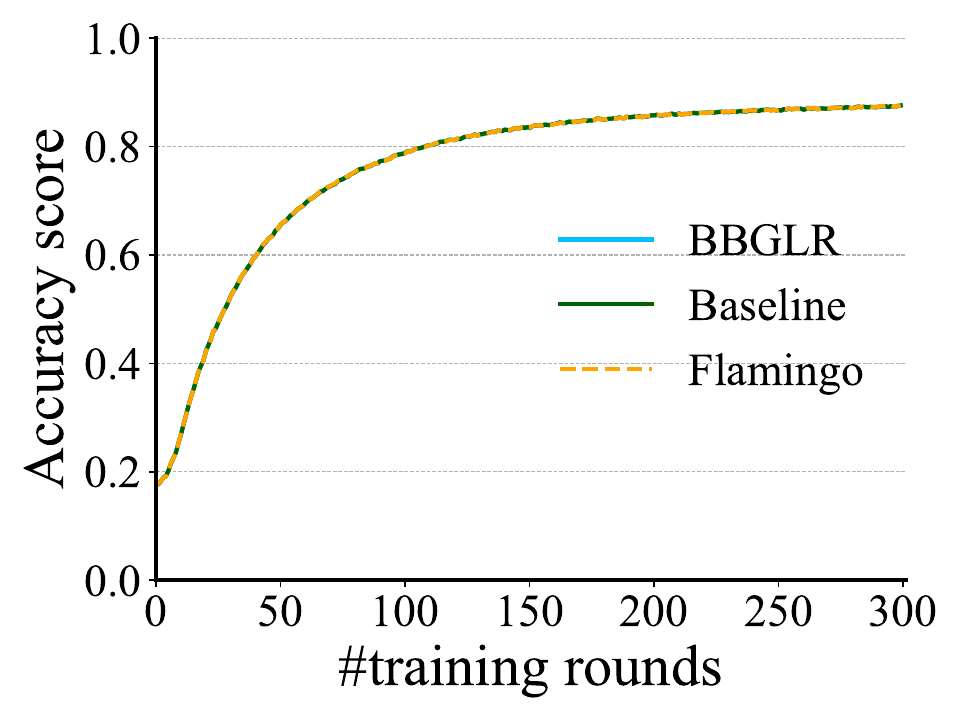}
 \caption{Accuracy. \texttt{CIFAR100}}
    \label{fig:sgd-acc-letter}
\end{subfigure}

  \caption{Evaluation for full training sessions on \texttt{EMNIST} and \texttt{CIFAR100} datasets. The number of clients per round is 128, the batch and epoch sizes for \texttt{FedAvg} are 10 and 20, respectively. \sys's setup cost is included during the first round, and it performs a secret key transfer every 20 rounds,  which adds to the total run time.
  The accuracy score is TensorFlow's sparse categorical accuracy score~\cite{tensorflow2015-whitepaper}.}
 \label{fig:training}
\end{figure}

We use two of TensorFlow's federated datasets~\cite{tensorflow2015-whitepaper}: (1) \texttt{EMNIST}, the Extended MNIST letter dataset from the Leaf repository~\cite{caldas2018leaf, caldasleafgithub}; and (2) \texttt{CIFAR100}, from the CIFAR-100 tiny images dataset~\cite{krizhevskycifar, krizhevsky2009learning}.
The \texttt{EMNIST} dataset has $\sim$340K training/$\sim$40K test samples, each with a square of $28\times 28$ pixels and 10 classes (digits). Each client has $\sim$226 samples.
During training, we use weight vectors with 8K 32-bit entries. 
The \texttt{CIFAR100} dataset has 50K training/10K test samples, each with a square of $32\times 32$ pixels and 100 classes. Each pixel additionally has red/blue/green values. Each client has 100 samples. To achieve good accuracy for the \texttt{CIFAR100} dataset, we use a more complex convolutional neural network than we do for the \texttt{EMNIST} dataset, with extra layers to build the model, normalize inputs between layers, and handle activation functions and overfitting. This results in longer weight vectors, with 500K 32-bit entries.

We randomly divide the datasets equally among 128 clients to create local data. Local models are trained with small batch sizes.
In \sys  and BBGLR, all weights (often floating point numbers) are encoded as positive integers. 
We do this by adding a large positive constant, multiplying by $2^{12}$, and truncating the weight to an unsigned 32-bit integer. 
Figure~\ref{fig:training} shows the result with $\delta = 1\%$. 

\heading{Running time.}
From Figures~\ref{fig:sgd-acc-mnist} and~\ref{fig:sgd-acc-letter}, we see that the \texttt{EMNIST} and \texttt{CIFAR100} datasets do not converge until about round 150 and 200, respectively, though their accuracy continues to improve slightly after that. Figures~\ref{fig:sgd-runtime-mnist} and~\ref{fig:sgd-runtime-letter} show \sys's running time is about 5.5$\times$ lower (i.e., better) than BBGLR for \texttt{EMNIST} and 4.8$\times$ for \texttt{CIFAR100} and about 1.4$\times$ higher (i.e., worse) than the non-private baseline for \texttt{EMNIST} and 1.7$\times$ for \texttt{CIFAR100}.
We believe these results provide evidence that \sys{} is an effective secure aggregation protocol for multi-round
  settings such as those required in federated learning.

\heading{Training accuracy.}
We measure training accuracy with TensorFlow's sparse categorical accuracy score, which is derived based on the resulting model's performance on test data.
Due to the way in which we encode floating points as integers, a small amount of precision from the weights is lost each round. 
We compare the accuracy of \sys{} and BBGLR's final global model to a model trained on the same datasets with the baseline version of \texttt{FedAvg} (which works over floating points) in Figures~\ref{fig:sgd-acc-mnist} and~\ref{fig:sgd-acc-letter}. 
We find that the encoding itself does not measurably affect accuracy.

\section{Extension with robustness}\label{s:enhanced}
Recall that our threat model (\S\ref{s:threat}) assumes that the server is controlled
  by a malicious adversary.
However, in some cases, the server is actually honest and wants to obtain
  a meaningful result---this motivates the need of having the protocol
  be \emph{robust} against malicious clients
  in addition to the other security properties. 
Here we discuss two types of robustness: 1) guaranteed output delivery, where 
  the protocol always outputs the correct result even if some clients are malicious;
  2) detect-and-abort, where incorrect output (if already computed) can be detected
  but the protocol will not output the correct result.
Clearly the first notion is stronger.

In Section~\ref{s:bbglr}, we discussed that in BBGLR, when the server is honest, 
  a malicious client may cause the output
  to be wrong or the protocol to abort: a malicious client changes the received share of a secret, 
  and the server may not reconstruct the secret (in which case the protocol aborts) 
  or the server may reconstruct to different secret. Or a malicious client does 
  not use the established pairwise secret in the setup phase to mask the input vector.
The protocol we give in Section~\ref{s:overview} has neither type of robustness 
  for the same reason as in BBGLR, but we provide an extension such that 
  malicious behaviors of the clients can be detected, though it does
  not have guaranteed output delivery.  
The full description of the extension is given in 
\ifthenelse{\boolean{longver}}
{Appendix~\ref{app:robust-protocol}.}
{Appendix E in the full version~\cite{ma23flamingo}.}

As a result, the extended protocol, if terminates, always output the correct sum, i.e., 
  the sum of the inputs from clients who participated; or it aborts, in which case the server
  detects that the output is incorrectly computed (but it cannot detect which
  client is malicious). 
That said, we want to emphasize that this notion of security is not 
  practically meaningful at the moment since malicious clients
  are free to provide any input they want (see the paragraph
  on input correctness in \S\ref{s:properties}).
However, it could become important in the future: if there is ever 
  some mechanism that could confirm the validity of clients' inputs (e.g.,~\cite{lycklama23rofl,bell22acorn}), 
  then this additional guarantee would ensure that an honest server always 
  gets valid outputs even if malicious clients exist in the system.

\section{Related Work}\label{s:related}

In this section we discuss alternative approaches to compute
  private sums and the reasons why they do not fit well in the 
  setting of federated learning.
Readers may also be interested in a recent survey of this 
  area~\cite{mansouri23sok}.

\heading{Pairwise masking.}
Bonawitz et al.~\cite{bonawitz17practical} and follow up works~\cite{bell20secure, so21turbo}, of which BBGLR~\cite{bell20secure} is the state-of-the-art, adopt the idea of DC networks~\cite{chaum88dining} in which pairwise masks are
used to hide individuals' inputs.
Such a construction is critical for both client-side and server-side efficiency:
first, since the vectors are long, one-time pad is the most efficient way to encrypt a 
vector; second, the server just needs to add up the vectors,
which achieves optimal server computation (even without privacy, the server at least has to do a similar sum). 
Furthermore, pairwise masking protocols support flexible input 
vectors, i.e., one can choose any $b$ (the number of bits for each 
component in the vector) as desired.
\sys{} improves on this line of work by reducing the overall round trip complexity for multiple sums.

\heading{MPC.}
Works like FastSecAgg~\cite{kadhe20fastsecagg} use a secret-sharing based
MPC to compute sums, which tolerates dropouts, but it has high communication
  as the inputs in federated learning are large.
Other results use non-interactive MPC protocols for addition~\cite{so22lightsecagg, so21turbo} where 
  all the clients establish shares 
  of zero during the setup.
  And then when the vectors of clients are requested,
  each client uses the share to mask the vector and sends it to the server.
However, to mask long vectors, the clients need to establish many shares of zeros, which is communication-expensive. Such 
  shares cannot be reused over multiple summation rounds (which is precisely what we address with \sys). 
Furthermore, the non-interactive protocols are not resilient against even one dropout client.

\heading{Additively homomorphic encryption.}
One can construct a single-server aggregation protocol using threshold additive 
  homomorphic encryption~\cite{elahi14privex,melis16efficient,popa09vpriv,popa11privacy, truex19hybrid, dasu22prov}. 
Each client encrypts its vector as a ciphertext under the public key of the threshold scheme, 
  and sends it to the committee.
The committee adds the ciphertexts from all the clients and gives the result to the server.
However, this does not work well for large inputs (like the large weight vectors found in federated learning)
  because encrypting the vector (say, using Paillier or a lattice-based scheme) and performing the threshold decryption will be very expensive.

A recent work~\cite{stevens22efficient} uses LWE-based homomorphic
  PRGs. 
This is an elegant approach but it has higher computation 
  and communication costs than works based on pairwise masking, including \sys.
The higher cost stems from one having to choose parameters (e.g., 
  vector length and the size of each component) that satisfy the 
  LWE assumption, and particularly the large LWE
  modulus that is required.

\heading{Multi-round setting.}
Recent work~\cite{guo22microfedml} designs a new multi-round 
secure aggregation protocol with reusable secrets that is very
  different from \sys's design. 
The protocol works well for small input domains (e.g., vectors with 
  small values) but cannot efficiently handle large domains as it requires brute forcing a discrete log during decryption.
In contrast, \sys{} does not have any restriction on the input domain.
A variant of the above work can also accommodate arbitrary input domains by 
 relying on DDH-based class groups~\cite{castagnos15linearly} (a more involved assumption than DDH). 

\section{Discussion}

We have focused our discussion on computing sums, but \sys{} can also 
  compute other functions such as max/min using affine aggregatable encodings~\cite{beimel14non, corrigan-gibbs17prio, halevi18best, AddankiGJOP22}.
  
\heading{Limitations.}
\sys{} assumes that the set of all clients $(N)$ involved in a
  training session is fixed before the training starts and
  that in each round $t$ some subset $S_t$ from $N$ is chosen.
We have not yet explored the case of handling clients that 
  dynamically join the training session.

Another aspect that we have not investigated in this work is
  that of handling an \emph{adaptive} adversary that can 
  dynamically change the set of parties that it compromises
  as the protocol executes.
In BBGLR, the adversary can be adaptive across rounds but not
  within a round; in \sys{} the adversary is static across all the rounds.
To our knowledge, an adversary that can be adaptive within a single round 
  has not been considered before in the federating learning setting.
It is not clear that existing approaches from other 
  fields~\cite{gentry21yoso} can be used due to different communication models.

Finally, secure aggregation reduces the leakage of
  individuals' inputs in federated learning but does not 
  fully eliminate it.
It is important to understand what information continues to leak.
Two recent works in this direction are as follows.
Elkordy et al.~\cite{elkordy23how} utilize tools from 
  information theory to bound the leakage with secure aggregation:
  they found that the amount of leakage reduces linearly with the
  number of clients.
Wang et al.~\cite{wang23eavesdrop} then show a new 
  inference attack against federated learning systems that use secure aggregation in which they are able to obtain the proportion
  of different labels in the overall training data.
While the scope of this attack is very limited, it may 
  inspire more advanced attacks.

\subsection*{Acknowledgments}
We thank the S\&P reviewers for their comments which improved the content of this paper.
We also thank Fabrice Benhamouda for discussion on elliptic curve scalar multiplication efficiency,
  Varun Chandrasekaran for advice on machine learning datasets,
  Yue Guo for answering questions about the ABIDES simulator,
  and Riad Wahby for pointers on hashing to elliptic curves. 
We thank Mariana Raykova and Adrià Gascón for helpful 
  discussions about the BBGLR paper.
Finally, we thank Dario Pasquini for making us aware of model
  inconsistency attacks in federated learning.
This work was partially funded by NSF grant CNS-2045861, DARPA contract HR0011-17-C0047, and a 
JP Morgan Chase \& Co Faculty Award.

This paper was prepared in part for information purposes by Artificial Intelligence Research Group and the AlgoCRYPT CoE of JPMorgan Chase \& Co and its affiliates (``JP Morgan'') and is not a product of the Research Department of JP Morgan. JP Morgan makes no representation and warranty whatsoever and disclaims all liability, for the completeness, accuracy, or reliability of the information contained herein. This document is not intended as investment research or investment advice, or a recommendation, offer, or solicitation for the purchase or sale of any security, financial instrument, financial product, or service, or to be used in any way for evaluating the merits of participating in any transaction, and shall not constitute a solicitation under any jurisdiction or to any person, if such solicitation under such jurisdiction or to such person would be unlawful.

\frenchspacing

{
\footnotesize
\begin{flushleft}
\balance
\setlength{\parskip}{0pt}
\setlength{\itemsep}{0pt}
\bibliographystyle{abbrv}
\bibliography{conferences,paper}

\begin{thebibliography}{10}

\bibitem{cloudflare-beacon}
Cloudflare randomness beacon.
\newblock \url{https://developers.cloudflare.com/randomness-beacon/}.

\bibitem{tensorflow2015-whitepaper}
M.~Abadi, A.~Agarwal, P.~Barham, E.~Brevdo, Z.~Chen, C.~Citro, G.~S. Corrado,
  A.~Davis, J.~Dean, M.~Devin, S.~Ghemawat, I.~Goodfellow, A.~Harp, G.~Irving,
  M.~Isard, Y.~Jia, R.~Jozefowicz, L.~Kaiser, M.~Kudlur, J.~Levenberg,
  D.~Man\'{e}, R.~Monga, S.~Moore, D.~Murray, C.~Olah, M.~Schuster, J.~Shlens,
  B.~Steiner, I.~Sutskever, K.~Talwar, P.~Tucker, V.~Vanhoucke, V.~Vasudevan,
  F.~Vi\'{e}gas, O.~Vinyals, P.~Warden, M.~Wattenberg, M.~Wicke, Y.~Yu, and
  X.~Zheng.
\newblock {TensorFlow}: Large-scale machine learning on heterogeneous systems.
\newblock Technical report, 2015.
\newblock \url{https://www.tensorflow.org/}.

\bibitem{abraham23bingo}
I.~Abraham, P.~Jovanovic, M.~Maller, S.~Meiklejohn, and G.~Stern.
\newblock Bingo: Adaptivity and asynchrony in verifiable secret sharing and
  distributed key generation.
\newblock In {\em Proceedings of the International Cryptology Conference
  (CRYPTO)}, 2023.

\bibitem{AddankiGJOP22}
S.~Addanki, K.~Garbe, E.~Jaffe, R.~Ostrovsky, and A.~Polychroniadou.
\newblock Prio+: Privacy preserving aggregate statistics via boolean shares.
\newblock In C.~Galdi and S.~Jarecki, editors, {\em Proceedings of the
  International Conference on Security and Cryptography for Networks(SCN)},
  2022.

\bibitem{anderson21aggregate}
E.~Anderson, M.~Chase, W.~Dai, F.~B. Durak, K.~Laine, S.~Sharma, and C.~Weng.
\newblock Aggregate measurement via oblivious shuffling.
\newblock Cryptology ePrint Archive, Paper 2021/1490, 2021.
\newblock \url{https://eprint.iacr.org/2021/1490}.

\bibitem{angel22efficient}
S.~Angel, A.~J. Blumberg, E.~Ioannidis, and J.~Woods.
\newblock Efficient representation of numerical optimization problems for
  {SNARKs}.
\newblock In {\em Proceedings of the USENIX Security Symposium}, 2022.

\bibitem{bagdasaryan20how}
E.~Bagdasaryan, A.~Veit, Y.~Hua, D.~Estrin, and V.~Shmatikov.
\newblock How to backdoor federated learning.
\newblock In {\em Proceedings of the Artificial Intelligence and Statistics
  Conference (AISTATS)}, 2020.

\bibitem{beimel14non}
A.~Beimel, A.~Gabizon, Y.~Ishai, E.~Kushilevitz, S.~Meldgaard, and
  A.~Paskin-Cherniavsky.
\newblock Non-interactive secure multiparty computation.
\newblock In {\em Proceedings of the International Cryptology Conference
  (CRYPTO)}, 2014.

\bibitem{bell20secure}
J.~Bell, K.~A. Bonawitz, A.~Gascón, T.~Lepoint, and M.~Raykova.
\newblock Secure single-server aggregation with (poly) logarithmic overhead.
\newblock In {\em Proceedings of the ACM Conference on Computer and
  Communications Security (CCS)}, 2020.

\bibitem{bell22acorn}
J.~Bell, A.~Gascón, T.~Lepoint, B.~Li, S.~Meiklejohn, M.~Raykova, and C.~Yun.
\newblock Acorn: Input validation for secure aggregation.
\newblock Cryptology ePrint Archive, Paper 2022/1461, 2022.
\newblock \url{https://eprint.iacr.org/2022/1461}.

\bibitem{bonawitz17practical}
K.~Bonawitz, V.~Ivanov, B.~Kreuter, A.~Marcedone, H.~McMahan, S.~Patel,
  D.~Ramage, A.~Segal, and K.~Seth.
\newblock Practical secure aggregation for privacy-preserving machine learning.
\newblock In {\em Proceedings of the ACM Conference on Computer and
  Communications Security (CCS)}, 2017.

\bibitem{boneh21lightweight}
D.~Boneh, E.~Boyle, H.~Corrigan-Gibbs, N.~Gilboa, and Y.~Ishai.
\newblock Lightweight techniques for private heavy hitters.
\newblock In {\em Proceedings of the IEEE Symposium on Security and Privacy
  (S\&P)}, 2021.

\bibitem{swu}
E.~Brier, J.-S. Coron, T.~Icart, D.~Madore, H.~Randriam, and M.~Tibouchi.
\newblock Efficient indifferentiable hashing into ordinary elliptic curves.
\newblock In {\em Proceedings of the International Cryptology Conference
  (CRYPTO)}, 2010.

\bibitem{abides}
D.~Byrd, M.~Hybinette, and T.~H. Balch.
\newblock {ABIDES}: Agent-based interactive discrete event simulation
  environment.
\newblock \url{https://github.com/abides-sim/abides}, 2020.

\bibitem{byrd20abides}
D.~Byrd, M.~Hybinette, and T.~H. Balch.
\newblock {ABIDES}: Towards high-fidelity multi-agent market simulation.
\newblock In {\em Proceedings of the 2020 ACM SIGSIM Conference on Principles
  of Advanced Discrete Simulation}, 2020.

\bibitem{caldasleafgithub}
S.~Caldas, S.~M.~K. Duddu, P.~Wu, T.~Li, J.~Kone{\v{c}}n{\`y}, H.~B. McMahan,
  V.~Smith, and A.~Talwalkar.
\newblock Leaf: A benchmark for federated settings.
\newblock \url{https://github.com/TalwalkarLab/leaf}.

\bibitem{caldas2018leaf}
S.~Caldas, S.~M.~K. Duddu, P.~Wu, T.~Li, J.~Kone{\v{c}}n{\`y}, H.~B. McMahan,
  V.~Smith, and A.~Talwalkar.
\newblock Leaf: A benchmark for federated settings.
\newblock {\em arXiv preprint arXiv:1812.01097}, 2018.
\newblock \url{https://arxiv.org/abs/1812.01097}.

\bibitem{canny04practical}
J.~Canny and S.~Sorkin.
\newblock Practical large-scale distributed key generation.
\newblock In {\em Proceedings of the International Conference on the Theory and
  Applications of Cryptographic Techniques (EUROCRYPT)}, 2004.

\bibitem{castagnos15linearly}
G.~Castagnos and F.~Laguillaumie.
\newblock Linearly homomorphic encryption from ddh.
\newblock Cryptology ePrint Archive, Paper 2015/047, 2015.
\newblock \url{https://eprint.iacr.org/2015/047}.

\bibitem{chan11privacy}
T.-H.~H. Chan, E.~Shi, and D.~Song.
\newblock Privacy-preserving stream aggregation with fault tolerance.
\newblock In {\em Proceedings of the International Financial Cryptography
  Conference}, 2011.

\bibitem{chase19seemless}
M.~Chase, A.~Deshpande, E.~Ghosh, and H.~Malvai.
\newblock Seemless: Secure end-to-end encrypted messaging with less trust.
\newblock In {\em Proceedings of the ACM Conference on Computer and
  Communications Security (CCS)}, 2019.

\bibitem{chaum88dining}
D.~L. Chaum.
\newblock The dining cryptographers problem: Unconditional sender and recipient
  untraceability.
\newblock {\em Journal of Cryptology}, 1(1), 1988.

\bibitem{chowdhury21eiffel}
A.~R. Chowdhury, C.~Guo, S.~Jha, and L.~van~der Maaten.
\newblock Eiffel: Ensuring integrity for federated learning.
\newblock In {\em Proceedings of the ACM Conference on Computer and
  Communications Security (CCS)}, 2022.

\bibitem{corrigan-gibbs17prio}
H.~Corrigan-Gibbs and D.~Boneh.
\newblock Prio: Private, robust, and scalable computation of aggregate
  statistics.
\newblock In {\em Proceedings of the USENIX Symposium on Networked Systems
  Design and Implementation (NSDI)}, 2017.

\bibitem{das22spurt}
S.~Das, V.~Krishnan, I.~M. Isaac, and L.~Ren.
\newblock Spurt: Scalable distributed randomness beacon with transparent setup.
\newblock In {\em Proceedings of the IEEE Symposium on Security and Privacy
  (S\&P)}, 2021.

\bibitem{das22practical}
S.~Das, T.~Yurek, Z.~Xiang, A.~Miller, L.~Kokoris-Kogias, and L.~Ren.
\newblock Practical asynchronous distributed key generation.
\newblock In {\em Proceedings of the IEEE Symposium on Security and Privacy
  (S\&P)}, 2022.

\bibitem{dasu22prov}
V.~A. Dasu, S.~Sarkar, and K.~Mandal.
\newblock {PROV-FL}: Privacy-preserving round optimal verifiable federated
  learning.
\newblock In {\em Proceedings of the ACM Workshop on Artificial Intelligence
  and Security}, 2022.

\bibitem{desmedt89threshold}
Y.~Desmedt and Y.~Frankel.
\newblock Threshold cryptosystems.
\newblock In {\em Proceedings of the International Cryptology Conference
  (CRYPTO)}, 1989.

\bibitem{elahi14privex}
T.~Elahi, G.~Danezis, and I.~Goldberg.
\newblock {PrivEx}: Private collection of traffic statistics for anonymous
  communication networks.
\newblock In {\em Proceedings of the ACM Conference on Computer and
  Communications Security (CCS)}, 2014.

\bibitem{elkordy23how}
A.~R. Elkordy, J.~Zhang, Y.~H. Ezzeldin, K.~Psounis, and S.~Avestimehr.
\newblock {How Much Privacy Does Federated Learning with Secure Aggregation
  Guarantee?}
\newblock In {\em Proceedings of the Privacy Enhancing Technologies Symposium
  (PETS)}, 2023.

\bibitem{feldman87practical}
P.~Feldman.
\newblock A practical scheme for non-interactive verifiable secret sharing.
\newblock In {\em Proceedings of the IEEE Symposium on Foundations of Computer
  Science (FOCS)}, 1987.

\bibitem{gennaro08threshold}
R.~Gennaro, S.~Halevi, H.~Krawczyk, and T.~Rabin.
\newblock Threshold rsa for dynamic and adhoc groups.
\newblock In {\em Proceedings of the International Conference on the Theory and
  Applications of Cryptographic Techniques (EUROCRYPT)}, 2008.

\bibitem{gennaro06secure}
R.~Gennaro, S.~Jarecki, H.~Krawczyk, and T.~Rabin.
\newblock Secure distributed key generation for discrete-log based
  cryptosystems.
\newblock In {\em Journal of Cryptology}, 2006.

\bibitem{gentry21yoso}
C.~Gentry, S.~Halevi, H.~Krawczyk, B.~Magri, J.~B. Nielsen, T.~Rabin, and
  S.~Yakoubov.
\newblock {YOSO}: You only speak once / secure {MPC} with stateless ephemeral
  roles.
\newblock In {\em Proceedings of the International Cryptology Conference
  (CRYPTO)}, 2021.

\bibitem{gilbert59random}
E.~N. Gilbert.
\newblock Random graphs.
\newblock In {\em The Annals of Mathematical Statistics}, 1959.

\bibitem{guo22microfedml}
Y.~Guo, A.~Polychroniadou, E.~Shi, D.~Byrd, and T.~Balch.
\newblock {MicroFedML}: Privacy preserving federated learning for small
  weights.
\newblock Cryptology ePrint Archive, Paper 2022/714, 2022.
\newblock \url{https://eprint.iacr.org/2022/714}.

\bibitem{ha08cubic}
S.~Ha, I.~Rhee, and L.~Xu.
\newblock Cubic: A new tcp-friendly high-speed tcp variant.
\newblock {\em ACM SIGOPS operating systems review}, 2008.

\bibitem{halevi18best}
S.~Halevi, Y.~Ishai, E.~Kushilevitz, and T.~Rabin.
\newblock Best possible information-theoretic {MPC}.
\newblock In {\em Proceedings of the Theory of Cryptography Conference (TCC)},
  2018.

\bibitem{hashec}
A.~F. Hernández, S.~Scott, N.~Sullivan, R.~S. Wahby, and C.~Wood.
\newblock Hashing to elliptic curves.
\newblock
  \url{https://www.ietf.org/archive/id/draft-irtf-cfrg-hash-to-curve-10.html},
  2021.

\bibitem{Herzberg95proactive}
A.~Herzberg, S.~Jarecki, H.~Krawczyk, and M.~Yung.
\newblock Proactive secret sharing or how to cope with perpetual leakage.
\newblock In {\em Proceedings of the International Cryptology Conference
  (CRYPTO)}, 1995.

\bibitem{hu21merkle2}
Y.~Hu, K.~Hooshmand, H.~Kalidhindi, S.~J. Yang, and R.~A. Popa.
\newblock Merkle$^2$: A low-latency transparency log system.
\newblock In {\em Proceedings of the IEEE Symposium on Security and Privacy
  (S\&P)}, 2021.

\bibitem{kadhe20fastsecagg}
S.~Kadhe, N.~Rajaraman, O.~O. Koyluoglu, and K.~Ramchandran.
\newblock {FastSecAgg}: Scalable secure aggregation for privacy-preserving
  federated learning.
\newblock In {\em ICML Workshop on Federated Learning for User Privacy and Data
  Confidentiality}, 2020.

\bibitem{kairouz21advances}
P.~Kairouz, H.~B. McMahan, B.~Avent, A.~Bellet, M.~Bennis, A.~N. Bhagoji,
  K.~Bonawit, Z.~Charles, G.~Cormode, R.~Cummings, R.~G.~L. D’Oliveira,
  H.~Eichner, S.~El~Rouayheb, D.~Evans, J.~Gardner, Z.~Garrett, A.~Gascón,
  B.~Ghazi, P.~B. Gibbons, M.~Gruteser, Z.~Harchaoui, C.~He, L.~He, Z.~Huo,
  B.~Hutchinson, J.~Hsu, M.~Jaggi, T.~Javidi, G.~Joshi, M.~Khodak, J.~Konecný,
  A.~Korolova, F.~Koushanfar, S.~Koyejo, T.~Lepoint, Y.~Liu, P.~Mittal,
  M.~Mohri, R.~Nock, A.~Özgür, R.~Pagh, H.~Qi, D.~Ramage, R.~Raskar,
  M.~Raykova, D.~Song, W.~Song, S.~U. Stich, Z.~Sun, A.~Theertha~Suresh,
  F.~Tramèr, P.~Vepakomma, J.~Wang, L.~Xiong, Z.~Xu, Q.~Yang, F.~X. Yu, H.~Yu,
  and S.~Zhao.
\newblock Advances and open problems in federated learning.
\newblock In {\em Foundations and Trends in Machine Learning}, 2021.

\bibitem{kokoriskogias20asynchronous}
E.~Kokoris-Kogias, D.~Malkhi, and A.~Spiegelman.
\newblock Asynchronous distributed key generation for computationally-secure
  randomness, consensus, and threshold signatures.
\newblock In {\em Proceedings of the ACM Conference on Computer and
  Communications Security (CCS)}, 2020.

\bibitem{krishna19churp}
S.~Krishna, D.~Maram, F.~Zhang, L.~Wang, A.~Low, Y.~Zhang, A.~Juels, and
  D.~Song.
\newblock Churp: Dynamic-committee proactive secret sharing.
\newblock In {\em Proceedings of the ACM Conference on Computer and
  Communications Security (CCS)}, 2019.

\bibitem{krizhevsky2009learning}
A.~Krizhevsky.
\newblock Learning multiple layers of features from tiny images.
\newblock Technical report, 2009.
\newblock \url{https://www.cs.toronto.edu/~kriz/learning-features-2009-TR.pdf}.

\bibitem{krizhevskycifar}
A.~Krizhevsky, V.~Nair, and G.~Hinton.
\newblock The {CIFAR-100} dataset.
\newblock \url{https://www.cs.toronto.edu/~kriz/cifar.html}.

\bibitem{lai21oort}
F.~Lai, X.~Zhu, H.~V. Madhyastha, and M.~Chowdhury.
\newblock Oort: Efficient federated learning via guided participant selection.
\newblock In {\em Proceedings of the USENIX Symposium on Operating Systems
  Design and Implementation (OSDI)}, 2021.

\bibitem{leung22aardvark}
D.~Leung, Y.~Gilad, S.~Gorbunov, L.~Reyzin, and N.~Zeldovich.
\newblock Aardvark: An asynchronous authenticated dictionary with short proofs.
\newblock In {\em Proceedings of the USENIX Security Symposium}, 2022.

\bibitem{lycklama23rofl}
H.~Lycklama, L.~Burkhalter, A.~Viand, N.~Küchler, and A.~Hithnawi.
\newblock {RoFL}: Robustness of secure federated learning.
\newblock In {\em Proceedings of the IEEE Symposium on Security and Privacy
  (S\&P)}, 2023.

\bibitem{mansouri23sok}
M.~Mansouri, M.~Önen, W.~B. Jaballah, and M.~Conti.
\newblock {SoK}: Secure aggregation based on cryptographic schemes for
  federated learning.
\newblock In {\em Proceedings of the Privacy Enhancing Technologies Symposium
  (PETS)}, 2023.

\bibitem{mcmahan17communication}
H.~B. McMahan, E.~Moore, D.~Ramage, S.~Hampson, and B.~A. y~Arcas.
\newblock Communication-efficient learning of deep networks from decentralized
  data.
\newblock In {\em Proceedings of the Artificial Intelligence and Statistics
  Conference (AISTATS)}, 2017.

\bibitem{melara15coniks}
S.~Melara, A.~Blankstein, J.~Bonneau, E.~W. Felten, and M.~J. Freedman.
\newblock {CONIKS}: bringing key transparency to end users.
\newblock In {\em Proceedings of the USENIX Security Symposium}, 2015.

\bibitem{melis16efficient}
L.~Melis, G.~Danezis, and E.~D. Cristofaro.
\newblock Efficient private statistics with succinct sketches.
\newblock In {\em Proceedings of the Network and Distributed System Security
  Symposium (NDSS)}, 2016.

\bibitem{melis18exploiting}
L.~Melis, C.~Song, E.~D. Cristofaro, and V.~Shmatikov.
\newblock Exploiting unintended feature leakage in collaborative learning.
\newblock In {\em Proceedings of the IEEE Symposium on Security and Privacy
  (S\&P)}, 2019.

\bibitem{pasquini22eluding}
D.~Pasquini, D.~Francati, and G.~Ateniese.
\newblock Eluding secure aggregation in federated learning via model
  inconsistency.
\newblock In {\em Proceedings of the ACM Conference on Computer and
  Communications Security (CCS)}, 2022.

\bibitem{pedersen91threshold}
T.~Pedersen.
\newblock A threshold cryptosystem without a trusted party.
\newblock In {\em Proceedings of the International Conference on the Theory and
  Applications of Cryptographic Techniques (EUROCRYPT)}, 1991.

\bibitem{popa09vpriv}
R.~A. Popa, H.~Balakrishnan, and A.~J. Blumberg.
\newblock {VPriv}: Protecting privacy in location-based vehicular services.
\newblock In {\em Proceedings of the USENIX Security Symposium}, 2009.

\bibitem{popa11privacy}
R.~A. Popa, A.~J. Blumberg, H.~Balakrishnan, and F.~H. Li.
\newblock Privacy and accountability for location-based aggregate statistics.
\newblock In {\em Proceedings of the ACM Conference on Computer and
  Communications Security (CCS)}, 2011.

\bibitem{roth21mycelium}
E.~Roth, K.~Newatia, Y.~Ma, K.~Zhong, S.~Angel, and A.~Haeberlen.
\newblock Mycelium: Large-scale distributed graph queries with differential
  privacy.
\newblock In {\em Proceedings of the ACM Symposium on Operating Systems
  Principles (SOSP)}, 2021.

\bibitem{roth19honeycrisp}
E.~Roth, D.~Noble, B.~H. Falk, and A.~Haeberlen.
\newblock Honeycrisp: Large-scale differentially private aggregation without a
  trusted core.
\newblock In {\em Proceedings of the ACM Symposium on Operating Systems
  Principles (SOSP)}, 2019.

\bibitem{roth20orchard}
E.~Roth, H.~Zhang, A.~Haeberlen, and B.~C. Pierce.
\newblock Orchard: Differentially private analytics at scale.
\newblock In {\em Proceedings of the USENIX Symposium on Operating Systems
  Design and Implementation (OSDI)}, 2020.

\bibitem{shi11privacy}
E.~Shi, T.-H.~H. Chan, E.~Rieffel, R.~Chow, and D.~Song.
\newblock Privacy-preserving aggregation of time-series data.
\newblock In {\em Proceedings of the Network and Distributed System Security
  Symposium (NDSS)}, 2011.

\bibitem{shoup02securing}
V.~Shoup and R.~Gennaro.
\newblock Securing threshold cryptosystems against chosen ciphertext attack.
\newblock In {\em Journal of Cryptology}, 2002.

\bibitem{so21turbo}
J.~So, B.~Güler, and A.~S. Avestimehr.
\newblock Turbo-aggregate: Breaking the quadratic aggregation barrier in secure
  federated learning.
\newblock In {\em Journal on Selected Areas in Information Theory}, 2021.

\bibitem{so22lightsecagg}
J.~So, C.~He, C.-S. Yang, S.~Li, Q.~Yu, R.~E. Ali, B.~Guler, and S.~Avestimehr.
\newblock {LightSecAgg}: a lightweight and versatile design for secure
  aggregation in federated learning.
\newblock In {\em Proceedings of Machine Learning and Systems}, 2022.

\bibitem{stevens22efficient}
T.~Stevens, C.~Skalka, C.~Vincent, J.~Ring, S.~Clark, and J.~Near.
\newblock Efficient differentially private secure aggregation for federated
  learning via hardness of learning with errors.
\newblock In {\em Proceedings of the USENIX Security Symposium}, 2022.

\bibitem{tomescu19transparency}
A.~Tomescu, V.~Bhupatiraju, D.~Papadopoulos, C.~Papamanthou, N.~Triandopoulos,
  and S.~Devadas.
\newblock Transparency logs via append-only authenticated dictionaries.
\newblock In {\em Proceedings of the ACM Conference on Computer and
  Communications Security (CCS)}, 2019.

\bibitem{truex19hybrid}
S.~Truex, N.~Baracaldo, A.~Anwar, T.~Steinke, H.~Ludwig, R.~Zhang, and Y.~Zhou.
\newblock A hybrid approach to privacy-preserving federated learning.
\newblock In {\em Proceedings of the ACM workshop on artificial intelligence
  and security}, 2019.

\bibitem{tyagi22versa}
N.~Tyagi, B.~Fisch, A.~Zitek, J.~Bonneau, and S.~Tessaro.
\newblock {VeRSA}: Verifiable registries with efficient client audits from
  {RSA} authenticated dictionaries.
\newblock In {\em Proceedings of the ACM Conference on Computer and
  Communications Security (CCS)}, 2022.

\bibitem{tzialla22transparency}
I.~Tzialla, A.~Kothapalli, B.~Parno, and S.~Setty.
\newblock Transparency dictionaries with succinct proofs of correct operation.
\newblock In {\em Proceedings of the Network and Distributed System Security
  Symposium (NDSS)}, 2022.

\bibitem{sswu}
R.~S. Wahby and D.~Boneh.
\newblock Fast and simple constant-time hashing to the {BLS12-381} elliptic
  curve.
\newblock In {\em Proceedings of the Conference on Cryptographic Hardware and
  Embedded Systems (CHES)}, 2019.

\bibitem{wang23eavesdrop}
L.~Wang, S.~Xu, X.~Wang, and Q.~Zhu.
\newblock Eavesdrop the composition proportion of training labels in federated
  learning.
\newblock arXiv:1910/06044, 2023.
\newblock \url{https://arxiv.org/abs/1910.06044}.

\bibitem{yuan20federated}
H.~Yuan and T.~Ma.
\newblock Federated accelerated stochastic gradient descent.
\newblock In {\em Neural Information Processing Systems (NeurIPS)}, 2020.

\bibitem{zhong22ibex}
K.~Zhong, Y.~Ma, and S.~Angel.
\newblock Ibex: Privacy-preserving ad conversion tracking and bidding.
\newblock In {\em Proceedings of the ACM Conference on Computer and
  Communications Security (CCS)}, 2022.

\bibitem{zhong23addax}
K.~Zhong, Y.~Ma, Y.~Mao, and S.~Angel.
\newblock Addax: A fast, private, and accountable ad exchange infrastructure.
\newblock In {\em Proceedings of the USENIX Symposium on Networked Systems
  Design and Implementation (NSDI)}, 2023.

\bibitem{zhu19deep}
L.~Zhu, Z.~Liu, and S.~Han.
\newblock Deep leakage from gradients.
\newblock In {\em Neural Information Processing Systems (NeurIPS)}, 2019.

\end{thebibliography}
\end{flushleft}
}

\appendices

\section{Failure and threat model details}\label{app:failure-model}
In this section, we give the full details of the dropout rate and the corruption rate (\S\ref{s:threat}).

\heading{Dropout rate.}
Recall that we upper bound the dropout rate of the sum contributors ($S_t$)
  in one round as $\delta$.
For decryptors, we consider the dropout rate  
  in one summation round
  and assume it is at most $\delta_D$. 
Note that $\delta$ and $\delta_D$ are individually determined by
  the server timeout at those steps (recall that in each round, clients in $S_t$ only participate in
  the first step; the following two steps only involve decryptors).

\heading{Corruption rate.}
For corruption, we denote the corrupted rate in $S_t$ as $\eta_{S_t}$ 
  and the corrupted rate in \coms as $\eta_D$. 
In the \sys{} system, $\eta$ is given; $\eta_{S_t}$ and $\eta_D$ depends on $\eta$. 
Note that the fraction of malicious clients in a chosen
  subset of $[N]$ (e.g., $S_t$, $\comset$) may not be exactly $\eta$, but rather a random variable $\eta^*$
  from a distribution that is parameterized by
  $\eta$, $N$ and the size of the chosen set.
Since the expectation of $\eta^*$ is equal to $\eta$,
  and when the size of the chosen set is large (e.g., $S_t$),
  the deviation of $\eta^*$ from $\eta$ is negligible (i.e., $\eta^*$ is almost
  equal to $\eta$). 
Therefore, $\eta_{S_t}$ can be considered as equal to $\eta$.
On the other hand, since $\mathcal{D}$ is a small set, we cannot
  assume $\eta_D$ is equal to $\eta$.
Later in Appendix~\ref{app:param} we show how to choose $L$ to ensure $\eta_D$ 
  satisfies the inequality required in Theorem~\ref{thm:security-main} with overwhelming probability.

\heading{Security parameters.}
In the following definitions and proofs, we use $\kappa$ for the information-theoretic security parameter and $\lambda$ for the computational security parameter.

\section{Full Protocol Description}\label{app:full-mal}

\begin{figure}
\begin{tcolorbox}[enhanced, boxsep=1mm, left= 0mm, right=0.5mm, title={\textbf{\footnotesize Protocol $\Pi_{\text{setup}}$.}}]
\linespread{1.4}
\footnotesize
\textbf{Parties.} Clients $1, \ldots, N$ and a server.

\textbf{Parameters.} Number of pre-selected decryptors $L$. Let $L=3\ell+1$.

\textbf{Protocol outputs.} A set of $t$ clients ($2\ell+1 \le t \le 3\ell+1$) hold secret sharing of a secret key $SK$. All the clients in $[N]$ and the server hold the associated public key $PK$. 

\begin{myitemize2}
    \item The server and all the clients in $ [N]$ invoke $\mathcal{F}_{\text{rand}}$ and receive a binary string $v \xleftarrow{\$}\{0,1\}^\lambda$.
    \item The server and all the clients in $ [N]$ computes
    
    $\comset_0 \leftarrow \ChooseSet(v, 0, L, N)$.

    \item All the clients $u\in \mathcal{D}_0$ and the server run $\Pi_{\text{DKG}}$ (Fig.~\ref{fig:dkg-dropout}). 

    \item The server broadcasts the signed $PK$s received from the clients in $\mathcal{D}_0$ to all the clients in $[N]$. 
    
    \item A client in $[N]$ aborts if it received less than $2\ell+1$ valid signatures on $PK$s signed by the parties defined by $\ChooseSet(v, 0, L, N)$.

\end{myitemize2}

\end{tcolorbox}
\caption{Setup phase with total number of clients $N$. 
$\Frand$ is modeled as a random beacon service.} 
\label{fig:protocol-setup}
\end{figure}

\subsection{Definition of cryptographic primitives}\label{app:primitives}
In this section, we formally define the cryptographic primitives used in \sys{} protocol
   that are not given in Section~\ref{s:blocks}.

\begin{definition}[DDH assumption]\label{def:ddh}
    Given a cyclic group $\mathbb{G}$ with order $q$, and let the generator of $\mathbb{G}$ be $g$. 
    Let $a, b, c$ be uniformly sampled elements from $\mathbb{Z}_q$. 
    We say that DDH is hard if the two distributions $(g^a, g^b, g^{ab})$ and $(g^a, g^b, g^c)$ are computationally indistinguishable. 
\end{definition}

\begin{definition}[ElGamal encryption]\label{def:elgamal}
    Let $\mathbb{G}$ be a group of order $q$ in which DDH is hard. 
    ElGamal encryption scheme consists of the following three algorithms. 
    \begin{myitemize2}
        \item  $\AsymGen(1^\lambda) \rightarrow (SK, PK)$: sample a random element $s$ from $\mathbb{Z}_q$, and output $SK=s$ and $PK=g^s$. 
        \item $\AsymEnc (PK, h) \rightarrow (c_0, c_1)$: sample a random element $y$ from $\mathbb{Z}_q$ and compute $c_0 = g^y$ and $c_1 = h \cdot PK^y$. 
        \item $\AsymDec (SK, (c_0, c_1)) \rightarrow h$: compute $h=(c_0^{SK})^{-1} \cdot c_1$. 
    \end{myitemize2}
    We say that ElGamal encryption is secure if it has CPA security. 
    Note that if DDH assumption (Def.\ref{def:ddh}) holds, then ElGamal encryption is secure. 
\end{definition}

\begin{definition}[Authenticated encryption]\label{def:authenc}
    An authenticated encryption scheme consists of the following algorithms:
    \begin{myitemize2}
        \item $\SymGen(1^\lambda) \rightarrow k$: sample a key $k$ uniformly random from $\{0,1\}^\lambda$. 
        \item $\SymEnc(k, m) \rightarrow c$: take in a key $k$ and a message $m$, output a ciphertext $c$. 
        \item $\SymDec(k, c)$: take in a key $k$ and a ciphertext $c$, output a plaintext $m$ or $\perp$ (decryption fails). 
    \end{myitemize2}
    We say that the scheme is secure if it has CPA security and ciphertext integrity. 
\end{definition}

For simplicity, we use $\AsymEnc$ and $\SymEnc$ to refer to the encryption schemes.

\begin{definition}[Signature scheme]\label{deef:sig}
    A signature scheme consists of the following algorithms:
    \begin{myitemize2}
        \item $SGen(1^\lambda) \rightarrow (sk, pk)$: generate a pair of siging key $sk$ and verfication key $pk$. 
        \item $Sign(sk, m) \rightarrow \sigma$: take in a signing key $sk$ and message $m$, outputs a signature $\sigma$. 
        \item $VerSig(pk, m, \sigma) \rightarrow b$: take in a verification key $pk$, a messagee $m$ and a signature $\sigma$,
        output valid or not as $b = 1, 0$. 
    \end{myitemize2}
    We say that the signature scheme is secure if the probability that, given $m_1, \ldots, m_z$, an attacker who can query the
    signing challenger and finds a valid $(m', \sigma')$ where $m' \not\in \{m_1, \ldots, m_z\}$ is negligible. 
\end{definition}

\subsection{Setup phase and distributed key generation}\label{app:dkg}
The setup protocol is conceptually simple, as shown in Figure~\ref{fig:protocol-setup}.
A crucial part of the setup phase is the distributed key generation (DKG).
We first describe the algorithms used in DKG. 

\heading{Algorithms.}
Let $\mathbb{G}$ be a group with order $q$ in which discrete log is hard. 
The discrete-log based DKG protocol builds on Feldman verifiable secret sharing~\cite{feldman87practical},
which we provide below.
The sharing algorithm takes in the threshold parameters
$L, \ell$, and a secret $s \in \mathbb{Z}_q$, 
chooses a polynomial with random coefficients except the constant term, i.e.,
$p(X) = a_0 + a_1 X + \ldots + a_\ell X^\ell ( a_0 = s)$,
and outputs the commitments
$ A_{k} = g^{a_{k}} \in \mathbb{G} $ for $k=0,1,\ldots, \ell.$
The $j$-th share $s_j$ is $ p(j)$ for $j=1, \ldots, L$. 

To verify the $j$-th share against the commitments,
the verification algorithm takes in $s_j$ and a set of commitments $\{A_k\}_{k=0}^\ell$,
  and checks if
\[g^{s_j} = \prod_{k=0}^\ell (A_{k})^{j^k}.\]

We define the above algorithms as 
\begin{itemize}
    \item $FShare(s, \ell, L) \rightarrow \{s_{j}\}_{j=1}^L, \{A_k\}_{k=0}^\ell$,  
    \item $FVerify( s_j, \{A_k\}_{k=0}^\ell ) \rightarrow b$ where $b\in\{0,1\}$. 
\end{itemize}

The GJKR-DKG uses a variant of the above algorithm, $PShare$ and $PVerify$ based on Pedersen
  commitment, for security reason~\cite{gennaro06secure}. 
The $PShare$ algorithm chooses two random polynomials 
\[p(X) = a_0 + a_1 X + \ldots + a_\ell X^\ell, \quad a_0 = s\]
\[p'(X) = b_0 + b_1 X + \ldots + b_\ell X^\ell\]
and outputs 
\[\{p(j)\}_{j=1}^L, \{p'(j)\}_{j=1}^L, C_k := g^{a_{k}} h^{b_{k}} \text{ for } k=0,\ldots, \ell,\]
  where $g, h\in \mathbb{G}$. 

To verify against the share $s_j = p(j)$, $PVerify$ takes in
$s_j' = p'(j)$ and $\{C_k\}_{k=0}^\ell$, and checks if 
\[g^{s_j} h^{s_j'} = \prod_{k=0}^\ell (C_{k})^{j^k} . \] 

The algorithms $PShare$ and $PVerify$ can be defined analogously to $Fshare$ and $FVerify$:
\begin{itemize}
    \item $PShare(s, \ell, L) \rightarrow \{s_{j}\}_{j=1}^L, \{s_j'\}_{j=1}^L, \{C_k\}_{k=0}^\ell$,  
    \item $PVerify( s_j, s_j', \{C_k\}_{k=0}^\ell ) \rightarrow b$ where $b\in\{0,1\}$. 
\end{itemize}

\heading{Protocol.}
We give the modified DKG protocol $\Pi_{\text{DKG}}$ from GJKR-DKG in Figure~\ref{fig:dkg-dropout}. 
The participating parties can drop out, as long as $\eta_D + \delta_D < 1/3$.

\heading{Correctness and security.}
\ifthenelse{\boolean{longver}}
{We analyze the properties of $\Pi_{\text{DKG}}$ in this section.
We start by revisiting the correctness and security definitions of GJKR-DKG,
  and then discuss how our definition differs from theirs because of
  a weakening of the communication model.
In GJKR-DKG, correctness has three folds:
\begin{myenumerate2}
    \item All subsets of honest parties define the same unique secret key.
    \item All honest parties have the same public key.
    \item The secret key is uniformly random.
\end{myenumerate2}
Security means that no information about the secret key  
  can be learned by the adversary except for what
  is implied by the public key.
}
{}
For $\Pi_{\text{DKG}}$,
if the server is honest, then our communication model (\S\ref{s:threat}) is equivalent to having
a fully synchronous channel, hence in this case the correctness and security 
properties in the prior work hold.
When the server is malicious, we show that $\Pi_{\text{DKG}}$ satisfies the following 
  correctness (C1, C2, C3, C4) and security (S). 
\ifthenelse{\boolean{longver}}
{}
{The proof of Lemma~\ref{lemma:dkg} is given in the full version~\cite{ma23flamingo}.}
\begin{itemize}
    \item[C1.] Each honest party either has no secret at the end
               or agrees on the same QUAL with other honest parties.
    \item[C2.] The agreed QUAL sets defines a unique secret key.
    \item[C3.] The secret key defined by QUAL is uniformly random. 
    \item[C4.] Each honest party, either has no public key, or outputs the same public key with other honest parties. 
    \item[S.] Malicious parties learns no information about the secret key except for what is implied by the public key.
\end{itemize}

\begin{lemma}\label{lemma:dkg}
Let the participants in DKG be $L$ parties and a server. 
If $\delta_D + \eta_D < 1/3$, then under the communication model defined in Section~\ref{s:comm-model}, protocol $\Pi_{\text{DKG}}$ (Fig.~\ref{fig:dkg-dropout}) has properties C1, C2, C3, C4 and S in the presence of a malicious adversary controlling the server and up to $\eta_D$ fraction of the parties.
\end{lemma}

\ifthenelse{\boolean{longver}}
{
\heading{Proof.}
Since $L=3\ell+1$, and by $\delta_D+\eta_D< 1/3$, at most $\ell$ are malicious 
(the dropouts can be treated as malicious parties who do not send the prescribed messages).
We first show under which cases the honest parties will have no share.
Note that the parties that are excluded from $\mathcal{D}_2$ are those
  who are honest but did not receive a complaint against a malicious 
  party who performed wrong sharing; 
  and parties that are excluded from $\mathcal{D}_3$ are those 
  who have complained at $i$ in step (b) but did not receive shares from $i$ at this step.
If the server drops messages (sent from one online honest party to another) in the above two cases,
  then in the cross-check step, the honest parties will not receive more than
  $2\ell+1$ QUALs, and hence will abort.
In this case, honest parties in the end has no share.

Now we prove C1 by contradiction.
Suppose there are two honest parties $P_1$ and $P_2$ at the end of the protocol who holds secrets (not abort) and they have different QUAL. 
Then this means $P_1$ received at least $2\ell + 1$ same QUAL sets $S_1$ and $P_2$ received at least same $2\ell+1$ QUAL sets $S_2$.
W.l.o.g., assume that there are $\ell-v$ malicious parties ($v\ge 0$). 
In the $2\ell+1$ sets $S_1$, at least $\ell+1+v$ of them are from honest parties.
Similarly, for the $2\ell+1$ sets $S_2$, at least $\ell+1+v$ are from other honest parties different than above (since an honest party cannot send both $S_1$ and $S_2$). 
However, note that we have in total $2\ell+1+v$ honest parties, which derives a contradiction.

Recall that at most $\ell$ parties are malicious, so the QUAL set has at least $\ell+1$ parties,
  and since we have C1, now C2 is guaranteed. 
Moreover, since QUAL contains at least one honest party, the secret key is uniform (C3). 
C4 follows exactly from the work GJKR.
The proof for S is the same as GKJR, except that the simulator additionally simulates the agreement protocol in Lemma~\ref{lemma:agreement-malicious}.

\heading{Remark.}
An important difference between $\Pi_{\text{DKG}}$ and standard DKG protocols is that 
we allow aborts and allow honest parties to not have shares at the end of the protocol.
When some honest parties do not have shares of the secret key, the server is still able to 
get sum (decryption still works) because malicious parties hold the shares. 
}
{}

\begin{figure*}[h!]
\begin{tcolorbox}[enhanced, boxsep=1mm, left= 0mm, right=0.5mm, title={\textbf{\footnotesize Protocol $\Pi_{\text{DKG}}$ based on discrete log}}]
\linespread{1.4}
\footnotesize
\textbf{Parameters.} A set of $L$ parties (denoted as $\mathcal{D}_0$), threshold $\ell$ where $3\ell+1= L$. $\delta_D + \eta_D < 1/3$. 

\textbf{Protocol outputs.} A subset of the $L$ parties hold secret sharing of a secret key $s \in \mathbb{Z}_q$; the server holds the public key $g^s$ signed by all the clients.

\textbf{Notes.} The parties have access to PKI (Section~\ref{s:detail:pki}). All messages sent from one party to another via the server are signed and end-to-end encrypted. 

    \begin{myenumerate2}[1.]
    \item Each party $i$ performs verifiable secret sharing (VSS) as a dealer:
        \begin{myenumerate2}
            \item \emph{Share}:
            
            Party $i\in \mathcal{D}_0$ randomly chooses $s_i \in \mathbb{Z}_q$, computes 
            $\{s_{i,j}\}_{j=1}^L, \{s_{i,j}'\}_{j=1}^L, \{C_{i,k}\}_{k=0}^\ell  \leftarrow PShare(s_i, \ell, L)$.
            
            It also computes $\{A_k\}_{k=0}^\ell$ from $FShare(s, \ell, L)$ and stores it locally. 
            
            Send $s_{i, j}$ and $s_{i, j}'$ to each party $j$, and $\{C_{i,k}\}_{i=0}^\ell$ to all parties $j\in\mathcal{D}_0$ via the server.
            
            \Comment{Denote the set of parties who received all the prescribed messages after this step as $\mathcal{D}_{1}$.}
            
            \item \emph{Verify and complain}:
            
            Each party $j\in \mathcal{D}_1$ checks whether it received at least $(1-\delta_D)L$ valid signed shares. If not, abort; otherwise continue. 
            
            Each party $j\in \mathcal{D}_1$, for each received share $s_{i, j}$, runs $b \leftarrow PVerify(j, s_{i,j}, s_{i,j}', \{C_{i,k}\}_{k=0}^\ell)$.
            
            If $b$ is 1, then party $i$ does nothing; otherwise party $i$ sends (\texttt{complaints}, $j$) to all the parties
            in $\mathcal{D}_0$ via the server.
            
            \Comment{Denote the set of parties who received all the prescribed messages after this step as $\mathcal{D}_{2}$.}

            \item \emph{Against complaint}: 
        
            Each party $i \in \mathcal{D}_2$, who as a dealer, if received a valid signed (\texttt{complaint}, $i$) from $j$, sends 
            $s_{i,j}, s_{i,j}'$ to all parties in $\mathcal{D}_0$ via the server.
            
            \Comment{Denote the set of parties who received all the prescribed messages after this step as $\mathcal{D}_{3}$.}

            \item \emph{Disqualify}:
              
            Each party $i \in \mathcal{D}_3$ marks any party $j$ as disqualified if 
            it received more than $2\ell+1$ valid signed 
            (\texttt{complaints}, $j$),
            or party $j$ answers with $s_{j,i}, s_{j,i}'$ such that 
            $PVerify(s_{j, i}, s_{j, i}', \{C_{j,k}\}_{k=1}^\ell)$ outputs 0.
            The non-disqualified parties form a set QUAL.
            
            Each party $i\in \mathcal{D}_3$ signs the QUAL set and sends to all parties in $\mathcal{D}_0$ to the server. 
            The server, on receiving a valid signed QUAL, signs and sends it to all parties in $\mathcal{D}_3$. 
            
            \item \emph{Cross-check QUAL}: 
            
            Each party $i\in \mathcal{D}_3$ checks whether it receives at least $2\ell+1$
            valid signed QUAL, if so, they sum up the shares in QUAL and derive a share of secret key. 
            If not, abort.

        \end{myenumerate2}
   
    \item Compute public key:
    
        \begin{myenumerate2}
            \item Each party $i \in \QUAL $ sends $ \{A_{i,k}\}_{k=1}^\ell$ to all parties via the server.
            
            \item Each party $i$ runs $b' \leftarrow FVerify(s_{j,i}, \{A_{j,k}\}_{k=1}^\ell)$ for $j \in \QUAL$.
            If $b'$ is 0, then party $i$ sends to all the parties in $\mathcal{D}_3 \cap \QUAL$ via the server
            a message
            (\texttt{complaint}, $j, s_{j,i}, s_{j,i}'$) 
            for those $(s_{j,i}, s_{j,i}')$ such that $b$ is 1 and $b'$ is 0. 
            
            \Comment{For $b=1$ and $b'=0$: The check in step 1.d) passes but fails this step}
            
            \item For each $j$ such that (\texttt{complaint}, $j, s_{j,i}, s_{j,i}'$) is valid, parties reconstruct $s_j$.
                For all parties in $\QUAL$, set $y_i=g^{s_i}$, and compute $PK=\prod_{i\in QUAL} y_i$.
                Parties in $\QUAL$ sign $PK$ using their own signing keys and send the signed $PK$ to the server. 
        \end{myenumerate2}

    \end{myenumerate2}
\end{tcolorbox}
\caption{Protocol $\Pi_{\DKG}$. 
} 
\label{fig:dkg-dropout}
\end{figure*}

\subsection{Collection phase}
The detailed protocol for each round in the collection phase is described in Figure~\ref{fig:collection}.
At the beginning of round $t$, the server notifies the clients who should be involved,
  namely $S_t$.
A client who gets a notification can download public keys
  of its neighbors $A_t(i)$ from PKI server (the server should tell clients how to map client
IDs to the names in PKI).
To reduce the overall run time, clients can pre-fetch public keys used in the coming rounds.

\begin{figure*}

\begin{tcolorbox}[enhanced, boxsep=1mm, left= 0mm, right=0.5mm, title={\textbf{\footnotesize Collection phase: $\Pi_{\text{sum}}$ for round $t$ out of $T$ total rounds}}]
\linespread{1.4}
\footnotesize
Initial state from setup phase:
each client $i\in [N]$ holds a value $v$ and public key $PK=g^s$ where $SK=s$;
each decryptor $u \in \mathcal{D}$ additionally holds a Shamir share of $SK$ (threshold $\ell$ with $3\ell+1 = L$).
We require $ \delta_D + \eta_D < 1/3$.  

Steps for round $t$: 
\begin{enumerate}[1.]
    \item \textbf{Report step.}
    
    \textbf{Server performs the following:} 
    
    \quad Compute a set $\mathcal{Q}_{graph} \leftarrow \ChooseSet(v, t, n_t, N)$ and a graph $G_t \leftarrow \GenGraph(v, t, \mathcal{Q}_{graph})$; store $\{A_t(i)\}_{i \in \mathcal{Q}_{graph}}$ computed from $G_t$.
    
    \quad Notify each client $i\in \mathcal{Q}_{graph}$ that collection round $t$ begins.
    
    \textbf{Each client $i\in \mathcal{Q}_{graph}$ performs the following:} 
    
    \quad Compute $\mathcal{Q}^{local}_{graph} \leftarrow \ChooseSet(v, t, n_t, N)$, and if $i \not\in \mathcal{Q}^{local}_{graph}$, ignore this round.  

    \quad Read from PKI $g^{a_j}$ for $j\in A_t(i)$, and compute $r_{i, j}$ by computing $(g^{a_j})^{a_i}$ and mapping it to $\{0,1\}^\lambda$. 
    
    \quad Sample $m_{i,t} \xleftarrow{\$} \{0,1\}^{\lambda}$ and compute $\{h_{i, j, t}\}_{j\in A(i)} \leftarrow \PRF(r_{ij}, t)$ for $j\in A_t(i)$.

    \quad Send to server a message $msg_{i,t}$ consisting of 

    \qquad \qquad $Vec_{i,t} = \vec{x}_{i,t} + \PRG(m_{i,t}) + \sum_{j\in A_t(i)} \pm \PRG(h_{i,j,t}), \quad
        \SymEnc(k_{i,u}, m_{i,u,t} \| t), \text{ for } u\in \comset, \quad
        \AsymEnc(PK, h_{i,j,t})$ for $j \in A_t(i)$  

    \qquad \qquad where $m_{i,u,t} \leftarrow Share(m_{i,t},\ell, L)$,
    $A_t(i) \leftarrow \FindNeighbors(v, S_t, i)$, 
    
    \qquad \qquad and $\AsymEnc$ (ElGamal) and $\SymEnc$ (authenticated encryption) are defined in Appendix~\ref{app:primitives}. 
    
    \qquad along with the signatures $\sigma_{i, j, t} \leftarrow Sign(sk_{i}, c_{i,j,t} \| t)$ 
      for all ciphertext $c_{i,j,t} =\AsymEnc(PK, h_{i,j,t}) \ \forall j\in A_t(i)$.

    \item \textbf{Cross check step.}
    
    \textbf{Server performs the following:} 
    
    \quad Denote the set of clients that respond within timeout as $\mathcal{Q}_{vec}$. 
    
    \quad Compute partial sum $\Tilde{z_t} = \sum_{i\in \mathcal{Q}_{vec}} Vec_{i,t}$.
    
    \quad Build decryption request $req$ ($req$ consists of clients in $S_t$ to be labeled):
    
    \quad Initialize an empty set $\mathcal{E}_i$ for each $i\in \mathcal{Q}_{graph}$, and
        
    \quad \quad if $i\in \mathcal{Q}_{vec}$, label $i$ with ``online'', 
    
    \quad \quad \quad and add $\SymEnc(k_{i,u}, m_{i,u,t} \| t)$ to $\mathcal{E}_i$, where $k_{i,j}$ is derived from PKI (Appendix~\ref{app:primitives});
    
    \quad \quad else label $i$ with ``offline'',
    
    \quad \quad \quad and add $\{ (\AsymEnc(PK, h_{i,j,t}), \sigma_{i, j, t}\}_{j\in A_t(i) \cap \mathcal{Q}_{vec})}$ to $\mathcal{E}_i$.
    
    \quad Send to each $u\in \mathcal{D}$ the request $req$ and $\mathcal{E}_i$ of all clients $i\in \mathcal{Q}_{graph}$.

    \textbf{Each \com $u\in \mathcal{D}$ performs the following:} 
    
    \quad Upon receiving a request $req$, compute $\sigma_u^* \leftarrow Sign(sk_u, req\| t)$,
    and send $(req, \sigma_u^*)$ to all other decryptors via the server.

    \item \textbf{Reconstruction step.}
    
    \textbf{Each \com $u\in \mathcal{D}$ performs the following:} 
    
    \quad Ignore messages with signatures ($\sigma_{i,j,t}$ or $\sigma_u^*$) with round number other than $t$.
    
    \quad Upon receiving a message $(req, \sigma_{u'}^*)$, run $b\leftarrow VerSig(pk_i, req, \sigma_{u'}^*)$. 
    Ignore the message if $b= 0$. 
    
    \quad Continue only if $u$ received $2\ell+1$ or more same $req$ messages that were not ignored. Denote such message as $req^*$. 
    
    \quad For $req^*$, continue only if 
    
    \quad \quad each client $i \in S_t$ is either labeled as ``online'' or ``offline'';
    
    \quad \quad the number of ``online'' clients is at least $(1-\delta) n_t$;
    
    \quad \quad all the ``online'' clients are connected in the graph;

    \quad \quad each online client $i$ has at least $k$ online neighbors such that $\eta^k < 2^{-\kappa}$.

    \quad For each $i \in \mathcal{Q}_{graph}$, 

    \quad \quad For each $\SymEnc(k_{i,u}, m_{i,u,t} \| t)$ in $\mathcal{E}_i$, use $k_{i,u}$ (derived from PKI) to decrypt; send $m_{i,u,t}$ to the server if the decryption succeeds;
    
    \quad \quad For each $(\AsymEnc(PK, h_{i,j,t}), \sigma_{i,j,t}) \in \mathcal{E}_i$, parse as $((c_0, c_1), \sigma)$ and
    send $c_0^{s_u}$ to the server if $VerSig((c_0, c_1), \sigma)$ outputs 1;

    \textbf{Server completes the sum:} 
    
    \quad Denote the set of decryptors whose messages have been received as $U$. Compute
    a set of interpolation coefficients $\{\beta_u\}_{u\in U}$ from $U$. 
    
    \quad For each $i \in \mathcal{Q}_{graph}$,
           reconstruct the mask $m_{i,t}$ or $\{h_{i,j,t}\}_{j\in A_t(i) \cap \mathcal{Q}_{vec}}$:

    \qquad For each parsed $(c_0, c_1)$ meant for $h_{i,j,t}$ in $\mathcal{E}_i$, 
           compute $h_{i,j,t}$ as $c_1\cdot (\prod_{u\in U} (c_0^{s_u})^{\beta_u})^{-1} $;

    \qquad For each set of received shares $\{m_{i,u,t}\}_{u\in U}$, compute $m_{i,t}$ as $Recon(\{ m_{i,u,t}\}_{u\in U})$. 

    \quad Output $z_t = \Tilde{z_t} - \PRG(m_{i, t}) + \sum_{j\in A_t(i) \cap \mathcal{Q}_{vec}} \pm \PRG(h_{i,j,t})$.

\end{enumerate}
\end{tcolorbox}
\caption{Collection protocol $\Pi_{\text{sum}}$.
}\label{fig:collection}
\end{figure*}

\subsection{Transfer shares}\label{app:transfer}

Every $R$ rounds, the current set of decryptors $\comset$ transfer shares of $SK$
  to a new set of decryptors, $\comset_{new}$.
To do so, each $u\in \comset$ computes a destination \coms set $\comset_{new}$ for round $t$,
   by $\ChooseSet(v, \ceil{t/R}, L, N)$. 
Assume now each decryptor $u \in \comset$ holds share $s_u$ of $SK$ (i.e., there is a polynomial $p$ such that
  $p(u) = s_u$ and $p(0) = SK$).
To transfer its share, 
  each $u\in \comset$ acts as a VSS dealer exactly the same as the first part in $\Pi_{\text{DKG}}$
  (Fig.\ref{fig:dkg}) to share $s_u$ to new \com $j\in \comset_{new}$.
In detail, $u$ chooses a polynomial $p^*_u$ of degree $\ell$ and sets $p^*_u(0)= s_u$ and all other 
  coefficients of $p^*_u$ to be random.
Then, $u$ sends $p^*_u(j)$ to each new \com $j \in \comset_{new}$.

Each new \com $j\in \comset_{new}$ receives the evaluation of the polynomials at point $j$ (i.e., $p^*_u(j)$ for all $u \in \comset$).
The new share of $SK$ held by $j$, $s_j'$, is defined to be a linear combination of the received shares:
$s_j' := \sum_{u \in \comset}  \beta_u \cdot p^*_u(j)$, where the combination coefficients $\{\beta_u\}_{u \in \comset}$ are constants 
  (given the set $\mathcal{D}$, we can compute $\{\beta_u\}_{u \in \comset}$).
Note that the same issue about communication model for DKG also exists here, 
  but the same relaxation applies.

Here we require that $\eta_D + \delta_D < 1/3$ for both $\mathcal{D}$ and $\mathcal{D}_{new}$. 
As a result, each client $j$ in a subset $\mathcal{D} \subseteq \comset_{new}$ holds
            a share $s_{u, j}$.
For each receiving \com $j \in \comset$, it computes
  $s_j' = \sum_{u \in \comset_{old}} \beta_u \cdot s_{u, j}$, 
  where each $\beta_u$ is some fixed constant.

\ifthenelse{\boolean{longver}}
{Since the sharing part is exactly the same as $\Pi_{\text{DKG}}$ and the share combination happens locally,
  the same correctness and security argument of DKG applies. 
Specifically, for correctness, each honest party either has no secret at the end or
agrees on the same QUAL with other honest parties; and the QUAL defines the unique same secret key
before resharing.
For security, a malicious adversary will not learn any information, but can cause the protocol aborts.}
{}

\section{Requirements on Parameters}\label{app:param}

\ifthenelse{\boolean{longver}}
{
\subsection{The number of decryptors}\label{app:nbr-dec}
In this section, we show how to choose $L$ such that, given $N, \eta, \delta_D$, 
  we can guarantee less than $1/3$ of the $L$ chosen decryptors are malicious.
Note that $\delta_D$ is given because this can be controlled by the server, i.e.,
  the server can decide how long to wait for the decryptors to respond. 
On the other hand, $\eta_D$ is a random variable.
  
To guarantee $2\delta_D + \eta_D < 1/3$ (Theorem~\ref{thm:security-main}),
  a necessary condition is that $\eta < 1/3 - 2\delta_D $.
This can be formalized as a probability question:
  given $N$ clients where $\eta$ fraction of them are malicious, 
  randomly sample $L$ clients (decryptors) from them; 
  the number of malicious clients $X$ in the decryptors should follow
  the tail bound of hypergeometric distribution~\cite[Section 2]{bell20secure}:
\[\Pr[X  \ge (\eta + (1/3 - 2\delta_D -\eta ))L  ] \le   e^{-2L(1/3 - \eta- 2\delta_D)^2},\]

For $\eta$ and $\delta_D$ both being 1\%, the choice of $L=60$ (which we used for 
  benchmarks in \S\ref{s:perf}) gives $1.6\times 10^{-5}$ probability.
If we double $L$ to be 120, then this guarantees $2.6 \times 10^{-10}$ probability.
}
{
\heading{The number of decryptors.}
The full version~\cite{ma23flamingo} gives detailed analysis; but briefly, for $\eta$ and $\delta_D$ both being 1\%, 
the choice of $L=60$ (which we used for benchmarks in \S\ref{s:perf}) gives $1.6\times 10^{-5}$ probability 
(that more than $1/3$ selected decryptors are malicious) 
and $L=120$ gives $2.6\times 10^{-10}$ probability.
}

\ifthenelse{\boolean{longver}}
{
\subsection{Proof of Lemma~\ref{lemma:graph-connectivity}}
The algorithm in Figure~\ref{fig:gengraph} gives a random graph 
$G(n, \epsilon)$.
A known result in random graphs is, when the edge probability $\epsilon > \frac{(1+\omega) \ln n}{n}$,
  where $\omega$ is an arbitrary positive value, the graph is almost surely connected when $n$ is large.
Note that in BBGLR, they also use this observation to build the graph that has significant less number of neighbors
  than a prior work by Bonawitz et al.~\cite{bonawitz17practical},
  but in their work since the clients chooses their neighbors, the resulting graph is a biased one; and they guarantee that
  there is no small isolated components in the graph.  
  
Concretely, from Gilbert~\cite{gilbert59random}, let $g(n, \epsilon)$ be 
the probability that graph $G(n,\epsilon)$ has disconnected components, and it
can be recursively computed as 
\[
g(n, \epsilon)=1-\sum_{i=1}^{n-1}  \binom{n-1}{i-1} g(i, \epsilon) (1-\epsilon)^{i(n-i)}.
\]

\begin{figure}[h!]
{\footnotesize
\begin{tabular*}{\columnwidth}{@{\extracolsep{\fill}}l r r r r }
\toprule
Number of nodes $n$ & 128 & 512 & 1024\\
\midrule
Parameter $\epsilon$ (failure probability $10^{-6}$) &  0.11 & 0.03 & 0.02 \\
\midrule
Parameter $\epsilon$ (failure probability $10^{-12}$) & 0.25 & 0.06 & 0.03\\
\bottomrule
\end{tabular*}
}
\caption{\small Parameters $\epsilon$ to ensure random graph connectivity.}\label{fig:graph-param}
\end{figure}
We numerically depict the above upper bound of the probability $g(n, \epsilon)$
  for different $\epsilon$ in Figure~\ref{fig:graph-param}. 
For example, when $n=1024$, to ensure less than $10^{-6}$ failure probability, 
  we need $\epsilon \ge 0.02$, hence the number of neighbors a client needs when $\delta = \eta=0.01$
  is at least $ \lceil (\epsilon + \delta + \eta) n \rceil = 41$.  
}
{
\heading{The number of neighbors.}
The full version~\cite{ma23flamingo} gives detailed analysis; but briefly, for 1K clients, $\epsilon = 0.02$ guarantees $10^{-6}$ probability (that the graph is disconnected) and $\epsilon = 0.03$ guarantees $10^{-12}$ probability. 
For example, in the former case, when $\eta$ and $\delta$ are both 1\%, each client needs $ \lceil n_t(\epsilon + \delta+\eta) \rceil = 41 $ neighbors.  
}

\ifthenelse{\boolean{longver}}
{

\section{Security Proofs}\label{app:proofs}

\subsection{Security definition}\label{app:security-def}

We say \emph{a protocol $\Pi$ securely realizes ideal functionality $\mathcal{F}$ 
in the presence of a malicious adversary $\Adv$}
if there exists a probabilistic polynomial time algorithm, or simulator, $\Sim$,
such that for any probabilistic polynomial time $\Adv$, the distribution
of the output of the real-world execution of $\Pi$ is (statistically or computationally) indistinguishable 
from the distribution of the output of the ideal-world execution invoking $\mathcal{F}$:
  the output of both worlds' execution includes the inputs and outputs of honest parties
  the view of the adversary $\Adv$. 
In our proof, we focus on the computational notion of security.
Note that $\Sim$ in the ideal world has one-time access to $\mathcal{F}$,
  and has the inputs of the corrupted parties controlled by $\Adv$.

\subsection{Ideal functionalities}\label{app:ideal}
We provide ideal functionality for \sys{} in Figure~\ref{fig:F-ideal-mal-weak}.
Looking ahead in the proof, we also need to define an ideal functionality for
  the setup phase and an ideal functionality for a single round in the collection phase,
  which we give as Figure~\ref{fig:Fsetup} and Figure~\ref{fig:Fsumt-weak}.  

We model the trusted source of initial randomness as a functionality $\Frand$;
  that is, a party or ideal functionality on calling $\Frand$ will
  receive a uniform random value $v$ in $\{0,1\}^\lambda$, where $\lambda$ is the
  length of the PRG seed.

\begin{figure}
 \begin{tcolorbox}[enhanced, boxsep=1mm, left= 0mm, right=0.5mm, title={\textbf{\footnotesize Functionality $\Fsumt^{(t)}$}}]
    \linespread{1.4}
    \footnotesize
    
    Parties: clients in $S_t$ and a server.
    
    Parameters: dropout rate $\delta$ and malicious rate $\eta$ over $n_t$ clients.
    
    \begin{myitemize2} 

        \item $\Fsumt^{(t)}$ receives a set $\mathcal{O}_{t}$ 
        such that $|\mathcal{O}_{t}| / |S_t| \le \delta$,
        and from the adversary $\mathcal{A}$ a set of corrupted parties, $\mathcal{C}$; 
        and $\vec{x}_{i, t}$ from client $i\in S_t \backslash  (\mathcal{O}_t \cup \mathcal{C})$.
          
        \item $\Fmal$ sends $S_t$ and $\mathcal{O}_t$ to $\Adv$, and asks $\Adv$ for a set $M_t$: if $\Adv$ replies with $M_t \subseteq S_t\backslash \mathcal{O}_t$ such that $|M_t|/|S_t| \ge 1-\delta$, then $\Fmal$ computes $\vec{z}_t = \sum_{i\in M_t \backslash \mathcal{C}} \vec{x}_{i,t}$ and continues; otherwise $\Fmal$ sends \texttt{abort} to all the honest parties.

        \item Depending on whether the server is corrupted by $\Adv$:
        \begin{myitemize2}
            \item If the server is corrupted by $\Adv$, then $\Fmal$ outputs $\vec{z}_t$ to all the parties corrupted by $\Adv$.
            \item If the server is not corrupted by $\Adv$, then $\Fmal$ asks $\Adv$ for a shift $\vec{a}_t$ and outputs $\vec{z}_t + \vec{a}_t$ to the server.
        \end{myitemize2}

    \end{myitemize2}
    
    \end{tcolorbox}
\caption{Ideal functionality for round $t$ in collection phase.}
\label{fig:Fsumt-weak}
\end{figure}

\subsection{Proof of Theorem~\ref{thm:security-setup}}\label{app:proof-setup}

The ideal functionality $\mathcal{F}_{\text{setup}}$ for the setup phase is defined in Figure~\ref{fig:Fsetup}. 
Depending on whether the server is corrupted or not, we have the following two cases. 
\begin{myenumerate2}
    \item When the server is not corrupted, then the communication model is equivalent
    to a secure broadcast channel. By security of GJKR, $\Pi_{\text{setup}}$ securely realizes $\mathcal{F}_{\text{setup}}$. 
    \item When the server is corrupted, we build a simulator for the adversary $\Adv$.
    We start by listing the messages that $\Adv$ sees throughout the setup phase:
        \begin{myitemize}
            \item A random value $v$ from $\Frand$;
            \item All the messages in $\Pi_{\DKG}$ that are sent via the server;
            \item All the messages in $\Pi_{\DKG}$ that are seen by the corrupted decryptors. 
        \end{myitemize}
    The simulator first calls $\Frand$ and receives a value $v$.
    Then the simulator interacts with $\Adv$ acting as the honest decryptors.
    The simulator aborts if any honest decryptors would abort in $\Pi_{\DKG}$. 
    There are two ways that $\Adv$ can cheat: 1) $\Adv$ cheats in $\Pi_{\DKG}$, and
    in Appendix~\ref{app:dkg}, we show the simulator can simulate the view of $\Adv$;
    2) $\Adv$ cheats outside $\Pi_{\DKG}$, this means $\Adv$ chooses a wrong set of
    decryptors, or it broadcasts wrong signatures. So the simulator aborts as long as it
    does not receive $2\ell+1$ valid signatures on $PK$s signed by the set defined by $v$. 
\end{myenumerate2}
Finally, note that our threat model (\S\ref{s:threat}) assumes that the server
  is also controlled by the adversary, i.e., the first case is not needed here;
  but it will be useful when we analyze the robust version in Appendix~\ref{app:robust-protocol}.

\subsection{Proof of Theorem~\ref{thm:dropout-resilience}}\label{app:proof-dropout-resilience}
The proof for dropout resilience is rather simple:
  in the setup phase, at most $\delta_D$ fraction of $L$ selected decryptors drop out;
  then in one round of the collection phase, another $\delta_D$ fraction of decryptors can drop out.
Since $2\delta_D+\eta_D<1/3$, and $3\ell +1= L$ (Fig.~\ref{fig:dkg-dropout}), the online decryptors can always help the server to reconstruct
  the secrets.

\subsection{Proof of Theorem~\ref{thm:security-main}}\label{app:proof-security-flamingo}

We first present the proof for a single round: the collection protocol $\Pi_{\text{sum}}$ (Fig.~\ref{fig:collection})
  for round $t$ securely realizes the ideal functionality $\mathcal{F}_{\text{sum}}^{(t)}$ (Fig.~\ref{fig:Fsumt-weak})
  in the random oracle model.
From the ideal functionality $\mathcal{F}_{\text{sum}}^{(t)}$ we can see that the output sum is not determined by the actual dropout set $\mathcal{O}_t$, but instead $M_t$ sent by the adversary.

In the proof below for a single round, for simplicity, we omit the round number $t$ 
  when there is no ambiguity. 
We assume the adversary $\Adv$ controls a set of clients in $[N]$, with the constraint
  $2\delta_D + \eta_D < 1/3$. 
Denote the set of corrupted clients in $[N]$ as $\mathcal{C}$
  and as before the set of the decryptors is $\mathcal{D}$;
  the malicious decryptors form a set $\mathcal{C} \cap \mathcal{D}$ 
  and $|\mathcal{C} \cap \mathcal{D}| < L / 3$.
From the analysis in Appendix~\ref{app:failure-model}, we have $|\mathcal{C}| / |S_t| \le \eta$. 

\medskip 

\heading{Case 1.} We start with the case where the server is corrupted by $\Adv$. 
Now we construct a simulator $\Sim$ in the ideal world that runs $\Adv$ as 
  a subroutine.
We assume both the simulator $\Sim$ and ideal functionality $\Fsumt^{(t)}$ (Fig.~\ref{fig:Fsumt-weak})
  have access to an oracle $\mathcal{R}_{\text{drop}}$ that provides the dropout sets $\mathcal{O}_t$. 
In other words, the dropout set is not provided ahead of the protocol but instead provided
  during the execution (similar notion appeared in prior work~\cite{bonawitz17practical}). 
Assume that in the ideal world, initially a secret key $SK$ is shared among
  at least $2L/3$ clients in $\mathcal{D}$. 
The simulation for round $t$ is as follows. 
\begin{enumerate}
    \item $\Sim$ received a set $\mathcal{O}_t$ from the oracle $\mathcal{R}_{\text{drop}}$. 
    \item $\Sim$ receives a set $M_t$ from the adversary $\Adv$. 
    \item $\Sim$ obtains $\vec{z}_t$ from $\Fsumt^{(t)}$. 
    
    \item (Report step) $\Sim$ interacts with $\Adv$ as in the report step acting 
          as the honest clients $i\in M_t \backslash \mathcal{C}$
          with masked inputs $\vec{x}_i'$, 
          such that $\sum_{i\in M_t \backslash \mathcal{C}} \vec{x}_i'
          = \vec{z}_t$.

          Here the input vector $\vec{x}_i'$ and the mask $m_{i}$ are
          generated by $\Sim$ itself, and the pairwise secrets are obtained by querying the PKI. 
    \item (Cross-check step) $\Sim$ interacts with $\Adv$ acting as honest decryptors as in the cross-check step. 
    
    \item (Reconstruction step) $\Sim$ interacts with $\Adv$ acting as honest decryptors in the reconstruction step,
          where $\Sim$ uses the shares of the secret key $SK$ to perform decryption of honest parties. 
          
    \item In the above steps, if all the honest decryptors would abort in the protocol
          prescription then $\Sim$ sends \texttt{abort} to $\Fsumt^{(t)}$, 
          outputs whatever $\Adv$ outputs and halts. 

\end{enumerate}

We construct a series of hybrid execution, starting from the real world to the ideal world execution. 

\heading{\small Hybrid 1.}
The view of $\Adv$ in the real world execution is the same as the view of $\Adv$ in the ideal world when $\Sim$ would have the actual inputs of honest parties, $\{\vec{x}_i\}_{i\in S_t \backslash (\mathcal{C}\cup \mathcal{O}_t)}$, the pairwise and individual masks, and the shares of the secret key $SK$. ($\Sim$ in fact would know $SK$ in full 
  because $3\ell + 1 = L$ and that the number of honest parties is $2\ell + 1$ or more.)

\heading{\small Hybrid 2.}
$\Sim$ now instead of using the actual secret key $s$, it replaces $s$ with 0
and sends the corresponding $|\mathcal{C}\cap \mathcal{D}| < L/3 $ shares of 0 in $\mathbb{Z}_q$ to $\Adv$. 
The joint distribution of less than $L/3$ shares (recall that the threshold is $\ell$ where $3\ell+1=L$),
  from the property of Shamir secret sharing, for $s$ and $0$ are the same.
Hence this hybrid has identical distribution to the previous hybrid.

\heading{\small Hybrid 3.}
$\Sim$ now instead of using the actual pairwise masks between honest parties,
  it samples a random pairwise mask $r_{i,j}'$ from $\{0,1\}^\lambda$ and computes 
  the corresponding ElGamal ciphertext as $(c_0', c_1')$. 
$\Sim$ does not change the pairwise mask between a client controlled by $\Adv$ and an honest client
  (such pairwise mask can be obtained by querying PKI to get $g^{a_j}$ for an honest client $j$,
  and compute $(g^{a_j})^{a_i}$ for malicious client $i$).   
We argue that $\Adv$'s view in this hybrid is
  computationally indistinguishable from the previous one as below.

First, we need to assume the mapping from $\mathbb{G}$
  to $\{0,1\}^{\lambda}$ is a random oracle.
To specify, 
  in the real world, the mask $r_{i,j}$ is computed from the mapping on $g^{a_i a_j} $; 
  and in the ideal world the mask $r_{i,j}'$ is randomly sampled.
Let $M_t$ be the set of online clients that $\Adv$ labels in the real world (recall the server is controlled by $\Adv$). 
$\Adv$ in both worlds observes $\PRF(r_{i,j}, t)$ between a client $i$ out of $M_t$
  and a client $j$ in $M_t$, hence we require $r_{i,j}$ to be random as a PRF key. 

Second, $\Adv$ in the ideal world does not observe the pairwise masks between clients
  in $M_t$, but only the ciphertexts generated from $r_{i,j}'$ for those clients;
  and the distribution of the ciphertexts is computationally indistinguishable 
  from what $\Adv$ observed from the real world by the security of ElGamal encryption (Definition~\ref{def:elgamal}). 

\heading{\small Hybrid 4.}
$\Sim$ now instead of using symmetric encryption ($\SymEnc$) of the shares of the actual individual mask $m_i$,
  it uses the symmetric encryption of a randomly sampled $m_i'$ from $\{0,1\}^\lambda$ as the individual mask.
Looking ahead in the proof, we also need to model the PRG as a random oracle $\mathcal{R}_{\PRG}$ that can
  be thought of as a ``perfect PRG'' (see more details in a prior work~\cite{bonawitz17practical}).
For all $i\in M_t/ \mathcal{C}$, $\Sim$ samples $Vec_i$ at random and programs $\mathcal{R}_{\PRG}$ to set $\PRG(m_i')$ as 
\[ \PRG(m_i') = Vec_i - \vec{x}_i - \sum \pm\PRG(r_{i,j}'),\]
where the vectors $Vec_i$'s are vectors observed in the real world. 
The view of $\Adv$ regarding $Vec_i$'s in this hybrid is statistically indistinguishable to that in the previous hybrid.

Moreover, $\Adv$ learns the $m_i$ in the clear for those $i \in M_t$ in both worlds, 
  and the distributions of those $m_i$'s in the ideal and real world are identical; 
for those $m_i$'s where $i\not\in M_t$ that $\Adv$ should not learn, 
from the semantic security of the symmetric encryption scheme (Definition~\ref{def:authenc}), 
  and the threshold $\ell<L/3$, $\Adv$'s view in this hybrid is
  computationally indistinguishable from the previous one.

\heading{\small Hybrid 5.}
$\Sim$ now instead of programming the oracle $\mathcal{R}_{\PRG}$ as in the previous hybrid,
  it programs the oracle as
  \[ \PRG(m_i') = Vec_i - \vec{x}_i' - \sum \pm\PRG(r_{i,j}'),\]
  where $\vec{x}_i'$'s are chosen such that 
  $\sum_{i\in M_t \backslash \mathcal{C}} \vec{x}_i = \sum_{i\in M_t \backslash \mathcal{C}} \vec{x}_i'$.
From Lemma 6.1 in a prior work~\cite{bonawitz17practical} under the same setting,
   the distribution of $\Adv$'s view in this hybrid is statistically indistinguishable to that in the previous hybrid,
   except probability $2^{-\kappa}$:
   if the graph becomes disconnected or there is an isolated node
  after removing $\mathcal{O}_t$ and $\mathcal{C}$ from $S_t$,
  then the server learns $x_i$ in the clear and thus can 
  distinguish between the two worlds.
When $\Adv$ cheats by submitting $M_t$ to $\Sim$ where the graph formed by nodes
  in $M_t$ is not connected, $\Sim$ simulates 
  honest decryptors and output abort. 
In this case, the distribution of $\Adv's$ view in the ideal world 
  is the same as that in the real world.

\heading{\small Hybrid 7.}
Same as the previous hybrid, except that the label messages $req$ from honest decryptors in the last step 
  are replaced with the offline/online labels obtained from the oracle. 
In all steps, $\Adv$ would cheat by sending invalid signatures to $\Sim$; in this case $\Sim$ will abort.
In the cross-check and reconstruction steps, there are following ways that $\Adv$ would cheat here: 
\begin{enumerate}
    \item $\Adv$ sends multiple different $M_t$'s to $\Fsumt$. $\Sim$ in the ideal world will simulate
    the protocol in Lemma~\ref{lemma:agreement-malicious} below, 
    and outputs whatever the protocol outputs. 
    \item $\Adv$ sends to $\Fsumt$ a set $M_t$ with less than $(1-\delta)n_t$ clients,
          or the clients in $M_t$ are disconnected,
          or there is a client in $M_t$ with less than $\eta^k$ online neighbors.
          In this case, $\Sim$ will abort, which is the same as in the real-world execution. 
\end{enumerate}
 
The last hybrid is exactly the ideal world execution.
To better analyze the simulation succeeding probability, we use $\kappa_1$ to denote the
  security parameter for the graph connectivity (Lemma~\ref{lemma:graph-connectivity}) and use $\kappa_2$ to denote
  the security parameter for the third checking in the cross-check round (\S\ref{s:detail:collection}).
The simulation can fail in two ways:
1) The graph gets disconnected (even when the server is honest);
2) There exists a client in $S_t$ such that all of its online neighbors are malicious. 
The former happens with probability $2^{-\kappa_1}$.
The latter is bounded by $ n \cdot 2^{-\kappa_2}$:
  the probability that the opposite event of 2) happens is $(1-\eta^k)^n \approx 1-n\eta^k $ (assuming $\eta^k$ is very small).
  Thus the failure probability $ n \eta^k \leq n\cdot 2^{-\kappa_2}$. 

\heading{Case 2.} For the case where the server is not corrupted by $\Adv$, the simulation is the same as Case 1, except that the simulator needs to compose the ``shifts'' added by $\Adv$ in each step to hit the value $\vec{a}_t$. 

\medskip

This completes the proof that for any single round $t\in[T]$, the protocol $\Pi_{\text{sum}}$ for round $t$
securely realizes $\Fsumt^{(t)}$ when $\delta_D+\eta_D<1/3$,
  except probability $2^{-\kappa_1} + n \cdot 2^{-\kappa_2} \le n \cdot 2^{-\kappa+1}$, where $\kappa = \min\{\kappa_1, \kappa_2\}$.

\begin{lemma}\label{lemma:agreement-malicious}
    Assume there exists a PKI and a secure signature scheme; 
    there are $3\ell+1$ parties with at most $\ell$ colluding malicious parties. 
    Each party has an input bit of 0 or 1 from a server.
    Then there exists a one-round protocol for honest parties to decide 
    if the server sent the same value to all the honest parties. 
\end{lemma}

\begin{proof}
If an honest party receives $2\ell+1$ or more messages with the same value, then it means
  the server sends to all honest parties the same value. 
If an honest party receives less than $2\ell+1$ messages with the same value,
  it will abort; in this case the server must have sent different messages to different honest parties. 
\end{proof}

\heading{Remark.}
The above analysis of the agreement protocol shows where the threshold $1/3$ comes from.
Consider the case where the threshold is 1/2 and $2\ell+1 = L$. 
For a target client, the (malicious) server can tell $\ell/2$ decryptors that the client is online
  and tell another $\ell/2+1$ decryptors that the client is offline.
Then combined with the $\ell$ malicious decryptors' shares, the server has $\ell/2 + \ell$ shares to reconstruct
  individual mask, and $\ell/2+1 + \ell$ shares to reeconstruct the pairwise mask.

\heading{Multi-round security.}
Our threat model assumes that $\Adv$ controls $\eta N$ clients throughout $T$ rounds (\S\ref{s:threat}).
There are two things we need to additionally consider on top of the single-round proof:
  1) the set $S_t$ is generated from $\PRG$, and 2) the pairwise mask $h_{i, j, t}$ computed from $\PRF(r_{i,j}, t)$. 
For the former, we program $\mathcal{R}_{\PRG}$ (like the single-round proof) such that the 
  $\ChooseSet$ outputs $S_t$. 

Now we analyze the per-round pairwise masks. 
Let the distribution of the view of $\Adv$ in round $t$ be $\Delta_t$. 
We next show that if there exists an adversary $\mathcal{B}$, and two round number 
  $t_1, t_2\in [T]$ such that $\mathcal{B}$ can distinguish 
  between $\Delta_{t_1}$ and $\Delta_{t_2}$,
  then we can construct an adversary $\mathcal{B}'$ who can break PRF security.
We call the challenger in PRF security game simply as challenger. 
There are two worlds (specified by $b=0$ or $1$) for the PRF game.
When $b=0$, the challenger uses a random function; when $b=1$, the
  challenger uses PRF and a random key for the PRF.
We construct $\mathcal{B}'$ as follows.
On input $t_1, t_2$ from $\mathcal{B}$, $\mathcal{B}'$ asks challenger for
  $h_{i,j,t_1}$ for all clients $i$ and $j$, and round $t_1, t_2$.
Then $\mathcal{B}'$ creates the messages computed from $h_{i,j,t}$'s as 
  protocol $\Pi_{\text{sum}}$ prescribed; it generates 
  two views $\Delta_{t_1}, \Delta_{t_2}$ and sends to $\mathcal{B}$. 
$\mathcal{B'}$ outputs whatever $\mathcal{B}$ outputs.

\heading{Failure probability for $T$ rounds.}
For a single round, we already showed that protocol $\Pi_{\text{sum}}^{(t)}$
  securely realizes $\Fsumt^{(t)}$ except probability $p = n \cdot 2^{-\kappa+1}$. 
The probability that for all the $T$ rounds the protocol is secure is therefore $1-(1-p)^T$,
  which is approximately $1- T\cdot p$ when 
  $T \cdot p \ll 1$. 
Therefore, the probability of failure (there exists a round that fails the simulation) is $Tn 2^{-\kappa+1}$.

\ifthenelse{\boolean{longver}}{%
\section{Extension for Robustness}
\subsection{Robust protocol}\label{app:robust-protocol}

As discussed in Section~\ref{s:enhanced}, we want to additionally detect incorrect output when the server is honest.
Figure~\ref{fig:F-ideal-mal} gives the ideal functionality of our robust version (detect-and-abort). 
Specifically, the protocol $\Pi_{\text{sum-robust}}$ has two key properties that $\Pi_{\text{sum}}$ does not have:
\begin{myenumerate2}
    \item the inputs of clients are fixed after being submitted in the report step (i.e., malicious \coms have no way to change the sum of those inputs); 
    \item when the server is honest, the incorrect output caused by malicious clients can always be detected by the server.
\end{myenumerate2}
Moreover, we present a proof-of-decryption technique which allows the server to identify malicious decryptors.

We next describe key techniques we use in $\Pi_{\text{sum-robust}}$; the full description is given as Figure~\ref{fig:collection-robust} where we highlight the changes from $\Pi_{\text{sum}}$. 

\heading{Checking aggregated masks.}
One can make a blackbox use of a technique in a very recent work by Bell et al.~\cite{bell22acorn} published after our work. To specify, each client $i$ computes the sum of its (expanded) masks, denoted as $\vec{s}_i$ which has the same length as its input vector. The client also generates a \emph{proof of masking} that it correctly computed $\vec{x}_i + \vec{s}_i$. Each client $i$ sends to the server the commitment to $\vec{s}_i$ and the proof of masking, besides what is supposed to send in protocol $\Pi_{\text{sum}}$. 

The server first verifies the received proofs of masking, and then reconstruct the masks exactly as in $\Pi_{\text{sum}}$. Denote the sum of the recovered mask as $\vec{s}$. For those clients whose proofs are valid, the server and the clients engage in a \emph{distributed key correctness} (DKC) protocol~\cite[Section 5]{bell22acorn}, where the server checks if recovered mask $\vec{s}$ is the sum of the $\vec{s}_i$'s underlying the commitments. If the server in the DKC protocol aborts, then it also aborts in $\Pi_{\text{sum-robust}}$; otherwise the server continues executing the rest of steps as in $\Pi_{\text{sum}}$.

Now that the server can detect incorrect output; but it does not identify the malicious clients or decryptors. If desired, we can further detect malicious decryptors using the following two techniques. 

\heading{ElGamal encryption for both types of seeds.}
Instead of using symmetric encryption for the shares of $m_{i,t}$, each client
  encrypts $m_{i,t}$ directly (but not the shares) using ElGamal encryption.
This reduces the number of ciphertexts appended to the masked vector compared to $\Pi_{\text{sum}}$;
  but more importantly, later we will see how it combines with
  another technique (zero-knowledge proof for discrete log) 
  to bring robustness.
See the details in the report step of Figure~\ref{fig:collection-robust}. 
This change in the report step will bring a change in the reconstruction step that
  the decryptors do threshold decryption on ciphertext $m_{i,t}$, 
  instead of decrypting the shares of $m_{i,t}$
  using symmetric keys. 

\heading{Proof of decryption.} 
In the reconstruction step, in addition to decrypting,
each decryptor $u$, for each ciphertext $(c_0, c_1)$, proves to the server that 
\[ \log_{c_0} (c_0)^{s_u} = \log_g g^{s_u} ,\]
  where $g^{s_u}$ is known to the server. 
Note that $g^{s_u}$ can be generated in the DKG along the way in the setup phase and stored by the server.

To prove that $\log_{c_0} (c_0)^{s_u} = \log_g g^{s_u} $,
each \com chooses a random $\beta \in \mathbb{Z}_q$ 
  and sends to the server $c_0^{\beta}, g^{\beta}$. 
The server then sends to the \com a challenge $e \in \mathbb{Z}_q$.
The \com computes $z = s_u \cdot e + \beta$
and sends $z$ to the server.
The sever checks that $(c_0^{s_u})^e \cdot (c_0)^\beta = c_0^z$ and $(g^{s_u})^e \cdot g^{\beta} = g^z$.
We use Fiat-Shamir transform to make the proof non-interactive. 

Due to this proof of decryption, the server can exclude
  bogus partial decryptions from malicious decryptors;
  since the sharing polynomial for secret key $s$ is of degree $\ell$
  and there are at least $2\ell+1$ honest online decryptors, 
  the server is always able to reconstruct the result.

\subsection{Security proof}\label{app:proof-enhance}
 
\begin{theorem}[Security of \sys extension]\label{thm:security-extend}
Let $\Psi_T$ be the sequential execution of $\Pi_{\text{setup}}$ (Fig.\ref{fig:protocol-setup}) and the $T$ rounds of $\Pi_{\text{sum-robust}}$ (Fig.~\ref{fig:collection-robust}).
Let $\kappa$ be a security parameter.
Let $\delta, \delta_D, \eta, \eta_D$ be threat model parameters as defined (\S\ref{s:param:setup},\S\ref{s:param:collection}). 
Let $\epsilon$ be the graph generation parameter (Fig.\ref{fig:gengraph}). 
Let $N$ be the total number of clients and $n$ be the number of clients in summation in each round.
Assuming the existence of a PKI, a trusted source of initial randomness, a PRG, a PRF, an asymmetric encryption $\AsymEnc$, a symmetric $\SymEnc$, and a signature scheme, if $2\delta_D + \eta_D < 1/3$ and $\epsilon \ge \epsilon^*(\kappa)$ (Lemma~\ref{lemma:graph-connectivity}),
then under the communication model defined in \S\ref{s:comm-model}, protocol $\Psi_T$ securely computes functionality $\Fmalrb$ (Fig.\ref{fig:F-ideal-mal}) 
in the presence of a static malicious adversary controlling $\eta$ fraction of the total $N$ clients (and the server) and per-round $\eta$ fraction of $n$ clients, except probability 
$T n \cdot 2^{-\kappa+1}$.
\end{theorem}

\begin{figure}
 \begin{tcolorbox}[enhanced, boxsep=1mm, left= 0mm, right=0.5mm, title={\textbf{\footnotesize Functionality $\Fmalrb$}}]
    \linespread{1.4}
    \footnotesize
    
    Parties: clients $1, \ldots, N$ and a server.
    
    Parameters: corrupted rate $\eta$, dropout rate $\delta$.
    
    \begin{myitemize2}

        \item $\Fmalrb$ receives from the adversary $\Adv$ a set of corrupted parties,
                denoted as $\mathcal{C} \subset[N]$, where $|\mathcal{C}|/N \le \eta$.
        \item For each round $t\in [T]$:
            \begin{enumerate}
                \item  $\Fmalrb$ receives a random subset $S_t\subset [N]$, a set of dropout clients $\mathcal{O}_t \subset S_t$, 
                where $|\mathcal{O}_t|/|S_t| \le \delta$ and $|\mathcal{C}|/|S_t| \le \eta$,
                and inputs $\vec{x}_{i,t}$ for client $i \in S_t \backslash (\mathcal{O}_t \cup \mathcal{C})$.
                    
                \item $\Fmal$ sends $S_t$ and $\mathcal{O}_t$ to $\Adv$, and asks $\Adv$ for a set $M_t$: if $\Adv$ replies with $M_t \subseteq S_t\backslash \mathcal{O}_t$ such that $|M_t|/|S_t| \ge 1-\delta$, then $\Fmal$ computes $\vec{z}_t = \sum_{i\in M_t \backslash \mathcal{C}} \vec{x}_{i,t}$ and continues; otherwise $\Fmal$ sends \texttt{abort} to all the honest parties.

                \item Depending on whether the server is corrupted by $\Adv$ or not:
                \begin{enumerate}
                    \item If the server is corrupted by $\Adv$, then $\Fmalrb$ outputs $\vec{z}_t$ to all the parties corrupted by $\Adv$;
                    \item If the server is not corrupted by $\Adv$, then $\Fmalrb$ asks $\Adv$ whether it should continue or not: 
                if $\Adv$ replies with \texttt{continue}, then $\Fmalrb$ outputs $\vec{z}_t$ to the server; otherwise it outputs $\texttt{abort}$ to all the honest parties.
                \end{enumerate}
                
            \end{enumerate}

    \end{myitemize2}
    
    \end{tcolorbox}
\caption{Ideal functionality for Flamingo with robustness (detect-and-abort).}
\label{fig:F-ideal-mal}
\end{figure}

\begin{figure*}

\begin{tcolorbox}[enhanced, boxsep=1mm, left= 0mm, right=0.5mm, title={\textbf{\footnotesize Collection phase with robustness: $\Pi_{\text{sum-robust}}$ for round $t$}}]
\linespread{1.4}
\footnotesize
Initial state from setup phase:
each client $i\in [N]$ holds a value $v$ and public key $PK=g^s$;
each decryptor $u \in \mathcal{D}$ additionally holds a Shamir share of $SK=s$ (threshold $\ell$ with $3\ell+1 = L$).
The server has $g^{s_u}$ for each decryptor $u \in \mathcal{D}$. 

Parameters: $ \delta_D + \eta_D < 1/3$. 

\begin{enumerate}[1.]
    \item \textbf{Report step.}
    
    \textbf{Server performs the following:} 
    
    \quad Compute a set $\mathcal{Q}_{graph} \leftarrow \ChooseSet(v, t, n_t, N)$ and a graph $G_t \leftarrow \GenGraph(v, t, \mathcal{Q}_{graph})$; store $\{A_i(t)\}_{i \in \mathcal{Q}_{graph}}$ computed from $G_t$.

    \quad Notify each client $i\in \mathcal{Q}_{graph}$ that collection round $t$ begins.
    
    \textbf{Each client $i\in \mathcal{Q}_{graph}$ performs the following:} 
    
    \quad Compute $\mathcal{Q}^{local}_{graph} \leftarrow \ChooseSet(v, t, n_t, N)$, and if $i \not\in \mathcal{Q}^{local}_{graph}$, ignore this round.

    \quad Sample $m_{i,t} \xleftarrow{\$} \{0,1\}^{\kappa}$ and compute $\{h_{i, j, t}\}_{j\in A(i)} \leftarrow \PRF(r_{ij}, t)$ for $j\in A(i)$, where $r_{ij}$ is derived from PKI.

    \quad Send to server a message $msg_{i,t}$ consisting of 

    \qquad \qquad $Vec_{i,t} = \vec{x}_{i,t} + \PRG(m_{i,t}) + \sum_{j\in A_t(i)} \pm \PRG(h_{i,j,t})$, \quad
         {\color{magenta}  $\AsymEnc(PK, m_{i,t}), \qquad  
        \AsymEnc(PK, h_{i,j,t})$ for $j \in A_t(i) $ }
        
    \qquad where $A_t(i) \leftarrow \FindNeighbors(v, S_t, i)$, and $\AsymEnc$ is ElGamal encryption (Def.\ref{def:elgamal}), 
           {\color{magenta} along with the signatures for each ciphertext}. 
    
    \item \textbf{Cross check step.}
    
    \textbf{Server performs the following:} 
    
    \quad Denote the set of clients that respond within timeout as $\mathcal{Q}_{vec}$. 
    
    \quad Compute partial sum $\Tilde{z_t} = \sum_{i\in \mathcal{Q}_{vec}} Vec_{i,t}$.
    
    \quad Build decryption request $req$ ($req$ consists of clients in $S_t$ to be labeled):
    
    \quad for each $i\in \mathcal{Q}_{graph}$,
        
    \quad \quad if $i\in \mathcal{Q}_{vec}$, label $i$ with ``online'', 
    
    \quad \quad \quad and {\color{magenta} attach $\AsymEnc(PK, m_{i,t})$};
    
    \quad \quad else label $i$ with ``offline'',
    
    \quad \quad \quad and attach $\{ \AsymEnc(PK, h_{i,j,t})\}_{j\in A_t(i) \cap \mathcal{Q}_{vec}}$.
    
    \quad Send to each $u\in \mathcal{D}$ the request $req$ and a set of the attached ciphertexts $\mathcal{E}_i$ (in order to recover the mask(s) of client $i$).

    \textbf{Each \com $u\in \mathcal{D}$ performs the following:} 
    
    \quad Upon receiving a request $req$, compute $\sigma_u^* \leftarrow Sign(sk_u, req\| t)$,
    and send $(req, \sigma_u^*)$ to all other decryptors via the server.

    \item \textbf{Reconstruction step.}
    
    \textbf{Each \com $u\in \mathcal{D}$ performs the following:} 
    
    \quad Ignore messages with signatures ($\sigma_i$ or $\sigma_u^*$) with round number other than $t$.
    
    \quad Upon receiving a message $(req, \sigma_u^*)$, run $b\leftarrow Vrf(pk_i, req, \sigma_u^*)$. 
    Ignore the message if $b= 0$. 
    
    \quad Abort if $u$ received less than $2\ell+1$ same messages that were not ignored. Denote such message as $req^*$. 
    
   \quad For $req^*$, continue only if 
    
    \quad \quad each client $i \in S_t$ is either labeled as ``online'' or ``offline'';
    
    \quad \quad the number of ``online'' clients is at least $(1-\delta) n_t$;
    
    \quad \quad all the ``online'' clients are connected in the graph;

    \quad \quad each online client $i$ has at least $k$ online neighbors such that $\eta^k < 2^{-\kappa}$.

    \quad For each $i \in \mathcal{Q}_{graph}$, 
    
    \quad \quad For each $(c_0, c_1) \in \mathcal{E}_i$, send $c_0^{s_u}$, {\color{magenta} along with the zero-knowledge proof for $\log _g (g^{s_u}) = \log_{c_0} (c_0)^{s_u}$}.

    \textbf{Server completes the sum:} 
    
    \quad Denote the set of decryptors whose messages have been received as $U$. Compute
    a set of interpolation coefficients $\{\beta_u\}_{u\in U}$ from $U$. 
    
    \quad For each $i \in \mathcal{Q}_{graph}$, {\color{magenta} verify the zero-knowledge proof}
          and reconstruct the mask $m_{i,t}$ or $\{h_{i,j,t}\}_{j\in A_t(i) \cap \mathcal{Q}_{vec}}$:

    \qquad for each ElGamal ciphertext $(c_0, c_1)$ of the mask, compute $c_1\cdot (\prod_{u\in U} (c_0^{s_u})^{\beta_u})^{-1} $;
    
    \quad Output $z_t = \Tilde{z_t} - \PRG(m_{i, t}) + \sum_{j\in A_t(i) \cap \mathcal{Q}_{vec}} \pm \PRG(h_{i,j,t})$.

\end{enumerate}
\end{tcolorbox}
\caption{Detect-and-abort version of collection protocol without DKC (Appendix~\ref{app:robust-protocol}). One can run DKC in parallel to the collection phase.}\label{fig:collection-robust}
\end{figure*}

\begin{proof}

When the server is corrupted by $\Adv$, the proof is essentially the same as the proof of Theorem~\ref{thm:security-main} (Appendix~\ref{app:proof-security-flamingo}).
Below we prove for the case where the server is not corrupted.
We construct a simulator $\Sim$ in the ideal world that runs $\Adv$ as 
  a subroutine.
We denote the ideal functionality for each round $t$ as $\mathcal{F}_{\text{sum-robust}}^{(t)}$, which executes steps 1), 2), and 3) in functionality $\mathcal{F}_{\text{sum-robust}}$.

\begin{enumerate}
    \item $\Sim$ obtains $z_t$ from $\mathcal{F}_{\text{sum-robust}}^{(t)}$.
    \item $\Sim$ received a set $\mathcal{O}_t$ from $\mathcal{R}_{\text{drop}}$ (see Appendix~\ref{app:proof-security-flamingo}). 
    \item (Report step) $\Sim$ computes offline/online labels from the set $\mathcal{O}_t$, and
          interact with $\Adv$ acting as the server.
          Here the collected masked inputs are all zeros and the seeds for the masks are generated by $\Sim$ itself.
    \item (Cross-check step) $\Sim$ acts as honest decryptors in the cross-check step, 
          and if all the honest decryptors would abort in the protocol
          prescription then $\Sim$ outputs whatever $\Adv$ outputs,
          and sends \texttt{abort} to $\Fsumt^{(t)}$ and halts.
    \item (Reconstruction step) $\Sim$ acts as the server in the reconstruction step, if the server would 
          abort in the protocol prescription, then $\Sim$ outputs 
          whatever $\Adv$ outputs, and sends \texttt{abort} to $\mathcal{F}_{\text{sum-robust}}^{(t)}$ and halts. 
    \item (DKC) $\Sim $ acts as the server in the DKC protocol~\cite[Figure 2]{bell22acorn}, if the server would 
          abort in DKC, then  $\Sim$ outputs 
          whatever $\Adv$ outputs, and sends \texttt{abort} to $\mathcal{F}_{\text{sum-robust}}^{(t)}$ and halts. 
\end{enumerate}

We construct a series of hybrid execution, starting from the real world to the ideal world execution. 

\heading{\small Hybrid 1.}
The view of $\Adv$ in the real world execution is the same as the view of $\Adv$ in the ideal world when $\Sim$ would have the actual inputs of honest parties, $\{\vec{x}_i\}_{i\in S_t \backslash \mathcal{C}}$, the masks, and the shares of the secret key $SK=s$ ($\Sim$ in fact would know 
the $s$ in full because of the threshold requirement $3\ell + 1 = L$ and $2\delta_D+\eta_D < 1/3$).

\heading{\small Hybrid 2.}
$\Sim$ now instead of using the actual secret key $s$, it replaces $s$ with 0
and sends the corresponding $|\mathcal{C}\cap \mathcal{D}| < L/3$ shares of 0 in $\mathbb{Z}_q$ to $\Adv$. 
This hybrid has identical distribution to the previous hybrid from the property of Shamir secret sharing. 

\heading{\small Hybrid 3.}
$\Sim$ now instead of using the actual masks of honest parties,
  it replaces each mask with a uniformly random mask generated by $\Sim$ itself. 
This hybrid has identical distribution to the previous hybrid. 
  
\heading{\small Hybrid 4.}
$\Sim$ now instead of using the actual inputs of honest parties
  $\{\vec{x}_i\}_{i\in S_t \backslash \mathcal{C}}$,
  it replaces each input vector with all zeros. 
Note that neither the masked vectors nor the vectors in the clear are
  seen by $\Adv$. 
This hybrid has identical distribution to the previous hybrid. 

\heading{\small Hybrid 5.}
Same as the previous hybrid, except that the label messages from honest decryptors in the last step 
  are replaced with the offline/online labels obtained from $\mathcal{R}_{\text{drop}}$. 
This hybrid is the ideal world execution, and the view of $\Adv$ in this hybrid is identical to the previous hybrid. 

\medskip

Combined with our analysis for the setup phase when the server is honest (Appendix~\ref{app:dkg}),
  we conclude that $\Psi_T$ securely realizes $\Fmalrb$ in the random oracle model. 
\end{proof}

}{%
}
}
{}

\label{lastpage}
\end{document}